\documentclass[ecta,nameyear,draft,nolinenumbers]{econsocart}

\usepackage{amsthm} 
\usepackage{amsmath}
\usepackage{amssymb}
\usepackage{mathtools}
\usepackage{array}
\usepackage{graphicx}
\usepackage{comment}

\usepackage{tikz}
\usepackage{pgfplots}
\pgfplotsset{compat=1.15}

\RequirePackage[colorlinks,citecolor=blue,linkcolor=blue,urlcolor=blue,pagebackref]{hyperref}

\startlocaldefs

\theoremstyle{plain}
\newtheorem{theorem}{Theorem}
\newtheorem{corollary}{Corollary}
\newtheorem{lemma}{Lemma}
\newtheorem{claim}{Claim}

\theoremstyle{remark} 
\newtheorem{definition}{Definition}

\newtheorem{remark}{Remark}

\newcommand{\nc}{\newcommand}
\nc{\eps}{\varepsilon}
\nc{\ul}{\underline}
\nc{\wh}{\widehat}
\nc{\wt}{\widetilde}
\nc{\ol}{\overline}
\nc{\df}{\mathrm{d}}
\nc{\Tau}{\mathcal{T}}
\nc{\R}{\mathbb{R}}
\nc{\E}{\mathbb{E}}
\nc{\Pf}{\mathcal P (\mu,\nu)}
\nc{\Pfhat}{\mathcal P (\hat \mu,\hat \nu)}
\nc{\Pfm}{\mathcal P^* (F,G)}
\nc{\Df}{\mathcal D (V)}
\nc{\DfL}{\mathcal D_L (v)}
\nc{\DfM}{\mathcal D_M (v,m)}
\DeclareMathOperator*{\supp}{supp}

\DeclareMathOperator*{\argmin}{arg\,min}

\DeclareMathOperator{\Lip}{Lip}
\DeclareMathOperator{\Conv}{Conv}

\nc{\1}{\mbox{\bf 1}}

\nc{\cred}[1]{{\color{red} #1}}
\nc{\cblue}[1]{{\color{blue} #1}}
\nc{\ak}[1]{\textcolor{blue}{AK: #1}}
\nc{\aw}[1]{\textcolor{red}{AW: #1}}
\nc{\kr}[1]{\left\lVert{#1}\right\rVert_{KR}}

\tikzset{
  Tstyle/.style={blue, very thick},
  Sstyle/.style={red, very thick, dashed, dash pattern=on 4pt off 3pt},
  T1style/.style={green!50!black, very thick, densely dotted},
  T2style/.style={orange!80!black, very thick, densely dotted},
  guideline/.style={dotted, gray!70},
}
\pgfplotsset{
  curveplotpaper/.style={
    width=4.6cm, height=4.6cm,
    scale only axis,
    axis lines=middle, axis line style={->},
    xlabel={$x$}, ylabel={$z$},
    xlabel style={xshift=4pt, font=\footnotesize},
    ylabel style={yshift=4pt, font=\footnotesize},
    legend columns=2,
    legend style={at={(0,0.95)}, anchor=north west, draw=none,
                  fill=none, font=\scriptsize,
                  /tikz/every even column/.append style={column sep=2pt}},
    legend cell align=left,
    every tick label/.append style={font=\scriptsize},
    tick style={black},
    clip=false,
  },
  curveplotpaperbig/.style={
    width=6.8cm, height=5.6cm,
    scale only axis,
    axis lines=middle, axis line style={->},
    xlabel={$x$}, ylabel={$z$},
    xlabel style={xshift=4pt, font=\footnotesize},
    ylabel style={yshift=4pt, font=\footnotesize},
    legend columns=2,
    legend style={at={(0,0.95)}, anchor=north west, draw=none,
                  fill=none, font=\scriptsize,
                  /tikz/every even column/.append style={column sep=2pt}},
    legend cell align=left,
    every tick label/.append style={font=\scriptsize},
    tick style={black},
    clip=false,
  },
}

\newcommand{\panelpic}[1]{%
  \begin{tikzpicture}[baseline={([yshift=-\dp\strutbox]current axis.outer south)}]
    #1
  \end{tikzpicture}%
}

\newcommand{\originzero}{\node[font=\scriptsize, anchor=north east, inner sep=1pt,
  xshift=-2pt, yshift=-2pt] at (axis cs:0,0) {$0$};}

\newcommand{\tlbl}[1]{\ensuremath{#1\vphantom{\tfrac{1}{2}}}}

\endlocaldefs

\begin{document}

\begin{frontmatter}

\title{Competitive Many-to-One Matching: Sorting vs.\ Equality}
\runtitle{Competitive Many-to-One Matching}

\begin{aug}
%
%
%
\author[id=au1,addressref={add1}]{\fnms{Anton}~\snm{Kolotilin}\ead[label=e1]{akolotilin@gmail.com}}
\author[id=au2,addressref={add2}]{\fnms{Alexander}~\snm{Wolitzky}\ead[label=e2]{wolitzky@mit.edu}}
\address[id=add1]{%
\orgdiv{School of Economics},
\orgname{UNSW Business School}}

\address[id=add2]{%
\orgdiv{Department of Economics},
\orgname{MIT}}
\end{aug}

\support{\emph{Date:} May 28th, 2026 \newline We thank Roberto Corrao for helpful early stage discussions on this research. We are grateful for feedback from seminar participants at AETW, APIOC, Bristol, BC, Brown, CMU, Columbia, ESWC, Harvard, LSE, MIT, Monash, NTU, NYU, Oxford, Penn, Penn State, and UCL. Kolotilin gratefully acknowledges financial support from the ARC DP240103257 and FT240100731.
}
\begin{abstract}
We study many-to-one matching with transfers and peer effects, such as matching workers to firms, students to schools, residents to neighborhoods, or consumers to status goods. With flexible prices (as in the labor market), competitive equilibrium exists and is efficient under general conditions. We characterize when workforces are segregated by skill and matched to firms in a positively assortative manner. In general, equilibrium features alternating intervals of workforce segregation and compression (mixing). Comparative statics characterize when workforces are more segregated or more compressed, and when profits and wages are more or less unequal. With uniform prices (as in school or neighborhood choice), the value generated by peer effects accrues to schools rather than students, and equilibrium can be excessively segregated. Our model generalizes both assignment models (optimal transport) and Bayesian persuasion.
\end{abstract}

\begin{keyword}
\kwd{matching}
\kwd{competitive equilibrium}
\kwd{segregation}
\kwd{compression}
\kwd{inequality}
\kwd{assignment model}
\kwd{optimal transport}
\kwd{Bayesian persuasion}
\end{keyword}
\begin{keyword}[class=JEL] 
\kwd{C78}
\kwd{D51}
\kwd{D82}
\kwd{J00}
\end{keyword}

\end{frontmatter}

\pagebreak

\setcounter{page}{1}

\section{Introduction}\label{sec:intro}

Many economic settings involve many-to-one matching among heterogeneous individuals. Firms of varying productivity hire workers of varying skill. Schools of varying quality admit students of varying ability. Neighborhoods with varying amenities house heterogeneous residents. Status goods are purchased by heterogeneous consumers.

What matching patterns arise in these markets? Do workers segregate by skill, so the best firms hire only the best workers? Or, do firms hire a mix of workers with different skill levels, leading to a more equal distribution of workforces? When are workforces (or student bodies, or neighborhoods) more sorted, and when are they more equal?

These questions can be seen at the heart of current social and economic debates. Influential research has linked rising wage inequality to rising assortativity between firms and workers \citep{card,song,kline}. Proposed explanations include skill-biased technological change \citep{acemoglu2011,haakanson,cortes}, outsourcing \citep{goldschmidt, handwerker}, and factor supply and demand shocks \citep{haanwinckel}. Yet, standard assignment models focus on one-to-one matching \citep{sattinger93,tervio2008}---where firms match with individual workers, not \emph{workforces}---and are thus silent as to what changes in production technologies lead firms to hire more or less homogeneous workforces. Another longstanding debate concerns whether neighborhoods, schools, and other social institutions are inefficiently segregated by income or ability: i.e., whether society is overly sorted or meritocratic. See, e.g., \cite{benabou1996}, \cite{beckermurphy}, \cite{durlauf}, and \cite{bishop} on residential segregation, \cite{arnott1987peer} and \cite{epple1998} on schools, and \cite{markovits} and \cite{sandel} on society more broadly. Again, existing models do not provide a tractable framework for studying equilibrium matching patterns, utilities, or comparative statics in these settings, especially with rich heterogeneity on both sides of the market.

The current paper proposes such a framework. We study competitive equilibrium in continuum-agent, many-to-one matching markets with quasilinear utility and both \emph{match effects}---output depends on the pattern of cross-side matching, such as workers to firms---and \emph{peer effects}---output also depends on the pattern of within-side matching, such as workers to coworkers. Competitive equilibrium exists and is efficient. The key tradeoff in the model is between sorting and equality. Suppose a firm of productivity $x$ that hires a workforce with mean skill $z$ generates value $v(x,z)$. (This is the baseline, \emph{mean-measurable} case of our model, where the skill distribution in a workforce affects output only through its mean $z$.) Assume that there are increasing differences between firm productivity and workforce skill ($v_{xz}>0$), but that there are also decreasing returns to workforce skill ($v_{zz}<0$). Increasing differences implies that, whatever distribution of workforce means $z$ forms in equilibrium, they will be matched with firms in a positively assortative manner. Thus, when the distribution of workforces is more spread out, more value is created by matching better workforces to better firms---but, at the same time, more value is wasted due to decreasing returns to skill. There is thus a tradeoff between sorting workers across workforces---to capitalize on increasing differences---and mixing workers to create more equal workforces---to economize on decreasing returns.

We show how this tradeoff resolves. In the mean-measurable case, equilibrium is described by the solution to a concave maximization problem, subject to majorization constraint, which says that the equilibrium distribution of workforce skills is compressed relative to the population distribution of worker skills. If increasing differences are stronger than decreasing returns (or if there are increasing returns, $v_{zz}\geq 0$), the equilibrium is \emph{positive assortative segregation} (\emph{PAS}): workforces are segregated by skill and matched to firms in a positively assortative manner. PAS plays an important role in our model, as the many-to-one analogue of the positive assortative matching pattern characterized by \cite{becker}. If instead increasing differences are weaker than decreasing returns, the equilibrium is \emph{full compression}: essentially, all workforces are mixed, but more productive firms still hire strictly more skilled workforces. We also provide conditions for other simple matching patterns, such as \emph{upper compression}, where less productive firms hire segregated workforces, while more productive firms hire mixed ones. In general, equilibrium is characterized by alternating intervals of segregation and compression, where at each point either the majorization constraint binds or the marginal value of workforce skill is increasing.

We also obtain comparative statics on the equilibrium assignment and utilities. In particular, we show what changes in the production function $v(x,z)$ make the equilibrium assignment more segregated by skill. Essentially, equilibrium becomes more segregated if and only if increasing differences become stronger relative to decreasing returns. We also show that comparative statics for wages have both similarities and differences from the one-to-one matching case (e.g., \citealp{tervio2008}). One notable difference is that, in our model, a worker's wage may fall when higher-skilled workers become even more productive, whereas her wage is unaffected by this change under one-to-one matching. The reason is that workers do not directly compete with higher-skilled workers under one-to-one matching, while in our model all workers are in competition, as their skills substitute for each other in improving workforces.

Our baseline, mean-measurable model admits a unique equilibrium matching between firm productivities $x$ and mean workforce skills $z$. However, many workforce compositions imply the same distribution of $z$, so the equilibrium matching between firms and individual workers is not uniquely determined. We show how this indeterminacy resolves when firms of different productivities $x$ value workforces according to slightly different moments, all of which are close to the mean $z$. In particular, we show when the equilibrium assignment is \emph{strictly single-dipped}---meaning that more productive firms employ more heterogeneous workforces---or, conversely, \emph{strictly single-peaked}. These results refine the relationship between firm productivity and workforce composition predicted in the mean-measurable case. In addition, we also characterize when the equilibrium is PAS in the general, non-mean-measurable case.

The above results pertain to competitive equilibrium, which is always efficient. However, competitive equilibrium implicitly allows personalized prices, which are not reasonable in all of our applications: firms pay different wages to different workers, but schools might or might not charge different tuition to different students, and the price of a house in a given neighborhood is generally not personalized to different prospective buyers. We therefore contrast competitive equilibrium with \emph{uniform price equilibrium}, where there is a single price to match with each agent on the ``one'' side of the market (e.g., a single price for each school or neighborhood). In markets without personalized prices, uniform price equilibrium is the appropriate descriptive equilibrium concept, and competitive equilibrium serves as an efficiency benchmark. Generalizing earlier results of \cite{benabou1996} and \cite{beckermurphy} (among others), we characterize when uniform price equilibria are \emph{inefficiently segregated}, meaning that the uniform price equilibrium is PAS, but the competitive equilibrium is not. We also obtain the apparently novel result that moving from competitive equilibrium to uniform price equilibrium impacts utilities precisely by shifting the value generated by peer effects from students to schools. This result suggests that agreements among colleges to forego merit scholarships may warrant regulatory scrutiny.\footnote{Such an agreement among elite US colleges was the subject of the ``Overlap case'' \citep{USvBrown} and \emph{Henry v.\ Brown} \citep{HenryvBrown}. A counterpoint is that the money colleges saved by forgoing merit aid supported need-based aid and other pro-social goals \citep{carlton}.}

\smallskip

\textbf{Related Literature.} Our model is a minimal mutual generalization of optimal transport and Bayesian persuasion. It also generalizes prior work on many-to-one matching with peer effects. These connections are explored in more detail in the body of the paper---here, we give a brief overview.

The special case of our model with no peer effects is one-to-one matching with transfers and a continuum of agents, aka the continuum assignment model, aka optimal transport. The discrete version of this model traces to \cite{koopmans}, \cite{shapley}, and \cite{becker}. \citeauthor{sattinger75} (\citeyear{sattinger75}, \citeyear{sattinger93}) derived the wage equation, \cite{tervio2008} derived its comparative statics, and \cite{gabaix2008} applied it to study CEO compensation. The continuum assignment model is due to \cite{GOZ92}. These and other applications of optimal transport are surveyed by \cite{galichon}.

The special case of our model with no match effects is mathematically equivalent to the Bayesian persuasion problem of \cite{kamenica2011}, where a ``sender'' splits a distribution of states into a distribution of posteriors, to maximize the expectation of a function of the posterior. As far as we know, the mean-measurable version of this problem first appeared in \cite{SaintPaul}, in the labor market context. In the persuasion literature, it was studied by \cite{gentzkow2016rothschild}, \cite{KMZL}, \cite{Kolotilin2017}, \cite{dworczak2019}, and \cite{KMS}, among others.

A few prior papers study many-to-one matching with linear peer effects. \cite{kremer} studies the production function $v(x,z)=xe^z$, with $x>0$ for all firms. This function has increasing returns to skill ($v_{zz}\geq 0$), which directly implies that equilibrium is PAS: see Remark~\ref{r:IRS}. \cite{boermaetal} study the production function $v(x,z)=xe^z$, with $x<0$ for all firms. This function has decreasing returns to skill ($v_{zz}<0$) and is thus covered by our main results.\footnote{\cite{kremer} and \cite{boermaetal} assume that production depends on the product of worker abilities, so their production functions are mean-measurable after taking $\log$'s. \citeauthor{boermaetal}'s production function is then $v(\tilde{x},z)=\tilde{x}(1-e^z)$ with $\tilde{x}>0$, which is equivalent to  $v(x,z)=xe^z$, with $x=-\tilde x <0$, after dropping the constant $\tilde{x}$. A difference from our model is that these authors assume that each firm hires a finite number of workers, rather than a continuum.} Other papers that study matching with linear peer effects in particular contexts include \cite{arnott1987peer} and \cite{epple1998} on schools, \cite{debartolome1990}, \cite{benabou1996}, and \cite{beckermurphy} on neighborhoods, and \cite{rayo2013} on status goods.\footnote{Farther afield, \cite{lindenlaub}, \cite{edmond}, and \cite{chone} study sorting and bundling workers with multidimensional skills.}

The literature on many-to-one matching in discrete economies dates to \cite{kelso}. Other than requiring substitutability between workers, their model is very general; consequently, equilibrium lacks a tight characterization. Continuum-agent generalizations, including \cite{che}, \cite{azevedo}, and \cite{leshno}, likewise mostly focus on existence and other basic properties, rather than tightly characterizing matching patterns and comparative statics. The same comment applies to the literature on general equilibrium with ``clubs'' (e.g., \citealp{ellickson}, \citealp{rahman}).

We also relate to a mathematics literature on \emph{weak optimal transport}, initiated by \cite{GRST} (see Remark~\ref{r:WOT}). The closest papers in this literature are \cite{ABC} and \cite{B-VBP}, whose results anticipate parts of our Theorem~\ref{t:e} and Corollary~\ref{c:eqm}, respectively.


\smallskip

The remainder of the paper is organized as follows. Section~\ref{s:model} present our model and competitive equilibrium notion. Section~\ref{s:dual} establishes equilibrium existence and efficiency. Section~\ref{sec:seg} characterizes when equilibrium is PAS. Section~\ref{sec:linear} specializes the model to the mean-measurable case. This section forms the core of our analysis: it characterizes the equilibrium matching and utilities and provides comparative statics. Section~\ref{sec:moment} shows how perturbing the model away from mean-measurability determines individual workers' assignments. Section~\ref{sec:uniform} defines uniform price equilibrium and contrasts it with competitive equilibrium. Section~\ref{s:discussion} discusses possible extensions.

\section{Model}\label{s:model}

We study many-to-one matching between unit measures of heterogeneous firms and workers. Firm types (productivity) $x\in X$ are distributed $\mu \in \Delta (X)$, and worker types (skill) $y \in Y$ are distributed $\nu \in \Delta(Y)$, where $X$ and $Y$ are compact metric spaces.\footnote{For any compact metric space $X$, $\rho_X$ denotes its metric, $\Lip (X)$ denotes the set of Lipschitz functions on $X$, and $\Delta(X)$ denotes the set of Borel probability measures on $X$, endowed with the Kantorovich-Rubinstein distance, induced by the norm \eqref{e:KR}. A subset of $X$ will always mean a Borel subset of $X$. For any $\mu\in \Delta(X)$, its support $\supp(\mu)$ is the smallest compact subset of $X$ of measure one.} Each firm $x$ must hire a unit mass of workers $\eta\in \Delta(Y)$, which we call a \emph{workforce}, and each worker must find employment. Thus, an \emph{assignment} is a distribution $\gamma\in \Delta(X\times \Delta(Y))$ over \emph{(firm, workforce)} pairs $(x,\eta)$ such that its marginal over $X$ equals $\mu$ (each firm hires a unit mass of workers) and the average workforce $\eta$ equals $\nu$ (each worker is employed): that is, $\gamma_X=\mu$ and $\E_\gamma [\eta] =\nu$.

All agents (firms and workers) have quasi-linear utility. A matched firm-workforce pair $(x,\eta)$ generates \emph{value} $V(x,\eta)$. Our main technical assumption is that $V$ is Lipschitz continuous: $V \in \Lip(X \times \Delta(Y))$.

We define a competitive equilibrium as follows.
\begin{definition}\label{d:CE}
A \emph{competitive equilibrium} is an assignment $\gamma$ together with utilities (or \emph{prices}) $(p,q) \in \Lip(X)\times \Lip (Y)$ satisfying
\begin{align*}
p(x) + \E_\eta[q(y)] &= V(x,\eta), \quad \text{for all } (x,\eta)\in \supp(\gamma), \\
p(x) + \E_\eta[q(y)] &\geq V(x,\eta), \quad \text{for all } (x,\eta) \in X\times \Delta(Y).
\end{align*}
\end{definition}
The first condition says that each firm $x$ makes profit $p(x)$ from hiring workforce $\eta$ at wages $q$;  the second condition says that firm $x$ cannot make more profit by hiring a different workforce. Equivalently, $(\gamma;p,q)$ is a core or group-stable outcome (see Remark~\ref{r:duality}).

Note that if $(\gamma;p,q)$ is an equilibrium, then so is $(\gamma;p+c,q-c)$, for any constant $c \in \R$. That is, utilities are only determined up to a constant that divides the surplus between the two sides of the market. This feature is typical of assignment models \citep{sattinger93,boermaetal}.\footnote{Some assignment models pin down the constant by introducing outside options \citep{kremer,tervio2008} or a size imbalance between the sides of the market. These features could likewise be introduced in our model.}

\begin{remark}[Relation to Assignment Models and Bayesian Persuasion]
The model is a minimal generalization of assignment models (one-to-one matching) and Bayesian persuasion. If $V$ is linear in $\eta$, it reduces to one-to-one matching: i.e., the continuum assignment model as in \cite{GOZ92,GOZ99}, or the optimal transport problem as in \cite{santambrogio2015} or \cite{galichon}.\footnote{A minor difference with \cite{GOZ92,GOZ99} is that they assume free disposal, so an assignment $\gamma$ is a positive measure on $X\times \Delta(Y)$ such that $\gamma_X\leq \mu$ and $\int \eta \gamma(\df x,\df \eta) \leq \nu$. If $V\geq 0$ and $\mu(X)=\nu(Y)$, free disposal is immaterial, as $\gamma_X= \mu$ and $\int \eta \gamma(\df x,\df \eta) = \nu$ in any competitive equilibrium.} If instead $X$ is a singleton, it reduces to Bayesian persuasion as in \cite{kamenica2011}.\footnote{Explicitly, our model with singleton $X$ is described by a distribution $\nu$ of workers, which must be split into a distribution $\tau \in \Delta(\Delta(Y))$ of workforces $\eta \in \Delta(Y)$, with $\E_\tau[\eta]=\nu$, to maximize $\E_{\tau}[V(\eta)]$. This is precisely the Bayesian persuasion problem of \cite{kamenica2011}, with workers replacing states of the world, the distribution $\tau$ replacing an experiment, workforces replacing posteriors, and the production function $V(\eta)$ replacing the sender's indirect utility from inducing posterior $\eta$.}
\end{remark}

\begin{remark}[Matching with Peer Effects and Personalized Prices]\label{r:microfoundation}
Our model captures matching with peer effects, because if a worker with skill $y$ working at a firm of productivity $x$ with co-workers with skill distribution $\eta$ generates value $W(x,y,\eta)$, we recover our model with $V(x,\eta)=\E_\eta [W(x,y,\eta)]$. With this interpretation, the special case of our model with no peer effects---where $W$ does not depend on $\eta$---is the continuum assignment model/optimal transport problem. Conversely, the special case with no match effects---where $W$ does not depend on $x$---is the Bayesian persuasion problem.

In the worker-firm interpretation of the model, the value $W(x,y,\eta)$ generated by firm-worker-workforce triple $(x,y,\eta)$ accrues to the firm: the firm's profit is $p(x)=\E_\eta [W(x,y,\eta)-q(y)]$, and worker $y$'s wage is $q(y)$. In other applications, $W(x,y,\eta)$ instead accrues to the ``worker'' (the agent on the ``many'' side of the market). For instance, in the student-school interpretation, $W(x,y,\eta)$ is the willingness-to-pay of a type-$y$ student to attend a type-$x$ school with student body $\eta$; the tuition charged to this student is $W(x,y,\eta)-q(y)$ (leaving her utility $W(x,y,\eta)-(W(x,y,\eta)-q(y))=q(y)$); and the total tuition received by the school is $p(x)=\E_\eta [W(x,y,\eta)-q(y)]$. The neighborhood effects and status goods applications are similar.\footnote{For example, in the neighborhood interpretation, the willingness-to-pay of a type-$y$ household to live in a type-$x$ neighborhood with composition $\eta$ is $W(x,y,\eta)$; the rent charged to this household is $W(x,y,\eta)-q(y)$; and the average rent in the neighborhood is $p(x)=\E_\eta [W(x,y,\eta)-q(y)]$.}

In applications where $W(x,y,\eta)$ accrues to the ``worker,'' competitive equilibrium implicitly involves personalized prices. For example, if two students $y$ and $y'$ attend the same school $x$ with student body $\eta$, their tuitions $W(x,y,\eta)-q(y)$ and $W(x,y',\eta)-q(y')$ can differ. Personalized prices may not be very realistic in the student-school and status goods applications (merit scholarships and giveaways to influencers notwithstanding), and they are likely unrealistic in the neighborhood application (as emphasized by \citealp{beckermurphy}). Thus, in these applications, our competitive equilibrium notion is relevant as an efficiency benchmark rather than a positive prediction, and the corresponding positive theory is instead given by the uniform-price equilibrium notion developed in Section~\ref{sec:uniform}.
\end{remark}

\section{Existence and Welfare Theorems}\label{s:dual}

Our first result is that (i) competitive equilibria exist, (ii) equilibrium assignments maximize total value, and (iii) equilibrium prices minimize total utility subject to a no-entry condition. Formally, define the \emph{planner's problem}
\begin{equation}\tag{P} \label{P}
\begin{gathered}
\max_{\gamma\in \Delta(X\times \Delta(Y))}\E_\gamma [V(x,\eta)]\\
\text{s.t.} \quad \gamma_X=\mu\quad\text{and}\quad \E_\gamma[\eta]=\nu,
\end{gathered}
\end{equation}
and define the \emph{dual problem}
\begin{equation}\tag{D} \label{D}
\begin{gathered}
\min_{(p,q)\in \Lip(X)\times\Lip(Y)}\E_\mu [p(x)] +\E_\nu [q(y)]\\
\text{s.t.} \quad p(x)+\E_\eta [q(y)] \geq V(x,\eta),\quad  \text{for all $(x,\eta)\in X\times \Delta(Y)$.}
\end{gathered}
\end{equation}

The planner's problem is to choose an assignment to maximize total value. This problem can be broken into a Bayesian persuasion problem followed by an optimal transport problem: first, choose a distribution $\tau\in \Delta(\Delta(Y))$ over workforces $\eta$ satisfying $\E_\tau [\eta]=\nu$; then, choose a transport plan $\gamma\in \Delta(X\times \Delta(Y))$ between the (exogenous) distribution $\mu$ of firms and the (endogenous) distribution $\tau$ of workforces satisfying $\gamma_X=\mu$ and $\gamma_{\Delta(Y)}=\tau$. The dual problem is to minimize total utility subject to the constraint that no entering \emph{(firm, workforce)} pair can make a strictly positive profit.\footnote{Mathematically, the dual formulation is justified by the Kantorovich-Rubinstein theorem that a continuous linear function on $M(X)\times M(Y)$ (where $M(X)$ is the space of signed measures on $X$ endowed with the Kantorovich-Rubinstein norm) is represented by $\int p\,\df \mu + \int q\,\df \nu$ for some Lipschitz functions $p$ and $q$ (e.g., Theorem 0 in \citealp{hanin}).}

\begin{theorem}\label{t:e} A competitive equilibrium exists. Moreover, $(\gamma;p,q)$ is a competitive equilibrium if and only if $\gamma$ solves the planner's problem~\eqref{P} and $(p,q)$ solves the dual problem~\eqref{D}. 
\end{theorem}

This strong duality result forms the basis of our analysis, as it lets us study competitive equilibria via the optimization problems~\eqref{P} and \eqref{D}. The bulk of the proof consists in showing that the value of \eqref{P} as a function of the distributions $\mu$ and $\nu$, denoted $U(\mu,\nu)$, inherits Lipschitz continuity from $V(x,\eta)$, at which point the theorem follows from an infinite-dimensional separating hyperplane theorem (namely, the Duality Theorem in \citealp{Gale1967}). In turn, Lipschitz preservation follows the usual line of argument for why the value function in parameterized optimization problems inherits Lipschitz continuity from the objective function \citep{MS}, once we show that the set of feasible assignments varies with the population distributions in a well-behaved manner. Specifically, the key technical step is showing that if $\gamma$ is a feasible assignment for population measures $(\mu,\nu)$, then for any population measures $(\hat \mu,\hat \nu)$, there exists a feasible assignment $\hat \gamma$ such that $\kr{\gamma-\hat \gamma} \leq \kr{\mu-\hat \mu}+\kr{\nu- \hat \nu}$, where $\kr{\cdot}$ is the Kantorovich-Rubinstein norm. This is proved by using optimal transport plans between $\mu$ and $\hat \mu$, and between $\nu$ and $\hat \nu$, to construct $\hat \gamma$.\footnote{In the special case of Bayesian persuasion (where $X$ is a singleton), this argument replicates the proof of Theorem 4 in \cite{DK}.}

\begin{remark}[Price Functions]\label{r:price} The equilibrium utilities satisfy
\begin{align*}
p(x)&=\max_{\eta\in \Delta(Y)} \{V(x,\eta)-\E_\eta[q(y)]\}, \quad \text{for all $x\in \supp(\mu)$}, \quad \text{and} \\
q &\in \argmin_{\hat q\in \Lip(Y)} \{\E_\nu [\hat q(y)]:\, \E_\eta[\hat q(y)]\geq \max_{x\in X}\{V(x,\eta)-p(x)\}, \: \: \text{for all $\eta\in \Delta(Y)$}\}.
\end{align*}
The first equation says that the profit $p(x)$ of firm $x$ equals its maximum profit from hiring any workforce $\eta$ at wages $q$. This is analogous to the profit equation in optimal transport (e.g., \citealp{santambrogio2015}, Proposition 1.11; \citealp{galichon}, Proposition 2.3). The second equation says that the wage schedule $q(y)$ is the lowest one such that the total wage $\E_\eta[\hat q(y)]$ of any workforce $\eta$ is at least the value it generates at any firm $x$, net of the firm's profit $p(x)$. This is analogous to the dual problem in Bayesian persuasion \citep{DK}.
\end{remark}

\begin{remark}[Relation to General Equilibrium] \label{r:duality}
Our economy is convex as we have a continuum of firms and workers, as in \cite{aumann}: mathematically, $U(\mu,\nu)$ is concave as the value of a linear program. As in \cite{aumann} and \cite{GOZ92,GOZ99}, say that $(p,q)\in \Lip(X)\times \Lip (Y)$ is the \emph{core} if $U(\hat \mu,\hat \nu)\leq \E_{\hat \mu} [p(x)] +\E_{\hat \nu}[q(y)]$ for all $(\hat \mu,\hat \nu )\in \Delta(X)\times \Delta(Y)$, with equality at $(\mu,\nu)$, meaning that there is no blocking coalition of firms and workers. Mathematically, this says that $(p,q)$ is a \emph{supergradient} of $U$ at $(\mu,\nu)$. By conjugate duality (e.g., Theorem 16 in \citealp{rockafellar74}), \eqref{D} has a solution and the values of \eqref{P} and \eqref{D} coincide at $(\mu,\nu)$ if and only if $U$ has a supergradient at $(\mu,\nu)$.\footnote{\eqref{P} has a solution by the Weierstrass Theorem.} Core and competitive equilibrium are thus equivalent in our model (as in \citealp{aumann} and \citealp{GOZ92,GOZ99}), and the contribution of Theorem~\ref{t:e} is showing that if $V$ is Lipschitz continuous then $U$ has a supergradient. To do so, we build on \cite{Gale1967}, whose duality theorem implies that $U$ has a supergradient at $(\mu,\nu)$ if and only if it has \emph{bounded steepness} at $(\mu,\nu)$, meaning that there exists a constant $L$ such that ${U(\hat\mu,\hat\nu)-U(\mu,\nu)}\leq L \kr{(\hat \mu,\hat \nu)-(\mu,\nu)}$, for all $(\hat\mu,\hat\nu) \in \Delta(X)\times \Delta(Y)$. So, bounded steepness is equivalent to equilibrium existence,\footnote{This observation applies to any general equilibrium setting with a normed vector space of commodities and quasi-linear utility (e.g., \citealp{ostroy,GOZ}). For comparison, the main technical result in the infinite-dimensional general equilibrium literature is that \emph{properness} of preferences is equivalent to quasi-equilibrium existence in a one-consumer economy \citep{mascollel,MZ}. Since any equilibrium is a quasi-equilibrium, bounded steepness implies properness under quasi-linear utility. Under mild additional conditions (e.g., positivity of $V$ and full support of $(\mu,\nu)$), any quasi-equilibrium is an equilibrium, so properness and bounded steepness are equivalent.} and is implied by Lipschitz continuity. Finally, as described above, we show that $U$ inherits Lipschitz continuity from $V$ by combining an envelope argument as in \cite{MS} and a generalization of the optimal transport argument in \cite{DK}.
\end{remark}

\begin{remark}[Relation to Weak Optimal Transport] \label{r:WOT}
If $V(x,\eta)$ is concave in $\eta$, then, by Jensen's inequality, there exists an optimal assignment $\gamma$ such that $(x,\eta),(x,\hat \eta)\in \supp(\gamma)\implies \eta=\hat \eta$, so \eqref{P} can be reformulated as
\begin{equation}\tag{P'} \label{P'}
\begin{gathered}
\max_{\pi\in \Delta(X\times Y)}\E_{\pi_X} [V(x,\pi_x)]\\
\text{s.t.} \quad \pi_X=\mu\quad\text{and}\quad \pi_Y=\nu,
\end{gathered}
\end{equation} 
where $\pi_x$ is the conditional distribution of $y$ given $x$ under $\pi$. This formulation is the premise of the weak optimal transport literature (e.g., \citealp{GRST,ABC,BBP}). This literature establishes primal attainment and no duality gap (i.e., $\max \eqref{P'}=\inf \eqref{D}$), assuming that $V$ is concave in $\eta$ but allowing $V$ to be upper-semicontinuous and $X$ and $Y$ to be non-compact.\footnote{Without concavity, we can have $\sup  \eqref{P'}<\sup  \eqref{P}=\inf \eqref{D}$, which clarifies the role of the concavity assumption in the weak optimal transport literature.} However, none of these papers provides tractable sufficient conditions on $V$ for dual attainment and hence for equilibrium existence.\footnote{The most related of these papers is \cite{ABC}, who establish dual attainment assuming Lipschitz continuity of $U$ (Theorem 5.1) and in some special cases (Theorems 5.6 and 6.8).}
\end{remark}

\section{Positive Assortative Segregation} \label{sec:seg}
This short section characterizes when workforces are segregated by skill and matched to firms in a positively assortative manner. This is the many-to-one analogue of positive assortative matching. This section also serves as a benchmark for Section~\ref{sec:uniform}, where we compare conditions for (and payoffs under) segregation with competitive and uniform prices.

For the remainder of the paper, we assume that firm and worker types are $1$-dimensional: $X=[\ul x,\ol x]$ and $Y=[\ul y, \ol y]$, with $\ul x <\ol x$ and $\ul y < \ol y$. We represent measures $\mu$ and $\nu$ by their cdfs $F\in \Delta(X)$ and $G\in \Delta(Y)$, which are assumed to have continuous and strictly positive densities $f$ and $g$.\footnote{Most results extend straightforwardly to general $F$ and $G$, which need not have full support or densities.} We also assume that $V(x,\delta_y)$ is twice continuously differentiable on $X\times Y$, where $\delta_y$ is the Dirac measure at $y$, and we denote partial derivatives of $(x,y)\mapsto V(x,\delta_y)$ with subscripts.

In one-to-one matching, strictly increasing differences in $(x,y)$ imply positive assortative matching between firms and workers: each firm $x$ hires worker $y=S(x)$, where $S=G^{-1}\circ F$ is the quantile matching function from $F$ to $G$ \citep{becker}. In our many-to-one model, the same matching arises (where each firm $x$ hires workforce $\delta_{S(x)}$) if all workforces are segregated and matched to firms positively assortatively. We thus refer to $S=G^{-1}\circ F$ as the \emph{positive assortative segregation (PAS) matching}, where the corresponding \emph{PAS assignment} $\phi\in \Delta(X\times \Delta(Y))$ is given by
\[
\phi (A,D)=\int_A \1\{\delta_{S(x)}\in D\} \, \mu(\df x), \quad \text{for all $A\subset X$ and $D\subset \Delta(Y)$}.
\]

The following result characterizes when the equilibrium is PAS.

\begin{theorem}\label{t:fs} The following hold:
\begin{enumerate}
\item PAS is a competitive equilibrium if and only if
\begin{equation}
\begin{gathered}
\int_{\ul y}^{\ol y} \int_{S(x)}^{y} V_y(S^{-1}(\tilde y),\delta_{\tilde y})\, \df \tilde y\, \eta (\df y)\geq V(x,\eta )-V(x,\delta_{S(x)}),\\
\text{for all $(x,\eta)\in X\times \Delta(Y)$}.\label{e:FS}	
\end{gathered}	
\end{equation}
In addition, if \eqref{e:FS} holds strictly for all $(x,\eta)\in X\times \Delta(Y)$ with $\eta\neq \delta_{S(x)}$, then PAS is the unique competitive equilibrium.

\item PAS is a competitive equilibrium for all $(F,G)$ if and only if $V(x,\delta_y)$ is \emph{supermodular in $(x,y)$},
\begin{equation}
V_{xy}(x,\delta_y)\geq 0,\quad \text{for all $(x,y)\in X\times Y$},\label{e:spm}	
\end{equation}
and $V(x,\eta)$ is \emph{outer convex in $\eta$},
\begin{equation}
\int_Y V(x,\delta_y)\,\eta(\df y)\geq V(x,\eta),\quad \text{for all $(x,\eta)\in X\times \Delta(Y)$}.\label{e:oc}
\end{equation}
In addition, if \eqref{e:spm} holds strictly for all $(x,y)$, then PAS is the unique competitive equilibrium for all $(F,G)$.
\item When PAS is a competitive equilibrium, functions $(p,q)$ are corresponding equilibrium prices if and only if there exists $c\in \R $ such that  
\begin{align}
p(x)&=\int_{\ul x}^{x} V_x(\tilde x,\delta_{S(\tilde x)})\, \df  \tilde x+c, \quad \text{for all }x\in X,\label{e:ps}\\
q(y)&=V(\ul x,\delta_{\ul y})+ \int_{\ul y}^y V_y(S^{-1}(\tilde y),\delta_{\tilde y})\, \df \tilde y-c,\quad\text{for all }y\in Y, \label{e:qs}
\end{align}
\end{enumerate}
\end{theorem}

The third part shows that, when PAS is a competitive equilibrium, the corresponding utilities are given by the envelope theorem (i.e., $p'(x)=V_x(x,\delta_{S(x)})$ and $q'(y)=V_y(S^{-1}(y),\delta_{y})$). The first part then follows from substituting these utilities in Definition~\ref{d:CE}. 

The second part says that the equilibrium is PAS for any distribution of firm and worker types if and only if $V(x,\delta_y)$ is supermodular in $(x,y)$ and $V(x,\eta)$ is outer convex in $\eta$.\footnote{This condition is a mutual generalization of the conditions that characterize positive assortative matching in optimal transport (supermodularity of $V$; \citealp{becker}) and full disclosure in Bayesian persuasion (outer convexity of $V$; \citealp{dworczak2019}).} Intuitively, segregation by skill is optimal if and only if $V(x,\eta)$ is outer convex; and, if workforces are segregated, positive assortative matching between firms and workforces is optimal if and only if $V(x,\delta_y)$ is supermodular.

\begin{remark}[Sufficiency of Checking 2-Type Workforces]\label{r:moment}
Suppose there exist functions $v:X\times Y\times \R \to \R $ and $m:X\times Y\to \R $ such that $V(x,\eta)=\E_\eta [v(x,y,\E_\eta[m(x,\tilde y)])]$ for all $(x,\eta)\in X\times \Delta(Y)$. This assumption holds throughout the remainder of the paper: in Section~\ref{sec:moment}, $v$ does not depend on $y$; in the linear peer effects model in Section~\ref{sec:uniform}, $m(x,y)=y$; and in Section~\ref{sec:linear}, both of these additional assumptions hold. Then, to check the conditions of Theorem~\ref{t:fs}, it suffices to consider workforces with at most two worker types: \eqref{e:FS} and \eqref{e:oc} hold if and only if they hold for all $(x,\eta)\in X\times \Delta(Y)$ with $|\supp(\eta)|\leq 2$.\footnote{In addition, if \eqref{e:FS} holds strictly for all $(x,\eta)$ with $\eta \neq \delta_{S(x)}$ and $|\supp(\eta)|\leq 2$, then it holds strictly for all $(x,\eta)$ with $\eta\neq \delta_{S(x)}$.}
\end{remark}

\section{Mean-Measurable Production} \label{sec:linear}

We now focus on the tractable, baseline model specification where the productivity of a workforce is determined by the mean skill of the workers it contains: there exists a function $v\in \Lip(X\times Y)$ such that 
\[
V(x,\eta)=v\left (x,\E_\eta [y]\right),\quad \text{for all $(x,\eta)\in X\times \Delta(Y)$}.
\]
We refer to this specification as the \emph{mean-measurable case}.

The mean-measurable case generalizes \cite{SaintPaul} (as well as mean-measurable Bayesian persuasion following \citealp{gentzkow2016rothschild}, \citealp{Kolotilin2017}, and \citealp{dworczak2019}) by allowing multiple firm types. Other special cases include \cite{kremer}, who assume that $v(x,z)=xe^z$ with $x>0$ for all firms, and \cite{boermaetal}, who assume that $v(x,z)=xe^z$ with $x<0$ for all firms. The mean-measurable case also generalizes models with binary worker types, such as \cite{beckermurphy} and \cite{benabou1996}, as well as models where workers contribute different quantities of efficiency units of labor.

This section characterizes competitive equilibrium in the mean-measurable case. Section~\ref{s:CEL} shows how the primal and dual problems defining a competitive equilibrium simplify. Section~\ref{s:matching} characterizes the equilibrium, showing how the sorting and pooling forces determine the pattern of workforce segregation and compression. Section~\ref{s:segregation} characterizes when simple matching patterns arise, such as PAS or full compression. Section~\ref{s:CS} gives comparative statics on the equilibrium matching and prices.

\subsection{Competitive Equilibrium} \label{s:CEL}

In the mean-measurable case, a workforce $\eta \in \Delta(Y)$ can be identified with its mean, $z=\E_\eta [y]\in Y$. We can thus abstract from the details of how individual workers and firms match (returning to this question only in Section~\ref{sec:moment}), and instead simply ask what distribution $H$ of workforce mean skills $z$ forms, and how workforces with different mean skills $z$ match with firms with different productivities $x$. We therefore define an \emph{assignment} as a cdf $J \in \Delta (X \times Y)$ over \emph{(firm, workforce mean)} pairs $(x,z)$ such that its $X$-marginal is $J_X=F$ and its $Y$-marginal $J_Y=H$ can be obtained by splitting the population worker distribution $G$ into a distribution of workforces $\eta$ that induces distribution $H$ of workforce means. By \cite{blackwell1953equivalent}, such a distribution $H$ is feasible if and only if it \emph{majorizes} $G$, denoted $H\succeq  G$, meaning that $\int_{\ul y}^{y} H(\tilde y)\df \tilde y \leq \int_{\ul y}^{y} G(\tilde y)\df \tilde y$ for all $y \in Y$, with equality at $y= \ol y$.\footnote{Equivalently, $H\succeq  G$ if and only if $G$ is a mean-preserving spread of $H$ (or, if and only if $G$ dominates $H$ in the convex order).}  More generally, the same integral inequality defines the majorization order for any pair of increasing functions on $Y$.

The definition of competitive equilibrium in the mean-measurable case can be simplified by simultaneously viewing $q(z)$ as both the wage of a type $z$ worker and the total wage bill associated with hiring a workforce with mean skill $z$. For such a wage schedule $q$ to be consistent with a competitive assignment $J$, it must be convex---as otherwise a workforce $z$ could be obtained more cheaply by hiring a distribution of workers $\eta$ such that $\E_\eta [q(y)]<q(\E_\eta [z])$---and it must satisfy $\E_{J_Y} [q(z)]=\E_G [q(y)]$---so that the total wages paid by firms, $\E_{J_Y} [q(z)]$, equals the total wages received by workers, $\E_G [q(y)]$. We can thus define a competitive equilibrium as follows (where $\Conv(Y)$ is the set of convex functions on $Y$).

\begin{definition}\label{d:CEL}
A \emph{competitive equilibrium} is an assignment $J$ together with prices $(p,q) \in \Lip(X)\times \Lip (Y)$ satisfying
\begin{align*}
&p(x) + q(z) = v(x,z), \quad \text{for all } (x,z)\in \supp(J), \\
&p(x) + q(z) \geq v(x,z), \quad \text{for all } (x,z) \in X\times Y, \\
&q\in \Conv(Y), \quad \text{and} \quad \E_{J_Y} [q(z)]=\E_G [q(z)].
\end{align*}
\end{definition}

This definition coincides with the general one.

\begin{lemma}\label{l:CEeqv}
Suppose that $\gamma\in \Delta(X\times \Delta(Y))$ induces $J\in \Delta(X\times Y)$.\footnote{Formally, $J(\tilde x,\tilde z)=\E_\gamma[\1\{(x,\E_\eta[y])\in [\ul x,\tilde x]\times [\ul y,\tilde z]\}]$ for all $\tilde x\in  X$ and $\tilde z\in  Y$.} Then, $(J;p,q)$ is a competitive equilibrium in the sense of Definition~\ref{d:CEL} if and only if $(\gamma;p,q)$ is a competitive equilibrium in the sense of Definition~\ref{d:CE}.
\end{lemma}

The primal and dual problems simplify analogously: the planner's problem is
\begin{equation}\tag{P$_L$} \label{PL}
\begin{gathered}
\max_{J\in \Delta(X\times Y)}\E_J [v(x,z)],\\
\text{s.t.} \quad J_X=F \quad\text{and}\quad J_Y \succeq  G,
\end{gathered}
\end{equation}
and the dual problem is
\begin{equation}\tag{D$_L$} \label{DL}
\begin{gathered}
\min_{(p,q)\in \Lip(X)\times\Lip(Y)}\E_F [p(x)] +\E_G [q(z)]\\
\text{s.t.} \quad p(x)+ q(z) \geq v(x,z),\quad  \text{for all }(x,z)\in X\times Y, \\
q \in \Conv(Y).
\end{gathered}
\end{equation}

Theorem~\ref{t:e} specializes to the mean-measurable case as follows.\footnote{Theorem~\ref{t:el} simplifies Theorem~\ref{t:e} analogously to how \cite{dworczak2019} simplifies \cite{kamenica2011} in the mean-measurable case in Bayesian persuasion.}

\begin{theorem}\label{t:el} A competitive equilibrium exists. Moreover, $(J;p,q)$ is a competitive equilibrium if and only if $J$ solves the planner's problem~\eqref{PL} and $(p,q)$ solves the dual problem~\eqref{DL}. 
\end{theorem}

Theorem~\ref{t:el} lets us characterize the competitive equilibrium assignment and utilities via the primal and dual problems~\eqref{PL} and \eqref{DL}.

\subsection{Equilibrium Characterization} \label{s:matching}

For the remainder of this section, we assume that $v$ is twice continuously differentiable, with strictly increasing differences between firm productivity and workforce skill ($v_{xz} (x,z) > 0$ for all $(x,z)$) and strictly decreasing returns to workforce skill ($v_{zz}(x,z)<0$ for all $(x,z)$).\footnote{The increasing returns to skill case, $v_{zz}(x,z)>0$ for all $(x,z)$, is also economically relevant---e.g., \cite{kremer} assumes increasing returns---and we characterize the equilibrium in this case in Remark~\ref{r:IRS}.}

By the same logic as in one-to-one matching, strictly increasing differences implies positive assortative matching between firms and workforces: each firm $x$ hires workforce $z=H^{-1}(F(x))=T(x)$, where $H=J_Y$ is the endogenous distribution of workforce means, and $T$ is the quantile matching function from $F$ to $H$. Thus, \eqref{PL} simplifies to
\begin{align*}
\max_{H\succeq  G}\int_{X} v(x,H^{-1}(F(x)))\, F(\df x),
\end{align*}
or, equivalently, by the equivalence of $H\succeq  G$ and $H^{-1}\preceq  G^{-1} $ (\citealp{SS}, Theorem 3.A.5),
\begin{align} 
&\max_{T:X \to Y}\int_{X} v(x,T(x))\, F(\df x), \tag{M}  \label{M} \\
\text{s.t.} \quad & T \text{ is increasing,} \quad \text{and} \notag \\
&\int_{\ul x}^{x} T(\tilde x)\, F( \df \tilde x) \geq \int_{\ul x}^{x} S(\tilde x)\, F(\df \tilde x), \quad \text{for all } x\in X, \quad \text{with equality at } x=\ol x, \tag{F} \label{F}
\end{align}
where, again, $S=G^{-1}\circ F$ is the segregation matching that arises when $H=G$.\footnote{In addition, $H^{-1}$ denotes the generalized inverse, although we will see that in equilibrium $T$ is continuous and strictly increasing, so it coincides with the ordinary inverse.} 

Finding a competitive equilibrium thus simplifies to finding the increasing matching function $T$ that maximizes $\int v(x,T(x))\, F(\df x)$, such that the resulting distribution of workforce means $F \circ T^{-1}$ majorizes $G$. In particular, if $F$ is uniform---as is without loss up to a change of variables---then a matching $T$ is feasible if and only if $T \preceq S$. As we have assumed that $v_{zz}<0$, the objective is concave, so we have a concave maximization problem with a majorization constraint.

Given an increasing matching function $T$ satisfying the feasibility constraint~\eqref{F}, we can partition the set of firm types $X$ into the \emph{segregation} and \emph{compression} sets
\begin{align*}
X^S &= \left\{ x\in X: \int_{\ul x}^{x} T(\tilde x) F( \df \tilde x) = \int_{\ul x}^{x} S(\tilde x) F(\df \tilde x) \right\}, \quad \text{and} \\
X^C &= \left\{ x\in X: \int_{\ul x}^{x} T(\tilde x) F( \df \tilde x) > \int_{\ul x}^{x} S(\tilde x) F(\df \tilde x) \right\}=X \backslash X^S.
\end{align*}
As we show in Lemma~\ref{l:I} in Appendix~\ref{a:proofs}, if $x \in X^S$, then all workers of types $\tilde y<S(x)$ are employed by firms of types $\tilde x < x$, and all workers of types $\tilde y > S(x)$ are employed by firms of types $\tilde x >x$. Consequently, if $T$ is continuous at $x\in (\ul x,\ol x)$ (which will hold for all $x$ in equilibrium), then $T(x)=S(x)$. If instead $x \in X^C$, then some workers of types $\tilde y<S(x)$ are employed by firms of types $\tilde x > x$, and some workers of types $\tilde y > S(x)$ are employed by firms of types $\tilde x <x$. Note that $X^C$ is open and thus is a union of at most countably many disjoint open intervals, $X^C=\bigcup_i (\ul x_i,\ol x_i)$.\footnote{In the optimal transport literature, the intervals $(\ul x_i,\ol x_i)$ are called \emph{irreducible} with respect to $(F \circ T^{-1},G)$ \citep{beiglbock2016}.}

The following theorem characterizes the equilibrium matching $T$ and prices $(p,q)$.

\begin{theorem}\label{t:eqm} The following hold:
\begin{enumerate}
\item There is a unique equilibrium matching $T$, which is continuous and strictly increasing. In particular, more productive firms hire strictly more skilled workforces.
\item The equilibrium matching is the unique increasing matching $T$ such that the marginal value of skill, $v_z(x,T(x))$, is increasing and is constant on each compression interval $(\ul x_i,\ol x_i)$.
\item Functions $(p,q)$ are corresponding equilibrium prices if and only if there exists $c\in \R$ such that
\begin{align}
p(x)&=\int_{\ul x}^x v_x(\tilde x,T(\tilde x))\, \df \tilde x +c,\quad \text{for all $x\in X$,}\label{e:p}\\
q(y)&=
\begin{cases}\label{e:q}
v(\ul x,T(\ul x))+\int_{T(\ul x)}^y v_z(T^{-1}(\tilde y),\tilde y)\,\df \tilde y-c, &\text{if } y\in [T(\ul x),T(\ol x)],\\
q(T(\ul x))-(T(\ul x)-y)v_z(\ul x,T(\ul x)), &\text{if } y\in [\ul y,T(\ul x)),\\
q(T(\ol x))+(y-T(\ol x)) v_z(\ol x,T(\ol x)), &\text{if } y\in (T(\ol x),\ol y].
\end{cases}
\end{align}
\end{enumerate}
\end{theorem}

The intuition is as follows. For part 1, the equilibrium matching is unique, because it solves \eqref{M}, which has a strictly concave objective and a convex constraint set. It is increasing, because $v_{xz}>0$ implies positive assortativity between firms and workforces. It is strictly increasing, because if $x<x'$ but $T(x)=T(x')$, then $v_z(x,T(x)) < v_z(x',T(x'))$ by $v_{xz}>0$, so the assignment is \emph{under-sorted}: marginally segregating workforces $T(x)$ and $T(x')$ (e.g., swapping a few high-skill firm $x$ employees for a few low-skill firm $x'$ employees) would increase total value. It is continuous, because if $T$ jumps up at $x$ then $v_z(x,T(x))$ jumps down at $x$ by $v_{zz}<0$, but then marginally compressing workforces $T(x-\varepsilon)$ and $T(x+\varepsilon)$ would increase total value.

For part 2, $v_z(x,T(x))$ is increasing, because if $x<x'$ but $v_z(x,T(x)) > v_z(x',T(x'))$ then the assignment is \emph{over-sorted}: marginally compressing workforces $T(x)$ and $T(x')$ would increase total value. It is constant on compression intervals, because if $x$ and $x'$ are in the same compression interval but $v_z(x,T(x)) < v_z(x',T(x'))$, then marginally segregating workforces $T(x)$ and $T(x')$ would increase total value. Conversely, if $v_z(x,T(x))$ is increasing and is constant on compression intervals, then it is optimal, because the matching cannot be improved on either compression intervals (where $v_z(x,T(x))$ is constant, so further compression or segregation are both suboptimal) or segregation intervals (where $v_z(x,T(x))$ is increasing and workforces are already segregated, so compression is suboptimal and further segregation is infeasible).

For part 3, equation~\eqref{e:p} and the middle equation in \eqref{e:q} follow from the envelope theorem, which implies that $p'(x)=v_x(x,T(x))$ and $q'(y)=v_z(T^{-1}(y),y)$.\footnote{The second and third equations in \eqref{e:q} follow from linearly extending the first equation to $[\ul y,T(\ul x))$ and $(T(\ol x),\ol y]$.} Note that the wage schedule $q$ is linear for workers employed by firms in compression intervals (i.e., workers $y$ where $T^{-1}(y) \in X^C$) and is convex for workers employed by firms in segregation intervals (workers $y$ where $T^{-1}(y) \in X^S$). Given the equilibrium matching function, equations~\eqref{e:p}--\eqref{e:q} are the same as in one-to-one matching \citep{sattinger93}.\footnote{A subtlety is that, even though in our model type $y$ workers may be employed by a range of firms in equilibrium---not just firms with productivity $T^{-1}(y)$---their marginal return to skill is still uniquely determined as $v_z(T^{-1}(y),y)$.} However, in one-to-one matching, the matching function is always $S$ (under $v_{xz}>0$), while in our model the matching function---and hence profits and wages---also depends on the type distributions $F$ and $G$. This feature will have important implications for our comparative statics results in Section~\ref{s:CS}.

\begin{remark}[The Convex-Loss Special Case]
An especially tractable special case of the model is the \emph{convex-loss case}, where $v(x,z)=a(z-x)+b(x)+cz$ with a strictly concave function $a$ (e.g., $v(x,z)=xz-z^2/2$).\footnote{This case has received attention in the weak optimal transport literature. In particular, parts 1 and 3 of Corollary~\ref{c:eqm} follow from \cite{B-VBP}, while part 2 is novel. It also covers \cite{boermaetal} (as $|x|e^{z}=e^{z-\tilde x}$, with $\tilde x=-\log |x|$).} In this case, $v_z(x,T(x))$ is increasing in $x$---so the matching function $T$ avoids over-sorting---if and only if $T'(x)\leq 1$: that is, iff $T$ is $1$-Lipschitz. The equilibrium matching is the same for any production function $v$ in the convex-loss case. It can be characterized in three equivalent ways. First, $T$ is the unique increasing $1$-Lipschitz matching with slope $1$ on each compression interval. Second, $T$ results from ironing the function $x-S(x)$ (cf.\ \citealp{myerson}, Section 6; \citealp{KMS}, Section 3.2.1): more precisely, defining $E:[0,1]\to \R $ by $E(t)=\int_0^{t}(F^{-1}(\tilde t)-G^{-1}(\tilde t))\, \df \tilde t$ for all $t\in [0,1]$, letting $\ol E:[0,1]\to \R $ be the convex envelope of $E$ (i.e., the largest convex function that lies below $E$), and letting $\ol e=\ol E'$ be its continuous derivative, we have $T(x)=x-\ol e(F(x))$ for all $x\in X$. Third, $T$  is the most segregated $1$-Lipschitz matching. In sum, letting $\Tau$ be the set of increasing, $1$-Lipschitz maps $T:X\to Y$ satisfying \eqref{F}, we have:
\begin{corollary}\label{c:eqm}
In the convex-loss case, the unique equilibrium matching $T$ has the following three equivalent characterizations:
\begin{enumerate}
\item $T\in \Tau$ such that $T'(x)=1$ for all $x\in X^C$.
\item $T(x)=x-\ol e(F(x))$ for all $x\in X$.
\item $T\in \Tau$ such that $T\circ F^{-1}\succeq  \tilde T\circ F^{-1} $  for all $\tilde T\in \Tau$.
\end{enumerate}
\end{corollary}
\end{remark}

\begin{remark}[Concave Optimization with Majorization Constraints]
Much recent interest in mechanism and information design centers on linear optimization problems over distributions satisfying majorization constraint \citep{bergemann2026information,KMS26}. Solutions to such problems lie at the extreme points of the constraint set, which were characterized by \cite{KMS}. In contrast, our problem~\eqref{M} is a concave optimization problem over distributions satisfying majorization constraints. Theorem~\ref{t:eqm} shows that the solution everywhere either satisfies the majorization constraint with equality ($x \in X^S$) or satisfies a first-order condition ($\frac{\df}{\df x}v_z(x,T(x))=0$). Investigating the generality of this property is a possible avenue for future research.
\end{remark}

\subsection{Conditions for Simple Matching Patterns} \label{s:segregation}

Theorem~\ref{t:eqm} implies transparent sufficient conditions for simple matching patterns. We say that the matching function $T$ is \emph{PAS} if $X^S=X$ (or, equivalently, $T=S$), \emph{full compression} if $X^C=(\ul x,\ol x)$, \emph{upper compression} if there exists $x^*$ such that $X^C=(x^*, \ol x)$, and \emph{lower compression} if there exists $x^*$ such that $X^C=(\ul x, x^*)$. Roughly speaking, under PAS, all workforces are segregated; under full compression, all workforces are mixed; under upper compression, low-skill workforces are segregated and high-skill workforces are mixed; and under lower compression, low-skill workforces are mixed and high-skill workforces are segregated.\footnote{This sentence is imprecise because a firm in a compression interval need not actually hire a mixed workforce in the equilibrium assignment (see, e.g., Figure~\ref{f:examples5}(a)). However, such a firm is always willing to hire a mixed workforce at the equilibrium wages.}

\begin{theorem}\label{c:psm} The equilibrium matching is:
\begin{enumerate}
\item PAS if and only if $v_z(x,S(x))$ is increasing on $X$.
\item Full compression if $v_z(x,S(x))$ is strictly decreasing on $X$.
\item Upper compression if $v_z(x,S(x))$ is strictly quasiconcave on $X$.
\item Lower compression if $v_z(x,S(x))$ is strictly quasiconvex on $X$.
\end{enumerate}
\end{theorem}

To interpret Theorem~\ref{c:psm}, note that
\begin{align}\label{PS}
\frac{\df}{\df x}v_z(x,S(x))=v_{xz}(x,S(x))+v_{zz}(x,S(x))S'(x).
\end{align}
Thus, $v_z(x,S(x))$ is increasing if and only if the sorting force $v_{xz}$ is stronger than the pooling force $v_{zz}$, in that the magnitude of $v_{xz}$ is larger than that of $v_{zz}$, scaled by the change-of-measure term $S'(x)=f(x)/g(S(x))$.

To see the intuition, consider matching firm types $x<x'$ with worker types $S(x)<S(x')$. Under segregation, these matches generate total value $$v(x,S(x))+v(x',S(x')).$$ If instead the workforces are slightly mixed, they generate total value $$v(x,S(x)+\varepsilon)+v(x',S(x')-\varepsilon).$$ Taking $\varepsilon \downarrow 0$, slight mixing is suboptimal if and only if $$v_z(x,S(x))\leq v_z(x',S(x')).$$ Moreover, by $v_{zz}<0$, if slight mixing is suboptimal, then so is any larger mixing. Thus, PAS is optimal (and hence arises in equilibrium) if and only if $v_z(x,S(x))\leq v_z(x',S(x'))$ for all $x<x'$: that is, if and only if $v_z(x,S(x))$ is increasing. Conversely, if $v_z(x,S(x))$ is decreasing, then it is suboptimal to segregate any interval of workforces, so the equilibrium is full compression. Similarly, if $v_z(x,S(x))$ is first increasing and then decreasing (i.e., quasiconcave), then it is optimal to segregate low-skill workforces and mix high-skill workforces, so the equilibrium is upper compression. Likewise, if $v_z(x,S(x))$ is quasiconvex, then the equilibrium is lower compression.\footnote{Part 1 of Theorem~\ref{c:psm}, as well as Remark~\ref{r:IRS}, can alternatively be obtained using Theorem~\ref{t:fs} and Remark~\ref{r:moment}.}

We illustrate Theorem~\ref{c:psm} and Corollary~\ref{c:eqm} with some examples, where the production function is assumed to lie in the convex-loss case (e.g., $v(x,z)=xz-z^2/2$). Recall that, by Corollary~\ref{c:eqm}, the equilibrium matching is the same for any such production function, and is the unique increasing $1$-Lipschitz matching with slope $1$ on each compression interval.

Figure~\ref{f:examples1} illustrates the segregation mapping $S=G^{-1}\circ F$ and the equilibrium matching $T$ for three different type distributions. In panel (a), $F=U[0,1]$ and $G=U[1/4,3/4]$---so, in the population, worker types are concentrated relative to firm types. The segregation mapping $S$ has slope $1/2<1$ and thus coincides with the equilibrium matching $T$. Hence, equilibrium is PAS.

In panel (b), $F=U[1/4,3/4]$ and $G=U[0,1]$. Now, worker types are dispersed relative to firm types. The segregation mapping $S$ has slope $2>1$, so the equilibrium matching $T$ is everywhere compressed, with slope $1$. Equilibrium is fully compressed.

In panel (c), $F=U[0,1]$ and $G(y)=\sqrt{y}$. Here, low-skill workers are plentiful relative to low-productivity firms, but high-skill workers are scarce relative to high-productivity firms. The segregation mapping is $S(x)=x^2$, so $v_z(x,S(x))=a'(x^2-x)+c$. As this function is strictly quasiconcave (as $a$ is strictly concave), Theorem~\ref{c:psm} implies that the equilibrium matching is upper compression.\footnote{Specifically, the equilibrium matching is given by $T(x)=x^2$ if $x\leq 1/4$, and $T(x)=x-3/16$ if $x>1/4$.}

\begin{figure}[ht]
\centering
\begin{minipage}[b]{0.33\textwidth}\centering
\panelpic{%
\begin{axis}[curveplotpaper,
  xmin=0, xmax=1.05, ymin=0, ymax=1.05,
  xtick={1}, xticklabels={\tlbl{1}},
  ytick={0.25, 0.75, 1},
  yticklabels={$\tfrac{1}{4}$, $\tfrac{3}{4}$, $1$},
]
\draw[guideline] (axis cs:0,0.25) -- (axis cs:1,0.25);
\draw[guideline] (axis cs:1,0)    -- (axis cs:1,0.75) -- (axis cs:0,0.75);
\addplot[Tstyle, samples=2, domain=0:1] {x/2 + 0.25}; \addlegendentry{$T$}
\addplot[Sstyle, samples=2, domain=0:1] {x/2 + 0.25}; \addlegendentry{$S$}
\originzero
\end{axis}}\\[-4pt]
{\scriptsize (a)}
\end{minipage}%
\begin{minipage}[b]{0.33\textwidth}\centering
\panelpic{%
\begin{axis}[curveplotpaper,
  xmin=0, xmax=1.05, ymin=0, ymax=1.05,
  xtick={0.25, 0.75, 1},
  xticklabels={\tlbl{\tfrac{1}{4}}, \tlbl{\tfrac{3}{4}}, \tlbl{1}},
  ytick={0.25, 0.75, 1},
  yticklabels={$\tfrac{1}{4}$, $\tfrac{3}{4}$, $1$},
]
\draw[guideline] (axis cs:0.25,0)  -- (axis cs:0.25,0.25) -- (axis cs:0,0.25);
\draw[guideline] (axis cs:0.75,0)  -- (axis cs:0.75,1)    -- (axis cs:0,1);
\draw[guideline] (axis cs:0.75,0.75) -- (axis cs:0,0.75);
\addplot[Tstyle, samples=2, domain=0.25:0.75] {x};
\addlegendentry{$T$}
\addplot[Sstyle, samples=2, domain=0.25:0.75] {2*x - 0.5};
\addlegendentry{$S$}
\originzero
\end{axis}}\\[-4pt]
{\scriptsize (b)}
\end{minipage}%
\begin{minipage}[b]{0.33\textwidth}\centering
\panelpic{%
\begin{axis}[curveplotpaper,
  xmin=0, xmax=1.05, ymin=0, ymax=1.05,
  xtick={0.25, 1}, xticklabels={\tlbl{\tfrac{1}{4}}, \tlbl{1}},
  ytick={0.0625, 0.8125, 1},
  yticklabels={$\tfrac{1}{16}$, $\tfrac{13}{16}$, $1$},
]
\draw[guideline] (axis cs:0.25,0) -- (axis cs:0.25,0.0625) -- (axis cs:0,0.0625);
\draw[guideline] (axis cs:1,0) -- (axis cs:1,1) -- (axis cs:0,1);
\draw[guideline] (axis cs:1,0.8125) -- (axis cs:0,0.8125);
\addplot[Tstyle, samples=40, domain=0:0.25] {x^2};
\addlegendentry{$T$}
\addplot[Tstyle, samples=2, domain=0.25:1, forget plot] {x - 3/16};
\addplot[Sstyle, samples=120, domain=0:1] {x^2};
\addlegendentry{$S$}
\originzero
\end{axis}}\\[-4pt]
{\scriptsize (c)}
\end{minipage}
\caption{Equilibrium Matchings. (a) When $F=U[0,1]$, $G=U[1/4,3/4]$, equilibrium is PAS. 
(b) When $F=U[1/4,3/4]$, $G=U[0,1]$, equilibrium is full compression. 
(c) When $F=U[0,1]$, $G(y)=\sqrt{y}$, equilibrium is upper compression. 
}
\label{f:examples1}
\end{figure}
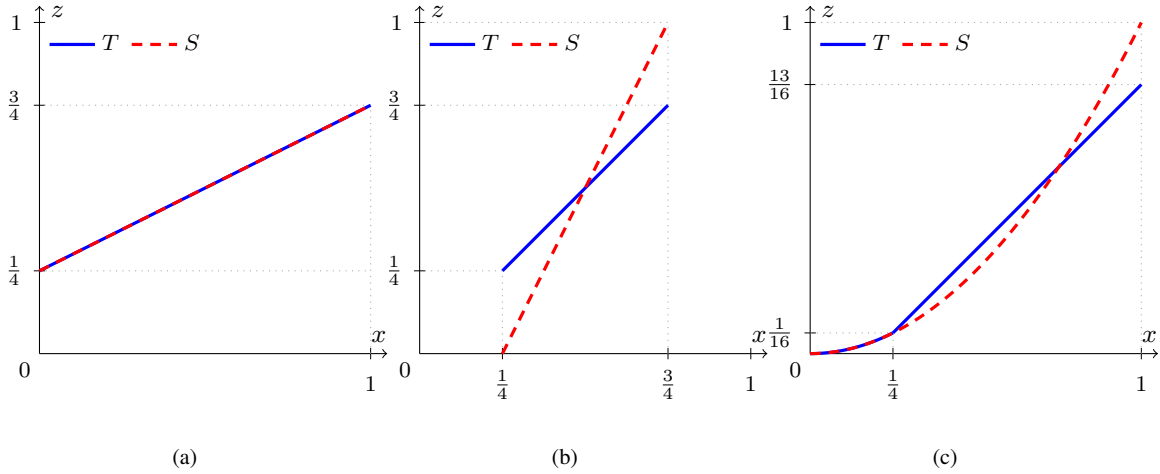

\begin{remark}[Comparison to Bayesian Persuasion]
Theorem~\ref{c:psm} parallels results on the optimality of full disclosure, no disclosure, and upper and lower censorship (i.e., disclosing states below a cutoff while concealing states above it, or vice versa) in the linear Bayesian persuasion problem \citep{KMZL,Kolotilin2017,dworczak2019,KMZ}. In persuasion, if the sender's marginal utility from increasing the receiver's posterior mean belief is $v_z$, then full disclosure is optimal if and only if $v_z$ is increasing, no disclosure is optimal if $v_z$ is decreasing, upper censorship is optimal if $v_z$ is quasiconcave, and lower censorship is optimal if $v_z$ is quasiconvex. However, there does not appear to be a tight formal connection between these two sets of results. For example, no disclosure and upper censorship are deterministic policies, whereas full compression and upper compression involve mixing.
\end{remark}

\begin{remark}[Increasing Returns to Skill]\label{r:IRS}
Under increasing differences and decreasing returns to skill ($v_{zz}<0<v_{xz}$), the sign of \eqref{PS}---and hence the optimality of segregation or compression---depends on distributions $F$ and $G$. In contrast, under increasing differences and \emph{increasing} returns to skill ($v_{xz}>0$ and $v_{zz}\geq0$), PAS is optimal for any distributions $F$ and $G$. Mathematically, this follows from the Fan-Lorentz theorem (e.g., \citealp{KMS}, Theorem 4), which says that $\int_X v(x,T(x))\,\df x\leq \int_X v(x,S(x))\,\df x$ for all $T\preceq S$ if and only if $v(x,z)$ is supermodular in $(x,z)$ and convex in $z$.\footnote{To apply the theorem, first re-scale $x$ to make $F$ uniform. This re-scaling does not affect supermodularity in $(x,z)$ or convexity in $z$.} Intuitively, if $v_{xz}>0$ and $v_{zz}\geq0$ then there is no conflict between forces favoring sorting and pooling, but rather a confluence of forces favoring sorting and separation. This observation generalizes \cite{kremer} (see also \citealp{lazear2001educational}), where $v(x,z)=xe^z$, with $x>0$ for all firms.
\end{remark}

\subsection{Comparative Statics} \label{s:CS}

We now consider how changes in the production function and the distributions of firm and worker types affect equilibrium matching patterns and utilities.

We first give conditions for workforces to become more segregated, in that the distribution $H$ of workforce means spreads out (equivalently, the matching $T\circ F^{-1}$ increases in the majorization order). This occurs if the sorting force $v_{xz}$ becomes stronger relative to the pooling force $v_{zz}$, as captured by the following single-crossing condition.

\begin{definition} \label{d:MRS}
Given two production functions $v$ and $\hat v$, the marginal return to skill under $\hat v$ has a higher \emph{marginal rate of substitution} (\emph{MRS}) between firm productivity and workforce skill than that under $v$, denoted $\hat v_z \unrhd v_z$ if, for all $x'\geq x$ and $z' \geq z$,
\[v_z(x',z')\geq v_z(x,z)\implies  \hat v_z(x',z')\geq \hat v_z(x,z).\]
\end{definition}

In words, $\hat v$ has a higher MRS than $v$ if, whenever the marginal return to worker skill under $v$ is higher at a high-productivity firm $x'$ with a high-skill workforce $z'$ than at a low-productivity firm $x$ with a low-skill workforce $z$, then the same is true under $\hat v$.

\begin{remark}\label{r:rms}[MRS and Magnitude of Sorting vs.\ Pooling Forces]
The condition that $\hat v$ has a higher MRS than $v$ is closely related to the ratio of the magnitudes of the sorting and pooling forces, $v_{xz}/(-v_{zz})$, being higher under $\hat v$. Specifically, if $\hat v_z \unrhd v_z$ then
\[\frac{\hat v_{xz}(x,z)}{-\hat v_{zz}(x,z)}\geq \frac{v_{xz}(x,z)}{-v_{zz}(x,z)} \quad \text{for all } (x,z);\]
and, conversely, if 
\[\frac{\hat v_{xz}(x,z)}{-\hat v_{zz}(x',z')}\geq \frac{v_{xz}(x,z)}{-v_{zz}(x',z')} \quad \text{for all } (x,z)\text{ and }(x',z'),\]
then $\hat v_z \unrhd v_z$. In turn, the latter condition is equivalent to the existence of a constant $c>0$ such that $\hat v_{xz}(x,z)\geq c v_{xz}(x,z)$ and $-\hat v_{zz}(x,z)\leq -c v_{zz}(x,z)$ for all $(x,z)$. These conditions are the ``Spence--Mirrlees'' analogues of the single-crossing condition given in Definition~\ref{d:MRS} \citep{MS94}.
\end{remark}

Our next result shows that matching is more segregated if the sorting force is stronger relative to the pooling force, or if workers are more heterogeneous.\footnote{One might guess that matching is also more segregated when firms are more heterogeneous, but this is only true under very restrictive conditions. To see this, suppose that $\hat T\circ \hat F^{-1}\succeq  T\circ  F^{-1}$ for all $(v,G)$. Taking $v(x,z)=k(x)z-z^2/2$ and a sufficiently spread out, zero-mean distribution $G$, the equilibrium matchings are $T(x)=k(x)-\E_{F}[k(x)]$ and $\hat T(x)=k(x)-\E_{\hat F}[k(x)]$. Then $k\circ \hat F^{-1}-\E_{\hat F}[k(x)]\succeq  k\circ  F^{-1}-\E_{F}[k(x)]$ for all increasing $k$, which implies that $\hat F=F$, so $\hat T\circ \hat F^{-1}\succeq  T\circ  F^{-1}$ for all $(v,G)$ only in the trivial case $\hat F=F$.}

\begin{theorem}\label{t:CS} Let $H_{v,F,G}$ be the equilibrium workforce distribution with production function $v$, firm type distribution $F$, and worker type distribution $G$. Then:
\begin{enumerate}
\item $H_{\hat v,F,G} \preceq  H_{v,F,G}$ for all $(F,G)$ if and only if $\hat v_z \unrhd v_z $. (Workforces are more heterogeneous if and only if the sorting force is stronger relative to pooling force.)
\item $H_{v, F,\hat G} \preceq  H_{v,F,G}$ for all $(v,F)$ if and only if $\hat G \preceq G$. (Workforces are more heterogeneous if and only if workers are more heterogeneous.)
\end{enumerate}
\end{theorem}

To see the logic, suppose by contradiction that $\hat v_z \unrhd v_z $ but there exists an interval of firm types $(x_0,x_1)$ on which $\hat T$ is more compressed than $T$, in that  $\int_{\ul x}^{x} \hat T(\tilde x)\, F(\df \tilde x )> \int_{\ul x}^{x} T(\tilde x)\, F(\df \tilde x)$ for all $x \in (x_0,x_1)$, with equality at $x_0$ and $x_1$. Then, there exist $\tilde x_0<\tilde x_1$ in $(x_0,x_1)$ such that firm $\tilde x_0$'s workforce is better under $\hat T$ ($T(\tilde x_0)<\hat T(\tilde x_0)$), firm $\tilde x_1$'s workforce is better under $T$ ($T(\tilde x_1)>\hat T(\tilde x_1)$), and these firms face the same marginal return to skill under $(\hat v,\hat T)$ ($\hat v_z (\tilde x_0,\hat T(\tilde x_0))=\hat v_z (\tilde x_1,\hat T(\tilde x_1))$). This implies that
$$\hat v_z (\tilde x_0, T(\tilde x_0))>\hat v_z (\tilde x_0, \hat T(\tilde x_0))=\hat v_z (\tilde x_1,\hat T(\tilde x_1))>\hat v_z (\tilde x_1, T(\tilde x_1)),$$
where the inequalities hold because $T(\tilde x_0)<\hat T(\tilde x_0)$, $\hat T(\tilde x_1)< T(\tilde x_1)$, and $v_{zz}<0$. But then $\hat v_z \unrhd v_z $ implies that $v_z (\tilde x_0, T(\tilde x_0))>v_z (\tilde x_1, T(\tilde x_1))$, contradicting that $v_z(x,T(x))$ is increasing.

Recent trends in wage inequality have been linked to rising assortativity between firms and workers \citep{card,song,kline}. Theorem~\ref{t:CS} contributes to this literature by characterizing what changes in production functions and the worker type distribution lead to rising workforce homogeneity in a simple competitive model. In particular, we identify the relative strength of complementarity between firm productivity and workforce skill and of decreasing returns to workforce skill, $v_{xz}/(-v_{zz})$, as the key determinant of the extent of workforce segregation by skill. Specifically, greater firm-worker complementarity and less sharply decreasing returns to workforce skill (or a switch from decreasing to increasing returns) increase workforce segregation by skill.

Part 2 of Theorem~\ref{t:CS} says that, when the population distribution of worker skill becomes more heterogeneous, so does the equilibrium distribution of mean workforce skill. This comparative static is illustrated in panels (a) and (b) of Figure~\ref{f:examples1}, after re-scaling $x$ and $y$ in panel (b) so that $F$ and $v$ are the same in both panels (i.e., $\hat x=2x-1/2$ and $\hat y=2y-1/2$): as the distribution of worker skills spreads out by a factor of 4 (from $U[1/4,3/4]$ to $U[-1/2,3/2]$), the distribution of workforce skills spreads out by a factor of 2 (from $U[1/4,3/4]$ to $U[0,1]$). The same panels also illustrate Part 1 of Theorem~\ref{t:CS}, after re-scaling $x$ and $y$ in panel (b) so that $F$ and $G$ are the same in both panels (i.e., $\hat x=2x-1/2$ and $\hat y=y/2+1/4$): as the sorting force weakens by a factor of 4 relative to the pooling force (e.g., $v_z$ changes from $x-z$ to $x-4z$), the distribution of workforce skills contracts by a factor of 2 (from $U[1/4,3/4]$ to $[3/8, 5/8]$).

\cite{kremermaskin} and \cite{acemoglu1999} study different mechanisms that give a similar comparative static as Part 2 of Theorem~\ref{t:CS}.\footnote{\citeauthor{kremermaskin} consider a model of one-to-one matching between workers, and show that increasing the supply of high-skill workers can eliminate cross-skill matches and reduce low-skill workers' wages. \citeauthor{acemoglu1999} studies a search model with capital investment that complements skill, and shows that increasing the supply of high-skill workers leads firms to invest and then decline to hire low-skill workers.} Moreover, these authors emphasize that changes in equilibrium matching patterns can exacerbate rising skill dispersion. For example, a headline result of \citeauthor{kremermaskin} is that moving mass toward the tails of the worker skill distribution can trigger a shift from mixed workforces to segregated ones.

Figure~\ref{f:examples2} shows that, in our model, the equilibrium matching response can either exacerbate or mitigate rising skill dispersion. The production function is again convex-loss, and firm types are uniformly distributed: $F=U[0,1]$. In panel (a), $g(y)=2-|2-4y|$. Here, extreme worker types are scarce relative to extreme firm types, so extreme firms employ mixed workforces. In panel (b), the skill distribution is more dispersed---$g(y)=1$---and equal to the firm productivity distribution, so the equilibrium is PAS. Thus, in moving from panel (a) to panel (b), the change in equilibrium matching (from partial compression to PAS) further spreads out the distribution of mean workforce skills, similarly to the effect emphasized by \cite{kremermaskin} and \cite{acemoglu1999}. However, in panel (c), the skill distribution is yet more dispersed---$g(y)=|2-4y|$---so now moderate worker types are scare relative to moderate firm types, which leads all firms to employ mixed workforces. Therefore, in moving from panel (b) to panel (c), the change in equilibrium matching (from PAS to full compression) counteracts the rise in skill heterogeneity.

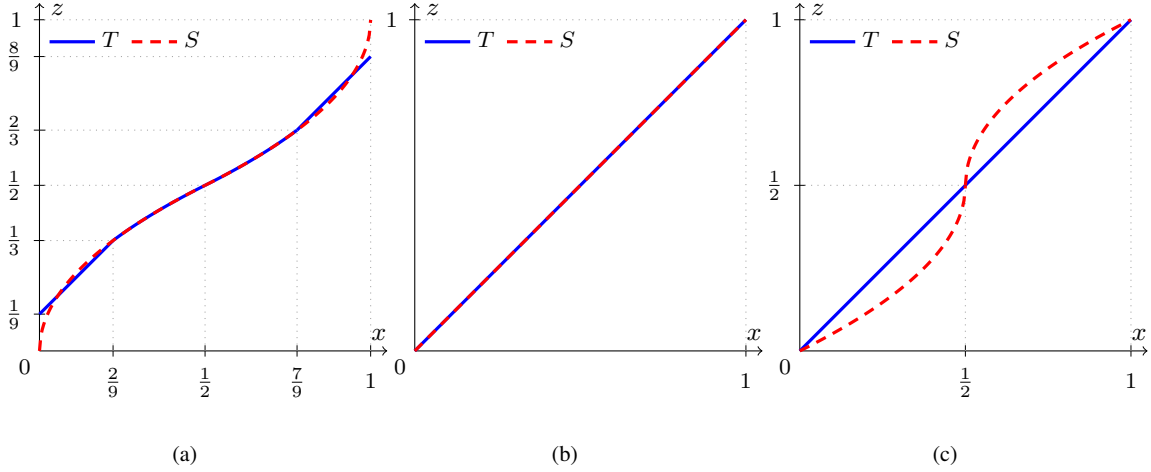
\begin{figure}[ht]
\centering
\begin{minipage}[b]{0.33\textwidth}\centering
\panelpic{%
\begin{axis}[curveplotpaper,
  xmin=0, xmax=1.05, ymin=0, ymax=1.05,
  xtick={0.22222, 0.5, 0.77778, 1},
  xticklabels={\tlbl{\tfrac{2}{9}}, \tlbl{\tfrac{1}{2}}, \tlbl{\tfrac{7}{9}}, \tlbl{1}},
  ytick={0.11111, 0.33333, 0.5, 0.66667, 0.88889, 1},
  yticklabels={$\tfrac{1}{9}$, $\tfrac{1}{3}$, $\tfrac{1}{2}$, $\tfrac{2}{3}$, $\tfrac{8}{9}$, $1$},
]
\draw[guideline] (axis cs:0.22222,0) -- (axis cs:0.22222,0.33333) -- (axis cs:0,0.33333);
\draw[guideline] (axis cs:0.5,0)     -- (axis cs:0.5,0.5)         -- (axis cs:0,0.5);
\draw[guideline] (axis cs:0.77778,0) -- (axis cs:0.77778,0.66667) -- (axis cs:0,0.66667);
\draw[guideline] (axis cs:1,0) -- (axis cs:1,1) -- (axis cs:0,1);
\draw[guideline] (axis cs:1,0.88889) -- (axis cs:0,0.88889);
\addplot[Tstyle, samples=2,  domain=0:0.22222]       {x + 1/9};
\addlegendentry{$T$}
\addplot[Tstyle, samples=80, domain=0.22222:0.5,    forget plot] {sqrt(x/2)};
\addplot[Tstyle, samples=80, domain=0.5:0.77778,    forget plot] {1 - sqrt((1-x)/2)};
\addplot[Tstyle, samples=2,  domain=0.77778:1,      forget plot] {x - 1/9};
\addplot[Sstyle, samples=120, domain=0:0.5] {sqrt(x/2)};
\addlegendentry{$S$}
\addplot[Sstyle, samples=120, domain=0.5:1, forget plot] {1 - sqrt((1-x)/2)};
\originzero
\end{axis}}\\[-4pt]
{\scriptsize (a)}
\end{minipage}%
\begin{minipage}[b]{0.33\textwidth}\centering
\panelpic{%
\begin{axis}[curveplotpaper,
  xmin=0, xmax=1.05, ymin=0, ymax=1.05,
  xtick={1}, xticklabels={\tlbl{1}},
  ytick={1}, yticklabels={$1$},
]
\draw[guideline] (axis cs:1,0) -- (axis cs:1,1) -- (axis cs:0,1);
\addplot[Tstyle, samples=2, domain=0:1] {x}; \addlegendentry{$T$}
\addplot[Sstyle, samples=2, domain=0:1] {x}; \addlegendentry{$S$}
\originzero
\end{axis}}\\[-4pt]
{\scriptsize (b)}
\end{minipage}%
\begin{minipage}[b]{0.33\textwidth}\centering
\panelpic{%
\begin{axis}[curveplotpaper,
  xmin=0, xmax=1.05, ymin=0, ymax=1.05,
  xtick={0.5, 1}, xticklabels={\tlbl{\tfrac{1}{2}}, \tlbl{1}},
  ytick={0.5, 1}, yticklabels={$\tfrac{1}{2}$, $1$},
]
\draw[guideline] (axis cs:0.5,0) -- (axis cs:0.5,0.5) -- (axis cs:0,0.5);
\draw[guideline] (axis cs:1,0) -- (axis cs:1,1) -- (axis cs:0,1);
\addplot[Tstyle, samples=2, domain=0:1] {x};
\addlegendentry{$T$}
\addplot[Sstyle, samples=120, domain=0:0.5]   {0.5*(1 - sqrt(1 - 2*x))};
\addlegendentry{$S$}
\addplot[Sstyle, samples=120, domain=0.5:1, forget plot] {0.5*(1 + sqrt(2*x - 1))};
\originzero
\end{axis}}\\[-4pt]
{\scriptsize (c)}
\end{minipage}
\caption{Impact of Rising Skill Polarization. In all panels, $F=U[0,1]$. In panel (a), $g(y)=2-|2-4y|$, and equilibrium is partially compressed. In panel (b), $g(y)=1$, and equilibrium is PAS. In panel (c), $g(y)=|2-4y|$, and equilibrium is full compression.
}
\label{f:examples2}
\end{figure}

We next consider comparative statics on profits and wages. The most pertinent result in the one-to-one matching context is due to \cite{tervio2008}, who shows that, under increasing differences, firms' marginal return to productivity is higher, and workers' marginal return to skill is lower, when the distribution of firm productivity is lower or the distribution of worker skill is higher in first-order stochastic dominance (FOSD). Thus, inequality in payoffs on either side of the market increases when the type distribution on the opposite side of the market improves.\footnote{This force has been argued to play a role in trends in wage inequality, e.g., among CEOs \citep{gabaix2008}. Moreover, if the utility of the lowest type on each side of the market is determined by an outside option, as in \cite{tervio2008}, comparative statics on payoff differences imply comparative statics on payoff levels. \label{f:tervio}} Under segregation, this result follows immediately from the price equations~\eqref{e:p}--\eqref{e:q}, because (for example) decreasing $F$ or increasing $G$ in the FOSD order increases $S(x)$ for all $x$, which increases $v_x( x,S(x))$ for all $x$, and hence increases $p(x')-p(x)$ for all $x<x'$. In many-to-one matching, the analogous result is not immediate, because it is not obvious that decreasing $F$ or increasing $G$ in the FOSD order increases $T(x)$, once changes in the equilibrium matching pattern are accounted for---that is, that improving the distribution of worker types relative to firm types implies that each firm type hires a more skilled workforce. However, the following result shows that this does hold. An intuition is that increasing $G$ clearly increases $T(x)$ for all $x$ under PAS, and it also increases $T(x)$ for all $x$ under full compression, because $v_z(x,T(x))$ must decrease for all $x$; and the theorem shows that the same conclusion holds in the general case where $T$ alternates between intervals of segregation and compression.

\begin{theorem}\label{t:priceCS} Let $T_{v,F,G}$ be the equilibrium matching and let $(p_{v,F,G},q_{v,F,G})$ be any equilibrium prices with production function $v$, firm type distribution $F$ and worker type distribution $G$. Then, writing $\geq$ for pointwise inequality, we have:
\begin{enumerate}
\item $T_{v,\hat F,G} \geq T_{v,F,G}$ for all $(v,G)$ if and only if $F\leq \hat F$. (Each firm type hires a better workforce if and only if the firm type distribution worsens.)
\item $T_{v,F,\hat G} \geq T_{v,F,G}$ for all $(v,F)$ if and only if $\hat G \leq G$. (Each firm type hires a better workforce if and only if the worker type distribution improves.)
\item For any $v$, $T_{v,\hat F,\hat G} \geq T_{v,F,G}$ if and only if $p'_{v,\hat F,\hat G} \geq p'_{v,F,G}$ if and only if $q'_{v,\hat F,\hat G} \leq q'_{v,F,G}$. (Firm returns to productivity increase, and worker returns to skill decrease, if and only if each firm type hires a better workforce.)
\end{enumerate}
\end{theorem}

Theorem~\ref{t:priceCS} shows that the main price comparative statics in one-to-one assignment models extend to our many-to-one model. Our model also features some new possibilities, due to the dependence of the equilibrium matching pattern on the type distributions. In one-to-one matching with increasing differences, a worker's marginal return to skill equals $q'(y)=v_z(F^{-1}(G(y)),y)$, which depends on the distribution of worker types $G$ only insofar as this changes the quantile of $y$. For example, holding fixed the wage of the lowest-skill workers in the market, an increase in the supply of high-skill workers does not affect the wages of low-skill workers in one-to-one matching.\footnote{Fixing the lowest wage is appropriate if it is determined by a fixed outside option, as in \cite{tervio2008}.\label{f:tervio2}} This result no longer holds in many-to-one matching, because increasing the supply of high-skill workers reduces the productivity $T^{-1}(y)$ of firms matched to low-skill workers, and hence reduces these workers' returns to skill $q'(y)=v_z(T^{-1}(y),y)$. Intuitively, in one-to-one matching, a worker faces competitive pressure only from lower-skill workers, whereas in many-to-one matching, all workers are potentially in competition. For instance, a medium-skill worker can be harmed by an increase in the supply of high-skill workers, because she is now more easily replaced by a mix of low-skill and high-skill workers. Thus, whereas in one-to-one assignment models a worker's return to skill is only reduced by an increase in the skill of less-skilled workers, in our model it is also reduced by an increase in the skill of more-skilled workers.

Figure~\ref{f:examples3big} illustrates this possibility. Panel (a) reproduces panel (b) of Figure~\ref{f:examples1}, where $F$ and $G$ are both uniform, and the equilibrium is full compression. Panel (b) then stretches the right tail of $G$: high-skill workers get even more skilled. This leads to an upward shift of the entire equilibrium matching function, which increases all firms' returns to productivity, and decreases all workers' returns to skill.

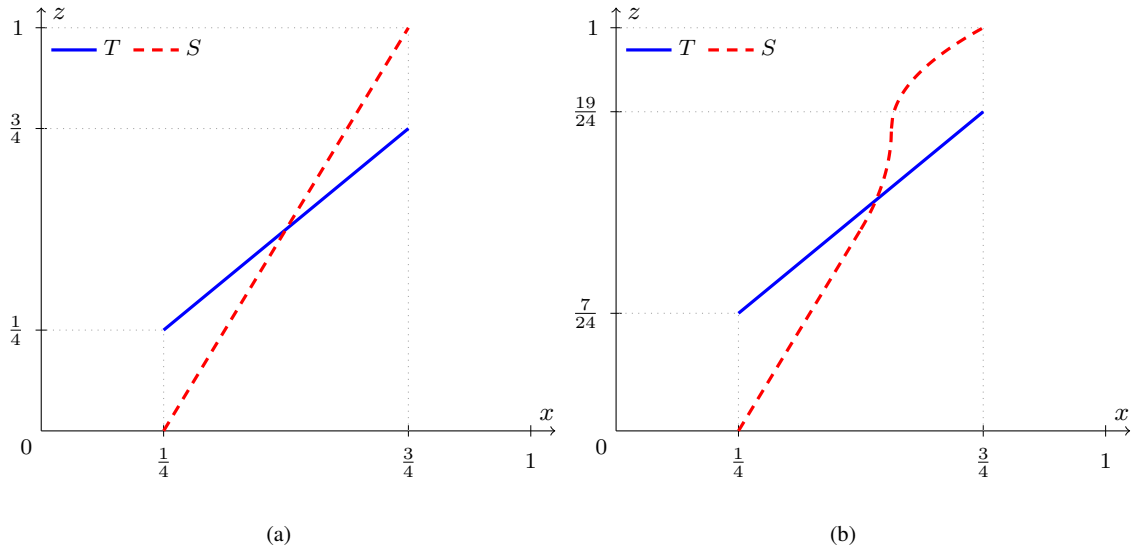
\begin{figure}[ht]
\centering
\begin{minipage}[b]{0.49\textwidth}\centering
\panelpic{%
\begin{axis}[curveplotpaperbig,
  xmin=0, xmax=1.05, ymin=0, ymax=1.05,
  xtick={0.25, 0.75, 1},
  xticklabels={\tlbl{\tfrac{1}{4}}, \tlbl{\tfrac{3}{4}}, \tlbl{1}},
  ytick={0.25, 0.75, 1},
  yticklabels={$\tfrac{1}{4}$, $\tfrac{3}{4}$, $1$},
]
\draw[guideline] (axis cs:0.25,0)  -- (axis cs:0.25,0.25) -- (axis cs:0,0.25);
\draw[guideline] (axis cs:0.75,0)  -- (axis cs:0.75,1)    -- (axis cs:0,1);
\draw[guideline] (axis cs:0.75,0.75) -- (axis cs:0,0.75);
\addplot[Tstyle, samples=2, domain=0.25:0.75] {x};
\addlegendentry{$T$}
\addplot[Sstyle, samples=2, domain=0.25:0.75] {2*x - 0.5};
\addlegendentry{$S$}
\originzero
\end{axis}}\\[-4pt]
{\scriptsize (a)}
\end{minipage}%
\begin{minipage}[b]{0.49\textwidth}\centering
\panelpic{%
\begin{axis}[curveplotpaperbig,
  xmin=0, xmax=1.05, ymin=0, ymax=1.05,
  xtick={0.25, 0.75, 1},
  xticklabels={\tlbl{\tfrac{1}{4}}, \tlbl{\tfrac{3}{4}}, \tlbl{1}},
  ytick={0.291667, 0.791667, 1},
  yticklabels={$\tfrac{7}{24}$, $\tfrac{19}{24}$, $1$},
]
\draw[guideline] (axis cs:0.25,0) -- (axis cs:0.25,0.291667) -- (axis cs:0,0.291667);
\draw[guideline] (axis cs:0.75,0) -- (axis cs:0.75,1) -- (axis cs:0,1);
\draw[guideline] (axis cs:0.75,0.791667) -- (axis cs:0,0.791667);
\addplot[Tstyle, samples=2, domain=0.25:0.75] {x + 1/24};
\addlegendentry{$T$}
\addplot[Sstyle, samples=2,  domain=0.25:0.5] {2*x - 0.5};
\addlegendentry{$S$}
\addplot[Sstyle, samples=80, domain=0.5:0.5625, forget plot]
  {0.75 - sqrt(9 - 16*x)/4};
\addplot[Sstyle, samples=80, domain=0.5625:0.75, forget plot]
  {0.75 + sqrt(48*x - 27)/12};
\originzero
\end{axis}}\\[-4pt]
{\scriptsize (b)}
\end{minipage}
\caption{Impact of a Fatter Right-Tail of Worker Skills. In both panels, $F=U[1/4,3/4]$. In panel (a), $g(y)=1$ for all $y$; in panel (b),
$g(y)=1$, for $y\in [0,1/2]$, $g(y)=3-4y$ for $y\in (1/2,3/4]$, and $g(y)=12y-9$  for $y\in (3/4,1]$. All firms hire more skilled workforces in panel (b).
}
\label{f:examples3big}
\end{figure}

\section{Moment-Measurable Production} \label{sec:moment}

Theorem~\ref{t:eqm} shows that our baseline, mean-measurable model admits a unique equilibrium matching between firm productivities $x$ and mean workforce skills $z$. However, the equilibrium matching between firms and individual workers is not uniquely determined, because many workforce compositions imply the same distribution of $z$. For example, in Figure~\ref{f:examples1}(b) or \ref{f:examples1}(c), the equilibrium distribution of $z$ can be attained when each firm in the compression interval hires a (different) mixture of all workers in the corresponding skill interval, or, alternatively, when each firm hires a mixture of only two worker types, with workers and co-workers matched in a negatively assortative manner. In this section, we show how this indeterminacy resolves for a class of production functions that approximate the mean-measurable case.

Specifically, we generalize the mean-measurable case by assuming that there exist two twice continuously differentiable functions $v:X\times \R \to \R $ and $m:X\times Y\to \R $, satisfying $v_z >0$ and $m_y >0$, such that
\begin{align*}
V(x,\eta)=v\left (x,\E_\eta [m(x,y)]\right),\quad \text{for all } (x,\eta)\in X\times \Delta(Y).
\end{align*}
In this specification, the function $v$ determines output as a function of firm productivity $x$ and an aggregate labor input $z$, and the function $m$ determines the labor input of a worker of skill $y$ at a firm of productivity $x$. We refer to this specification as the \emph{moment-measurable case}. We also continue to assume that $v_{zz}(x,z)<0<v_{xz}(x,z)$ for all $(x,z)$.

The mean-measurable case is the special case of the moment-measurable case where $m(x,y)=m^*(x,y):=y$, for all $x$. Thus, the moment-measurable case retains the assumption that a workforce $\eta$ impacts production at firm $x$ only through a single moment, but it allows this moment to differ across firms.\footnote{As far as we know, the moment-measurable case is mathematically novel. In the Bayesian persuasion case of our model where the firm type distribution is degenerate, the moment-measurable case reduces to the linear case. In contrast, \cite{DK} study Bayesian persuasion with multiple payoff-relevant moments.}

For example, consider a CES aggregator of labor inputs, where the elasticity of substitution between skill levels depends on firm productivity. Specifically, there exists $\alpha:X\to \R $ such that $V(x,\eta)$ depends on $\eta$ only through the CES aggregator $\big (\E_\eta [y^{\alpha(x)}]\big)^{ {1}/{\alpha(x)}}$, so $m(x,y)=(y^{\alpha(x)}-1)/\alpha(x)$. This aggregator is more dependent on the contribution of higher-skill workers at firms with higher $\alpha(x)$. For instance, note that
\begin{align*}
\lim_{\alpha(x)\to -\infty} \big (\E_\eta [y^{\alpha(x)}]\big)^{\frac {1}{\alpha(x)}} &=\inf_{y\in \supp(\eta)} y, \\
\lim_{\alpha(x)\to 0} \big (\E_\eta [y^{\alpha(x)}]\big)^{\frac {1}{\alpha(x)}} &=\exp (\E_{\eta}[\log (y)]), \quad \text{and} \\
\lim_{\alpha(x)\to +\infty} \big (\E_\eta [y^{\alpha(x)}]\big)^{\frac {1}{\alpha(x)}} &=\sup_{y\in \supp(\eta)} y,
\end{align*}
so the aggregator interpolates between \emph{weakest-link} (similar to \citealp{kremer}) and \emph{best-shot} (similar to \citealp{boermaetal}) as $\alpha(x)$ increases from $-\infty$ to $+\infty$. 

For another example (in the spirit of \citealp{garicano2000} and \citealp{garicano2006}), suppose that each worker tries to solve an independent problem whose difficulty depends on the firm's type, and the aggregate labor input is the total measure of problems solved. That is, there exists an increasing function $L:[\ul y-\ol x,\ol y-\ul x]\to [0,1]$ such that a skill-$y$ worker at a type-$x$ firm solves her problem with probability $L(y-x)$, and $V(x,\eta)$ depends on $\eta$ only through $\E_\eta [L(y-x)]$, so $m(x,y)=L(y-x)$.

We now consider the relationship between firm productivity and workforce heterogeneity in the moment-measurable case. We focus on the question of when more productive firms employ more or less heterogeneous workforces. We formalize this using the notion of single-dipped or single-peaked assignments.\footnote{\cite{KCW} introduced a related notion of single-dipped matching in Bayesian persuasion. The notion in \cite{KCW} was previously introduced by \cite{beiglbock2016} in optimal transport, under the name ``left-curtain coupling.''}

\begin{definition}\label{d:sdd}
An assignment $\gamma$ is \emph{strictly single-dipped} (resp., \emph{strictly single-peaked}) if there exists a $\mu$-measure-$1$ set $X^*$ and functions $\rho:X^*\to [0,1]$ and $T_d,T_u:X^*\to Y$, with $T_d(x)\leq T_u(x)$ for all $x\in X^*$, such that
\begin{enumerate}
\item $\eta=(1-\rho(x))\delta_{T_d(x)}+\rho(x)\delta_{T_u(x)}$ for all $(x,\eta)\in \supp(\gamma)$ with $x\in X^*$; and
\item $T_d(x'),T_u(x')\notin (T_d(x),T_u(x))$ for all $x<x'$ in $X^*$.
\end{enumerate}
\end{definition}

That is, an assignment is strictly single-dipped if whenever firm 0 employs both low- and high-skill workers and firm 1 employs medium-skill workers, firm 0 is more productive than firm 1. In other words, the mapping from worker skill $y$ to the productivity $x$ of the firm where this worker is employed is single-dipped over nested sets of workforces. Conversely, an assignment is strictly single-peaked if firms that employ a mix of low- and high-skill workers are less productive than firms that employ medium-skill workers.

Our next result says that the equilibrium assignment is strictly single-dipped if $m_y$ is strictly log-supermodular: $m_{xy}(x,y)/m_y(x,y)$ is strictly increasing in $y$, for all $x$. Intuitively, $m_y(x,y)$ is log-supermodular if labor input $m$ is more convex in skill $y$ at more productive firms.

\begin{theorem} \label{t:sd}
If $m_y$ is strictly log-supermodular (resp., log-submodular), then the equilibrium assignment is strictly single-dipped (resp., single-peaked).
\end{theorem}

For example, with a CES labor aggregator where $m(x,y)=(y^{\alpha(x)}-1)/\alpha(x)$, Theorem~\ref{t:sd} implies that equilibrium is strictly single-dipped (resp., single-peaked) if $\alpha$ is increasing (resp., decreasing). Intuitively, since labor input interpolates between weakest-link and best-shot as $\alpha(x)$ increases, low-$\alpha(x)$ firms hire homogeneous workforces, while high-$\alpha(x)$ firms hire heterogeneous ones. Similarly, with a Garicano-style aggregator where $m(x,y)=L(y-x)$, equilibrium is strictly single-dipped (resp., single-peaked) if $L'$ is strictly log-concave (resp., log-convex).

To see the logic of Theorem~\ref{t:sd}, suppose for contradiction that $m_y$ is strictly log-supermodular but that in equilibrium firm $x_0$ hires workers $y_0$ and $y_0'$ and firm $x_1$ hires worker $y_1$, where $x_0 < x_1$ and $y_0 < y_1 < y_0'$. Consider swapping a mix of $y_0$ and $y_0'$ workers from firm $x_0$ for an equal mass of $y_1$ workers from firm $x_1$, such that total labor input at firm $x_1$, $\E_{\eta_1}[m(x_1,y)]$, remains constant. By strict log-supermodularity of $m_y$, the marginal rate of substitution in the production of labor input between ``improving a worker's skill from $y_0$ to $y_1$'' and ``improving a worker's skill from $y_1$ to $y_0'$'' is higher at firm $x_0$ than firm $x_1$, which implies that this swap strictly increases total labor input at firm $x_0$, $\E_{\eta_0}[m(x_0,y)]$. The swap therefore leaves output at firm $x_1$ unchanged while increasing output at firm $x_0$, contradicting the optimality of the pre-swap assignment.\footnote{This logic is similar to that of Theorem 4 of \cite{KCW}, with the difference that the current result involves matching worker types and (exogenous) firm types, rather than matching states and (endogenous) actions.} Finally, $x$ deterministically maps to $z$ because $v_{zz}<0$, which, together with this swapping argument, implies that the equilibrium assignment is unique and is described as in Definition~\ref{d:sdd}.

We now show that if a mean-measurable production function is perturbed in the direction of log-supermodularity, the equilibrium matching is perturbed in the direction of strict single-dippedness.\footnote{For example, up to re-scaling, the CES aggregator with $m(x,y)=(y^{\alpha(x)}-1)/\alpha(x)$ converges to $m^*$ if $\alpha$ converges to $1$, and the Garicano-style aggregator with $m(x,y)=L(y-x)$ converges to $m^*$ if $L'$ converges to a constant.} This resolves the indeterminacy in the matching $\gamma$ in Theorem~\ref{t:eqm}.

\begin{theorem}\label{t:stability}
Suppose that $m^n$ converges uniformly to $m^*$, and each $m^n_y$ is strictly log-supermodular (resp., log-submodular). Let $T$ be the equilibrium matching function under $m^*$, given by Theorem~\ref{t:eqm}, and let $\gamma$ be the unique strictly single-dipped (resp., single-peaked) assignment that induces $T$. Then the equilibrium assignment $\gamma^n$ under $m^n$ converges weakly to $\gamma$.
\end{theorem}

We illustrate Theorem~\ref{t:stability} in the context of the examples in Figure~\ref{f:examples1}(b) and \ref{f:examples1}(c).

First, consider the convex-loss case with $F=U[1/4,3/4]$ and $G=U[0,1]$, as in Figure~\ref{f:examples1}(b). The equilibrium matching is $T(x)=x$. The corresponding strictly single-dipped assignment, illustrated in Figure~\ref{f:examples4}(a), is given by
\[
\rho(x)=\tfrac34,\qquad T_d(x)=\tfrac{3}{8}-\tfrac{x}{2},\qquad T_u(x)=\tfrac{3x}{2}-\tfrac{1}{8},\quad \text{for all }x\in X. 
\]
The corresponding strictly single-peaked assignment, illustrated in Figure~\ref{f:examples4}(b), is given by
\[
   \rho(x)=\tfrac14,\qquad T_d(x)=\tfrac{3x}{2}-\tfrac{3}{8},\qquad T_u(x)=\tfrac{9}{8}-\tfrac{x}{2},\quad \text{for all }x\in X.
\]
Theorem~\ref{t:stability} implies that, if the production function is moment-measurable and approximately convex-loss but $m_y$ is strictly log-supermodular, then the equilibrium assignment approximates the one in Figure~\ref{f:examples4}(a); and if the production function is moment-measurable and approximately convex-loss but $m_y$ is strictly log-submodular, then the equilibrium assignment approximates the one in Figure~\ref{f:examples4}(b).

\begin{figure}[ht]
\centering
\begin{minipage}[b]{0.49\textwidth}\centering
\panelpic{%
\begin{axis}[curveplotpaperbig,
  xmin=0, xmax=1.05, ymin=0, ymax=1.05,
  xtick={0.25, 0.75, 1},
  xticklabels={\tlbl{\tfrac{1}{4}}, \tlbl{\tfrac{3}{4}}, \tlbl{1}},
  ytick={0.25, 0.75, 1},
  yticklabels={$\tfrac{1}{4}$, $\tfrac{3}{4}$, $1$},
]
\draw[guideline] (axis cs:0.25,0)  -- (axis cs:0.25,0.25) -- (axis cs:0,0.25);
\draw[guideline] (axis cs:0.75,0)  -- (axis cs:0.75,1)    -- (axis cs:0,1);
\draw[guideline] (axis cs:0.75,0.75) -- (axis cs:0,0.75);
\addplot[Tstyle, samples=2, domain=0.25:0.75] {x};
\addlegendentry{$T$}
\addplot[Sstyle, samples=2, domain=0.25:0.75] {2*x - 0.5};
\addlegendentry{$S$}
\addplot[T1style, samples=2, domain=0.25:0.75] {-0.5*x + 0.375};
\addlegendentry{$T_d$}
\addplot[T2style, samples=2, domain=0.25:0.75] {1.5*x - 0.125};
\addlegendentry{$T_u$}
\originzero
\end{axis}}\\[-4pt]
{\scriptsize (a)}
\end{minipage}%
\begin{minipage}[b]{0.49\textwidth}\centering
\panelpic{%
\begin{axis}[curveplotpaperbig,
  xmin=0, xmax=1.05, ymin=0, ymax=1.05,
  xtick={0.25, 0.75, 1},
  xticklabels={\tlbl{\tfrac{1}{4}}, \tlbl{\tfrac{3}{4}}, \tlbl{1}},
  ytick={0.25, 0.75, 1},
  yticklabels={$\tfrac{1}{4}$, $\tfrac{3}{4}$, $1$},
]
\draw[guideline] (axis cs:0.25,0)  -- (axis cs:0.25,1);
\draw[guideline] (axis cs:0.25,0.25) -- (axis cs:0,0.25);
\draw[guideline] (axis cs:0.75,0)  -- (axis cs:0.75,1)    -- (axis cs:0,1);
\draw[guideline] (axis cs:0.75,0.75) -- (axis cs:0,0.75);
\addplot[Tstyle, samples=2, domain=0.25:0.75] {x};
\addlegendentry{$T$}
\addplot[Sstyle, samples=2, domain=0.25:0.75] {2*x - 0.5};
\addlegendentry{$S$}
\addplot[T1style, samples=2, domain=0.25:0.75] {1.5*x - 0.375};
\addlegendentry{$T_d$}
\addplot[T2style, samples=2, domain=0.25:0.75] {-0.5*x + 1.125};
\addlegendentry{$T_u$}
\originzero
\end{axis}}\\[-4pt]
{\scriptsize (b)}
\end{minipage}
\caption{Strictly Single-Dipped and Single-Peaked Full Compression. In both panels, $F=U[1/4,3/4]$, $G=U[0,1]$, and equilibrium is full compression. In panel (a), the assignment is strictly single-dipped; in panel (b), it is strictly single-peaked.}
\label{f:examples4}
\end{figure}
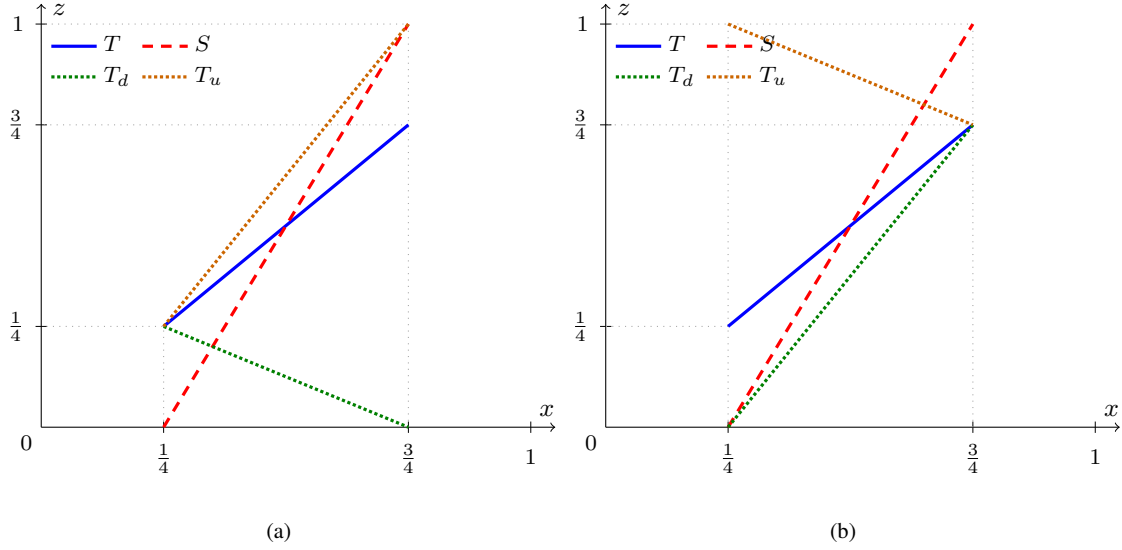

Next, consider the convex-loss case with $F=U[0,1]$ and $G(y)=\sqrt{y}$, as in Figure~\ref{f:examples1}(c). Recall that $S(x)=x^2$ and
\[
T(x)=
\begin{cases}
x^2, &x\in [0,\tfrac 14],\\
x-\tfrac{3}{16}, & x\in (\tfrac 14, 1].
\end{cases}
\]
The corresponding strictly single-dipped assignment, illustrated in Figure~\ref{f:examples5}(a), is given by
\begin{gather*}
T_d(x)=T_u(x)=T(x),\quad \text{for }x\in [0,\tfrac{7}{16}],\quad \text{and}\\
\rho(x)=\tfrac{1}{d(x)},\quad T_d(x)=\bigl( \tfrac 3 2 -d(x)\bigr)^2,\quad T_u(x)=\bigl(2d(x)-\tfrac 32\bigr)^2,\quad \text{for }x\in (\tfrac{7}{16},1],
\end{gather*}
with $d(x)=\sqrt{x+9/{16}}$.\footnote{Formally, this can be proved using Theorem 4.5 in \cite{H-LT}. Specifically, the bifurcation point $m=7/16$ and the functions $\rho$, $T_d$, and $T_u$ on $[m,1]$ are determined by the boundary condition, $T_d(m)=T_u(m)=T(m)$, the martingale condition, $T(x)=(1-\rho(x))T_d(x)+\rho(x)T_u(x)$, and the mass conservation conditions, $\df G(T_u( x))=\rho(x)\, \df F(x)$ and
$\df G(T_d( x))-\df F(T^{-1}(T_d( x)))=-(1-\rho(x))\,\df F( x)$.} 
Notably, firm types $x \in (1/4,7/16)$ lie in the compression interval $(1/4,1)$, yet they must hire homogeneous workforces. This shows that firms in compression intervals need not hire mixed workforces in the equilibrium assignment, even though they are always willing to do so at the equilibrium wages. The corresponding strictly single-peaked assignment, illustrated in Figure~\ref{f:examples5}(b), is given by
\begin{gather*}
T_d(x)=T_u(x)=T(x)=S(x),\quad \text{for }x\in [0,\tfrac{1}{4}],\quad \text{and}\\
\rho(x)=\tfrac{1}{2}-\tfrac{4-x}{6d(x)},\ \ T_d(x)=\tfrac14\bigl(d(x)+x-1\bigr)^{2},\ \ T_u(x)=\tfrac14\bigl(d(x)-x+1\bigr)^{2},\ \ \text{for } x\in (\tfrac{1}{4},1],
\end{gather*}
with $d(x)=\sqrt{(75-4(x-4)^{2})/12}$. 
Theorem~\ref{t:stability} implies that, if the production function is moment-measurable and approximately convex-loss but $m_y$ is strictly log-supermodular, then the equilibrium assignment approximates the one in Figure~\ref{f:examples5}(a); and if the production function is moment-measurable and approximately convex-loss but $m_y$ is strictly log-submodular, then the equilibrium assignment approximates the one in Figure~\ref{f:examples5}(b).

\begin{figure}[ht]
\centering
\begin{minipage}[b]{0.49\textwidth}\centering
\panelpic{%
\begin{axis}[curveplotpaperbig,
  xmin=0, xmax=1.05, ymin=0, ymax=1.05,
  xtick={0.25, 0.4375, 1},
  xticklabels={\tlbl{\tfrac{1}{4}}, \tlbl{\tfrac{7}{16}}, \tlbl{1}},
  ytick={0.0625, 0.25, 0.8125, 1},
  yticklabels={$\tfrac{1}{16}$, $\tfrac{1}{4}$, $\tfrac{13}{16}$, $1$},
]
\draw[guideline] (axis cs:0.25,0) -- (axis cs:0.25,0.0625) -- (axis cs:0,0.0625);
\draw[guideline] (axis cs:0.25,0.0625) -- (axis cs:1,0.0625);
\draw[guideline] (axis cs:0.4375,0) -- (axis cs:0.4375,0.25) -- (axis cs:0,0.25);
\draw[guideline] (axis cs:1,0) -- (axis cs:1,1) -- (axis cs:0,1);
\draw[guideline] (axis cs:1,0.8125) -- (axis cs:0,0.8125);
\addplot[Tstyle, samples=40, domain=0:0.25] {x^2};
\addlegendentry{$T$}
\addplot[Tstyle, samples=2, domain=0.25:1, forget plot] {x - 3/16};
\addplot[Sstyle, samples=120, domain=0:1] {x^2};
\addlegendentry{$S$}
\addplot[T1style, samples=80, domain=0.4375:1]
  {(1.5 - sqrt(x + 9/16))^2};
\addlegendentry{$T_d$}
\addplot[T2style, samples=80, domain=0.4375:1]
  {(2*sqrt(x + 9/16) - 1.5)^2};
\addlegendentry{$T_u$}
\originzero
\end{axis}}\\[-4pt]
{\scriptsize (a)}
\end{minipage}%
\begin{minipage}[b]{0.49\textwidth}\centering
\panelpic{%
\begin{axis}[curveplotpaperbig,
  xmin=0, xmax=1.05, ymin=0, ymax=1.05,
  xtick={0.25, 1}, xticklabels={\tlbl{\tfrac{1}{4}}, \tlbl{1}},
  ytick={0.0625, 0.8125, 1},
  yticklabels={$\tfrac{1}{16}$, $\tfrac{13}{16}$, $1$},
]
\draw[guideline] (axis cs:0.25,0) -- (axis cs:0.25,1);
\draw[guideline] (axis cs:0.25,0.0625) -- (axis cs:0,0.0625);
\draw[guideline] (axis cs:1,0) -- (axis cs:1,1) -- (axis cs:0,1);
\draw[guideline] (axis cs:1,0.8125) -- (axis cs:0,0.8125);
\addplot[Tstyle, samples=40, domain=0:0.25] {x^2};
\addlegendentry{$T$}
\addplot[Tstyle, samples=2, domain=0.25:1, forget plot] {x - 3/16};
\addplot[Sstyle, samples=120, domain=0:1] {x^2};
\addlegendentry{$S$}
\addplot[T1style, samples=80, domain=0.25:1]
  {(sqrt((75 - 4*(x-4)^2)/12) + x - 1)^2 / 4};
\addlegendentry{$T_d$}
\addplot[T2style, samples=80, domain=0.25:1]
  {(sqrt((75 - 4*(x-4)^2)/12) - x + 1)^2 / 4};
\addlegendentry{$T_u$}
\originzero
\end{axis}}\\[-4pt]
{\scriptsize (b)}
\end{minipage}
\caption{Strictly Single-Dipped and Single-Peaked Upper Compression. In both panels, $F=U[0,1]$, $G(y)=\sqrt{y}$, and equilibrium is upper compression. In panel (a), the assignment is strictly single-dipped; in panel (b), it is strictly single-peaked.}
\label{f:examples5}
\end{figure}
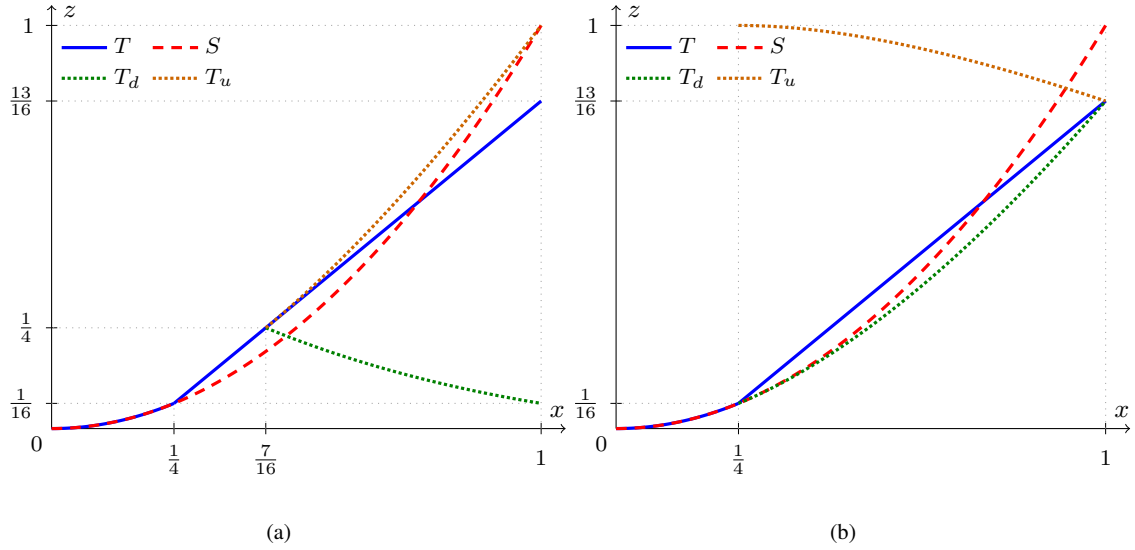

\section{Uniform Price Equilibrium} \label{sec:uniform}

We now return to the social economics applications described in Remark~\ref{r:microfoundation}, where match value accrues to agents on the ``many'' side of the market, such as students attending a school or households moving to a neighborhood. In these settings, competitive equilibrium implicitly involves personalized prices, which may be impractical. In this section, we introduce an alternative equilibrium notion, where there is a single price to match with each agent on the ``one'' side of the market (e.g., a single price for each school or neighborhood). This alternative notion is more realistic in some applications, and we will compare it to competitive equilibrium to assess the impact of uniform pricing on equilibrium assignments and utilities.

The model and the definition of an assignment are unchanged from Section~\ref{s:model} and Remark~\ref{r:microfoundation}, except that in this section we will refer to ``students'' and ``schools'' rather than ``workers'' and ``firms.'' Our alternative equilibrium notion is as follows.

\begin{definition}
A \emph{uniform price equilibrium} (\emph{UPE}) is an assignment $\gamma$ together with prices $p:\supp(\gamma) \to \R$ for schools and utilities $q:Y \to \R$ for students satisfying
\begin{align*}
p(x,\eta) + q(y) &= W(x,y,\eta), \quad \text{for all } (x,\eta)\in \supp(\gamma) \text{ and } y\in \supp(\eta), 
\\
p(x,\eta) + q(y) &\geq W(x,y,\eta), \quad \text{for all } (x,\eta) \in \supp(\gamma) \text{ and } y\in Y. 
\end{align*}
\end{definition}

The first condition says that each student $y$ obtains utility $q(y)=W(x,y,\eta) -p(x,\eta)$ from attending school $x$ with student body $\eta$ and paying tuition $p(x,\eta)$; and the second condition says that student $y$ cannot get a higher utility by attending a different school. Equivalently, $\gamma$ describes a Nash equilibrium assignment (and $q$ describes the corresponding payoffs) in the game where students simultaneously decide what schools to attend, given prices $p$.

UPE requires that two students who attend the same school must pay the same tuition. This could result from an agreement among schools, as in the ``Overlap case'' \citep{USvBrown}. It would also result if schools cannot observe students' types. In the residential choice context, UPE requires that two individuals who want to buy the same house face the same price. This could result from anti-discrimination laws, from the possibility of resale, or simply from the infeasibility of residents paying off to their departing neighbors to sell to more desirable buyers.\footnote{\citeauthor{beckermurphy}---hardly skeptics as to the prevalence of flexible prices---emphasized the latter rationale, writing, ``[W]ith large neighborhoods, such as those in a reasonably sized town or city district, residents generally do not collectively bid for other residents. Instead, they compete for houses, taking account of the composition of the residents in a neighborhood,'' (\citealp{beckermurphy}, p.\ 61).}

In the one-to-one matching case where $W$ is independent of $\eta$, competitive equilibrium (henceforth, CE) and UPE both reduce to the usual equilibrium notion for assignment games, which can equivalently be described as a competitive equilibrium, an optimal assignment, a group stable (core) assignment, or a pairwise stable assignment \citep{GOZ92}. However, they differ in many-to-one matching. For instance, in a CE, each school gets the same payoff from each attending student, inclusive of her peer effect on others, which implies that different students pay different tuitions. Instead, in a UPE, each school gets the same payoff from each student, exclusive of her peer effect, which implies that all students pay the same tuition.\footnote{UPE generalizes the equilibrium notion in \cite{benabou1996} and \cite{beckermurphy}, discussed below. Another related concept is the UPE notion of \cite{mailath2013}. They consider one-to-one matching, but otherwise their notion is more complex, as their model involves pre-match investment and match utilities on both sides of the market.}

CE and UPE are non-nested. In a CE, student $y$ pays tuition $W(x,y,\eta)-q(y)$ to attend school $x$ with student body $\eta$. Since this tuition may differ across students, a CE may not be a UPE. In other words, in a CE, $p(x)=\E_\eta[W(x,y,\eta)-q(y)]$ for all $(x,\eta) \in \supp(\gamma)$, but it is not necessarily true that $p(x)=W(x,y,\eta)-q(y)$ for all $y \in \supp(\eta)$, so the UPE conditions may not hold. Conversely, in a UPE, a school cannot increase its (uniform) tuition and continue to attract students, but an entrant that can charge personalized tuitions to attract a different student body could make a profit. That is, in a UPE, $p(x,\eta) \geq W(x,y,\eta)-q(y)$ for all $(x,\eta) \in \supp(\gamma)$ and $y \in Y$, but it is not necessarily true that $p(x,\eta) \geq \E_{\tilde \eta}[W(x,y,\tilde \eta)-q(y)]$ for all $\tilde \eta \in \Delta(Y)$, so the second condition in the definition of a CE may not hold.\footnote{Another difference is that in CE each type $x$ school gets the same payoff $p(x)$, while in UPE two type $x$ schools with different student bodies $\eta \neq \eta'$ can get different payoffs $p(x,\eta)\neq p(x,\eta')$.}

For the remainder of this section, we assume that peer effects are \emph{moment-measurable}, meaning that there exist two twice continuously differentiable functions $w:X \times Y \times \R  \to \R$ and $m:X \times Y \to \R$ such that 
\[
W(x,y,\eta)=w\left (x,y,\E_\eta[m(x,\tilde y)]\right),\quad \text{for all $(x,y,\eta)\in X\times Y \times \Delta(Y)$}.
\]
For example, the usual case of \emph{linear peer effects} (which generalizes \citealp{benabou1996} and \citealp{beckermurphy}) is the special case where $m(x,y)=y$.

We ask when UPE are positively assortatively segregated, and how this compares to CE (i.e., when PAS is efficient). Recall that PAS is described by the matching function $S(x)=G^{-1}(F(x))$, where student bodies are segregated by ability and matched with schools in a positively assortative manner. We denote partial derivatives of $(x,y,z)\mapsto W(x,y,\delta_z)$ with subscripts, and similarly for $w$.

\begin{theorem}\label{t:UPE} The following hold:
\begin{enumerate}
\item PAS is a UPE if and only if 
\begin{gather}
\int_{S(x)}^{y}\int_{x}^{S^{-1}(\tilde y)} \frac{\df^2 W(\tilde x,\tilde y, \delta_{S(\tilde x)})}{\df \tilde x\,\df \tilde y}\, \df \tilde x\, \df \tilde y\geq 0,\quad \text{for all $(x,y)\in X\times Y$}. 	\label{e:SUPE}
\end{gather}
\item PAS is a UPE for all $(F,G)$ if and only if
\begin{gather}
W_{xy}(x,y,\delta_y)\;\ge\;0,\quad\text{for all $(x,y)\in X\times Y$}, \quad \text{and}	\label{e:wxy}\\
\int_z^y \int_z^{\tilde y}
W_{yz}(x,\tilde y,\delta_{\tilde z})\,\df \tilde z\,\df \tilde y\;\ge\;0,\quad\text{for all $(x,y,z)\in X\times Y\times Y$}.\label{e:iwy}	
\end{gather}
\item PAS is the unique UPE for all $(F,G)$ if
\begin{gather}
W_{xy}(x,y,\delta_y)>0,\quad\text{for all $(x,y)\in X\times Y$}, \quad \text{and} \label{e:wxy>}\\
W_{yz}(x,y,\delta_z)=0,\quad\text{for all $(x,y,z)\in X\times Y\times Y$.}\label{e:iwy=}
\end{gather}
\item When PAS is a UPE, functions $(p,q)$ are corresponding equilibrium prices if and only if there exists $c\in \R$ such that
\end{enumerate}
\begin{align}
p(x)&=\int_{\ul x}^x \bigl( W_x(\tilde x,S(\tilde x),\delta_{S(\tilde x)})+ W_z(\tilde x,S(\tilde x),\delta_{S(\tilde x)})S'(\tilde x) \bigr)\df \tilde x +c,\quad \text{for all $x\in X$,}\label{e:pUPE}\\
q(y)&=W(\ul x,\ul y,\delta_{\ul y})+\int_{\ul y}^y W_y(S^{-1}(\tilde y),\tilde y,\delta_{\tilde y})  \, \df \tilde y -c,\quad \text{for all $y\in Y$.}\label{e:qUPE} 	
\end{align}
\end{theorem}

Parts 1 and 2 imply that PAS is a UPE if $W(x,y,\delta_{S(x)})$ has increasing differences in $(x,y)$, and PAS is a UPE for all $(F,G)$ if $W_{xy}(x,y,\delta_y)\geq 0$ and $W_{yz}(x,y,\delta_z)\geq 0$. Intuitively, PAS is a UPE if better students are willing to pay more to attend better schools with better peers. Conversely, PAS is not a UPE if $W(x,y,\delta_{S(x)})$ ever has strictly decreasing differences locally at $y=S(x)$, and PAS is not a UPE for some $(F,G)$ if $W_{xy}(x,y,\delta_y)< 0$ or $W_{yz}(x,y,\delta_y)< 0$ for some $(x,y)$. Intuitively, if $W_{xy}(x,y,\delta_y)< 0$ then PAS is not a UPE if $f$ is much more dispersed than $g$ around $(x,y)$ (so match effects dominate), and if $W_{yz}(x,y,\delta_y)< 0$ then PAS is not a UPE if $g$ is much more dispersed than $f$ around $(x,y)$ (so peer effects dominate).

Part 3 says that PAS is the unique UPE if $W_{xy}(x,y,\delta_y)> 0$ and $W_{yz}(x,y,\delta_z)= 0$. These conditions ensure that \emph{negative assortative segregation}---where student bodies are segregated but are matched with schools in a \emph{negatively} assortative manner---is not a UPE.

To understand part 4, consider for simplicity the case where $f=g$, so $S(x)=x$. Here, the theorem says that the UPE prices $(p_{\mathrm{UPE}},q_{\mathrm{UPE}})$ satisfy
\begin{align*}
p_{\mathrm{UPE}}'(x)=W_x(x,x,\delta_x)+W_z(x,x,\delta_x) \quad \text{and} \quad q_{\mathrm{UPE}}'(x)=W_y(x,x,\delta_x), \quad \text{for all $x\in X$.}
\end{align*}
In contrast, when PAS is a competitive equilibrium, we have
\begin{align*}
p_{\mathrm{CE}}'(x)=W_x(x,x,\delta_x) \quad \text{and} \quad q_{\mathrm{CE}}'(x)=W_y(x,x,\delta_x)+W_z(x,x,\delta_x), \qquad \text{for all $x\in X$.}
\end{align*}
Thus, in either case, schools receive their marginal contribution to social surplus, $W_x$, and students receive their marginal contribution exclusive of peer effects, $W_y$, but the two cases differ in which side of the market receives the marginal contribution of peer effects, $W_z$. In a CE, tuition reflects students' contribution to schools inclusive of peer effects, so this value accrues to students. In a UPE, tuition reflects schools' value to students inclusive of peer effects, so this value accrues to schools. 

This comparison has obvious implications for the distributional consequences of regulating private school tuition. Assume that (i) the utility of the lowest type of students and schools is fixed by outside options, (ii) the CE and UPE are both PAS, and (iii) peer effects are positive: $W_z>0$. Then an agreement among schools that each school must charge the same tuition to all its attendees (i.e., no merit scholarships) increases the profit of all schools and reduces the welfare of all students (except the lowest types on each side). Similarly, an agreement among luxury goods retailers not to give discounts to high-status consumers (e.g., ``influencers'') benefits all retailers and harms all consumers.

Parts 1--3 of Theorem~\ref{t:UPE} simplify in the mean-measurable case where $m(x,y)=y$ and $w$ is affine in $y$. In this case, $V(x,\eta)=\E_\eta[w(x,y,\E_\eta[\tilde y])]=w(x,\E_\eta [y],\E_\eta[y])$, so the production function is mean-measurable, as in Section~\ref{sec:linear}. 

\begin{corollary}\label{c:affine}
In the mean-measurable case, the following hold:\footnote{The proof of the ``only if'' direction of part 3 uses densities that are strictly positive on subintervals of $X$ and $Y$, while elsewhere in the paper we assume strictly positive densities on $X$ and $Y$.}
\begin{enumerate}
\item PAS is a UPE if and only if $w_y(x,S(x),S(x))$ is increasing in $x$.
\item PAS is a UPE for all $(F,G)$ if and only if $w_y(x,z,z)$ is increasing in both $x$ and $z$.
\item PAS is the unique UPE for all $(F,G)$ if and only if $w_y(x,z,z)$ is strictly increasing in $x$ and is independent of $z$.
\end{enumerate}
\end{corollary}

We can now compare the conditions in Theorem~\ref{t:UPE} and Corollary~\ref{c:affine} for PAS to be a UPE with conditions for PAS to be efficient. The following result is an immediate implication of Theorems~\ref{t:fs} and \ref{c:psm} and Remarks~\ref{r:moment} and \ref{r:IRS}.

\begin{corollary} \label{c:UPE}
PAS is efficient for all $(F,G)$ if and only if
\begin{align*}
&W_{xy}(x,y,\delta_y)+W_{xz}(x,y,\delta_y)\geq 0, \quad \text{and} \\
&(1-\rho)W(x,y,\delta_y)+\rho W(x,y',\delta_{y'}) \\
\geq  &(1-\rho)W(x,y,(1-\rho)\delta_y+\rho \delta_{y'})+\rho W(x,y',(1-\rho)\delta_y+\rho \delta_{y'}),
\end{align*}
for all $(x,y,y')\in X\times Y\times Y$ and all $\rho \in (0,1)$.

Moreover, in the mean-measurable case,
\begin{enumerate}
\item PAS is efficient if and only if $w_y(x,S(x),S(x))+w_z(x,S(x),S(x))$ is increasing in $x$.
\item PAS is efficient for all $(F,G)$ if and only if $w_y(x,z,z)+w_z(x,z,z)$ is increasing in $x$ and $z$.
\end{enumerate}
\end{corollary}

Comparing Corollaries~\ref{c:affine} and \ref{c:UPE} shows that PAS can be a UPE (and even the unique one) even if segregation is inefficient. In particular, by Corollaries~\ref{c:affine} and \ref{c:UPE}, PAS is a UPE but is not efficient if and only if $w_y(x,S(x),S(x))$ is increasing in $x$ but $w_y(x,S(x),S(x))+w_z(x,S(x),S(x))$ is not increasing in $x$, which tends to occur if $w_z(x,S(x),S(x))$ is decreasing in $x$. In turn, $w_z(x,S(x),S(x))$ can be decreasing in $x$ if $w_{xz}<0$ (``good schools'' and ``good peers'' are substitutes), if $w_{yz}<0$ (good students value good peers less), or if $w_{zz}<0$ (there are decreasing returns to peer quality). Intuitively, in these cases, a high-ability student's private return to attending a high-quality school tends to exceed the social return, resulting in inefficient segregation.

\cite{benabou1996} and \cite{beckermurphy} consider related models of neighborhood choice. \citeauthor{benabou1996} assumes that neighborhoods are identical ex ante (corresponding to degenerate $x$ in our model) and that there are two household types (corresponding to binary $y$). He shows that in this case the efficiency of PAS depends on the sign of ${\df} \left( w_y(x,z,z)+w_z(x,z,z) \right)/\df z$. Part 2 of the mean-measurable case of Corollary~\ref{c:UPE} generalizes this result to heterogeneous neighborhoods and continuous household types. \citeauthor{beckermurphy} focus on the case where there are two neighborhood types and two household types with equal mass, and $w(x,y,z)=a(z)+yb(x)$. Besides generality, the cases studied by these authors omit some potentially important effects. For instance, \citeauthor{benabou1996} and \citeauthor{beckermurphy} both assume that $w_{xz}=0$, while we saw that $w_{xz}<0$ creates a force for inefficient segregation. This case seems realistic, for example if learning from teachers and learning from peers are substitutes.\footnote{\citeauthor{benabou1996} and \citeauthor{beckermurphy} also do not compare utilities between competitive and uniform price equilibria.}

\section{Discussion} \label{s:discussion}

We conclude by mentioning three possible extensions of our model.

\textbf{Investment.} Several authors have considered productivity-enhancing investment prior to matching (e.g., \citealp{noldeke}). Our model gives a simple and stark result in this context: if workers are ex ante identical and the cost for a worker to acquire skill $y$ is $C(y)$, where
\begin{align*}
V(x,\eta)-\E_\eta [C(y)] < V(x,\delta_{\E_\eta[y]})-C({\E_\eta[y]}), \quad \text{for all } (x,\eta) \text{ s.t.\ } |\supp(\eta)|\geq 2,
\end{align*}
then the competitive equilibrium must be segregated, because it is inefficient to form any nondegenerate workforce $\eta$, rather than having the same workers each acquire skill $\E_\eta[y]$. For example, this condition holds in the mean-measurable case if $C$ is strictly convex. From this perspective, equilibrium workforce compression results from either ex ante worker heterogeneity (prior to investment) or from ex post inefficient investment decisions.

\textbf{Scale.} \cite{eeckhout} study a model of one-to-one matching with endogenous \emph{scale}, where each firm $x$ hires a single type $y$ of worker, but can hire any measure of them (e.g., because the output of worker $y$ at firm $x$ depends only on $x$, $y$, and the total measure of workers in the firm, so there is no reason for a firm to hire heterogeneous workers). Their main result characterizes positive assortative matching in this model. In contrast, our model allows flexible workforce composition, but a fixed scale. A natural mutual generalization would allow flexibility in both composition and scale. For instance, characterizing PAS in such a model is a clear open direction.

\textbf{Many-to-Many Matching.} Our results extend straightforwardly to many-to-many matching, such as forming schools consisting of students and faculty. The many-to-many case is actually slightly simpler, because the possibility of mixing on both sides symmetrizes the analysis. For example, if the production function is mean-measurable on both sides, equilibrium features mean-preserving contractions of $F$ and $G$, with compression regions corresponding to affine regions of $p$ and $q$, which are both convex.

\bibliographystyle{ecta-fullname} 
\bibliography{bibliography}  %

\newpage

\appendix

\section{Proofs} \label{a:proofs}

\subsection{Duality in the General Case}
\begin{proof}[Proof of Theorem~\ref{t:e}]	

We start with mathematical preliminaries. A function $q:Y\to \R $ is $L$-Lipschitz, denoted by $q\in \Lip_L(Y)$, if there exists a constant $L$ such that $|q(y)-q(\hat y)|\leq L\rho_Y(y,\hat y)$, for all $y,\hat y\in Y$, where $\rho_Y$ is a metric on $Y$. As in Section 8.10(viii) in \cite{bogachev2007}, we define the Kantorovich-Rubinstein norm on $M(Y)$ by 
\begin{equation}\label{e:KR}
\kr \eta := |\eta(Y)|+ \sup\left\{\int_Y q(y)\eta(\df y):\, q\in \Lip_1(Y),\, q(y_0)=0\right\},
\end{equation}
for a fixed $y_0\in Y$. By Theorem 6.9 and Remark 6.19 in \cite{villani2009}, $\kr{\cdot}$ metrizes the weak$^\star$ topology on the compact metric space $\Delta(Y)$. So $\Delta(X\times \Delta(Y))$ is also a compact metric space.

Primal attainment follows from standard compactness arguments. Let $\Pf$ denote the feasible set in \ref{P} (i.e., the set of all assignments). Because the functions $\gamma\to \gamma_X$ and $\gamma\to\E_\gamma[\eta] $ are continuous, $\Pf$ is compact, as a closed subset of the compact set $\Delta(X\times \Delta(Y))$. Moreover, $\Pf$ is  non-empty, as it contains $\gamma=\mu\otimes\delta_\nu $ (i.e., $\gamma(A,D)=\mu(A)\1\{\nu\in D\}$ for all $A\subset X$ and $D\subset \Delta(Y)$). Since $V$ is Lipschitz continuous, the function $\gamma \to \E_\gamma[V(x,\eta)]$ is also continuous and thus attains its maximum on the compact set $\Pf$, by the Weierstrass Theorem. Thus, \eqref{P} has a solution $\gamma$. 

Next, the value of \eqref{P} is concave and Lipschitz in $(\mu,\nu)$. Define $U:\Delta(X)\times \Delta(Y)\to \R$  by $U(\mu,\nu)=\max_{\gamma\in \Pf}\E_\gamma [V(x,\eta)]$ for all $(\mu,\nu)\in \Delta(X)\times\Delta(Y)$. 
 Fix $\mu,\hat \mu\in \Delta(X)$, $\nu,\hat \nu\in \Delta(Y)$, and $\lambda\in [0,1]$. By primal attainment, there exist $\gamma\in \Pf$ and $\hat \gamma\in \mathcal P(\hat \mu,\hat \nu)$ such that $U(\mu,\nu)=\E_\gamma [V(x,\eta)]$ and $U(\hat \mu,\hat \nu)=\E_{\hat\gamma} [V(x,\eta)]$. Note that $\gamma_\lambda\in \mathcal P(\mu_\lambda,\nu_\lambda)$, where $\mu_\lambda=(1-\lambda)\mu+\lambda\hat \mu$, $\nu_\lambda=(1-\lambda)\nu+\lambda\hat \nu$, and  $\gamma_\lambda=(1-\lambda)\gamma+\lambda\hat\gamma$. Thus, $U(\mu_\lambda,\nu_\lambda)\geq \E_{\gamma_\lambda} [V(x,\eta)]=(1-\lambda)U(\mu,\nu)+\lambda U(\hat \mu,\hat \nu)$, showing that $U$ is concave.
 
To show that $U$ is Lipschitz, we rely on the following lemma.

\begin{lemma}\label{l:KR}
Let $\mu,\hat \mu\in \Delta(X)$, $\nu, \hat \nu \in \Delta(Y)$, and $\gamma\in \Pf$. There exists $\hat \gamma\in \mathcal P(\hat \mu,\hat \nu)$ such that $\kr{\gamma-\hat \gamma}\leq \kr{\mu-\hat \mu}+\kr{\nu-\hat \nu}$.
\end{lemma}

Letting $\Gamma(\gamma)$ denote the assignment in Lemma~\ref{l:KR}, and letting $L$ denote the Lipschitz constant for $V$, we have
\begin{align*}
U(\mu,\nu)-U(\hat \mu,\hat \nu) &=\max_{\gamma \in \Pf} \E_\gamma [V(x,\eta)]-\max_{\hat \gamma \in \Pfhat} \E_{\hat \gamma} [V(x,\eta)] \\
 & \leq \max_{\gamma \in \Pf} \big(\E_\gamma [V(x,\eta)]-\E_{\Gamma(\gamma)} [V(x,\eta)]\big) \\
 & \leq \max_{\gamma \in \Pf} L \kr{\gamma-\Gamma(\gamma)} \leq L (\kr{\mu-\hat \mu}+\kr{\nu- \hat \nu}),
\end{align*}
establishing Lipschitz continuity (after interchanging $(\mu,\nu)$ and $(\hat \mu,\hat \nu)$).

\begin{proof}[Proof of Lemma~\ref{l:KR}]
For notational convenience, denote $\hat X=X$ and $\hat Y=Y$. By the Kantorovich-Rubinstein theorem (Theorem 8.10.45 in \citealp{bogachev2007}), there exist optimal transport plans between $\mu$ and $\hat \mu$ and between $\nu$ and $\hat \nu$: $\alpha\in \Delta(X\times \hat X)$ with $\alpha_X=\mu$ and $\alpha_{\hat X}=\hat \mu$ and $\beta\in \Delta(Y\times \hat Y)$ with $\beta_Y=\nu$ and $\beta_{\hat Y}=\hat \nu$ such that
\[
\kr{\mu-\hat \mu}=\int_{X\times \hat X} \rho_X(x,\hat x)\alpha(\df x,\df \hat x)\quad\text{and}\quad \kr{\nu-\hat \nu}=\int_{Y\times \hat Y} \rho_Y(y,\hat y)\beta(\df y,\df \hat y).
\]	
By Theorem 10.4.5 in \cite{bogachev2007}, $\alpha$ and $\beta$ have conditional distributions $\alpha(\hat x|x)$ and $\beta(\hat y|y)$. Define the pushforward mapping $N:\Delta (Y)\to \Delta (Y)$ by $N(\eta)(\hat A)=\int_Y \beta(\hat A|y) \eta(\df y)$ for all $\eta \in \Delta(Y)$ and $\hat A\subset \hat Y$. Define the joint distribution $\zeta \in \Delta(X\times \hat X\times \Delta(Y)\times \Delta(\hat Y))$ by
 \[
 \zeta (A, \hat A,D,\hat D)= \int_{A\times D}\1 \{N(\eta)\in \hat D \} \alpha(\hat A|x)\gamma (\df x,\df \eta) ,
 \]
for all $A,\hat A\subset X$ and $D,\hat D\subset \Delta(Y)$.
Note that $\zeta_{X\times\Delta(Y)}=\gamma$ and $\zeta_{X\times X}=\alpha$. Define the projection $\hat \gamma = \zeta_{\hat X\times \Delta(\hat Y)}$, and note that $\hat \gamma \in \mathcal P(\hat\mu,\hat \nu)$, as $\hat \gamma_{\hat X}=\hat \mu$ and $\E_{\hat \gamma}[\hat \eta]=\hat \nu$. Finally, define $\sigma=\zeta_{\Delta(Y)\times \Delta(\hat Y)}$, and note that $\E_\zeta[\eta]=\nu$ and $\E_\zeta[\hat \eta]=\hat \nu$.

Then, $\kr{\gamma-\hat \gamma}\leq \kr{\mu-\hat \mu}+\kr{\nu-\hat \nu}$, as, for any $\hat V\in \Lip_1(X\times \Delta(Y))$, we have
\begin{align*}
&\int_{X\times \Delta(Y)}\hat V(x,\eta)(\gamma -\hat \gamma)(\df x,\df \eta) \\
= & \int_{X\times \hat X\times \Delta(Y)\times \Delta(\hat Y)} (\hat V(x,\eta)-\hat V(\hat x,\hat \eta))\zeta(\df x,\df \hat x,\df \eta,\df \hat \eta) \\
=& \int_{X\times \hat X\times \Delta(Y)\times \Delta(\hat Y)} (\hat V(x,\eta)-\hat V(\hat x,\eta)+\hat V(\hat x,\eta)-\hat V(\hat x,\hat \eta))\zeta(\df x,\df \hat x,\df \eta,\df \hat \eta) \\
\leq &  \int_{X\times \hat X} \rho_X(x,\hat x)\alpha(\df x ,\df \hat x) +\int_{\Delta(Y)\times \Delta(\hat Y)} \kr{\eta -\hat \eta}\sigma(\df \eta,\df \hat \eta)\\
= & \kr{\mu -\hat \mu} + \kr{\nu-\hat \nu},
\end{align*}
where the last line is by Lemma 3 in \cite{DK}, restated as Claim~\ref{l:L}.
\begin{claim}\label{l:L}
$\int_{\Delta(Y)\times \Delta(\hat Y)} \kr{\eta-\hat \eta}\sigma(\df \eta,\df \hat \eta)=\kr{\nu-\hat \nu}$.
\end{claim}
\begin{proof}
Define $\tau=\sigma_{\Delta(Y)}$, $\chi\in \Delta(Y\times \Delta(Y))$ by $\chi (B,D)=\int_D \eta(B)\tau(\df \eta)$ for all $B\subset Y$ and $D\subset \Delta(Y)$, and $\xi\in \Delta(Y\times \hat Y\times \Delta(Y))$ by $\xi (B,\hat B,D)=\int_B \beta(\hat B|y)\chi(D|y)\nu (\df y)$ for all $B,\hat B\subset Y$ and $D\subset \Delta(Y)$. Note that $\xi (B,\hat Y|\eta)=\eta(B)$ and $\xi(Y,\hat B|\eta)=N(\eta)(\hat B)$ for all $B,\hat B\subset Y$, so $\kr{\eta-N(\eta)}\leq \int_{Y\times \hat Y}\rho_Y(y,\hat y)\xi(\df y,\df \hat y|\eta)$. Then
\begin{align*}
\int_{\Delta(Y)\times \Delta(\hat Y)} \kr{\eta-\hat \eta}\sigma(\df \eta,\df \hat \eta)&=\int_{\Delta(Y)} \kr{\eta -N(\eta)}\tau (\df \eta)\\
&\leq \int_{\Delta(Y)} \int_{Y\times \hat Y} \rho_Y(y,\hat y) \xi (\df y,\df \hat y|\eta) \tau(\df \eta)\\
&=  \int_{Y\times \hat Y} \rho_Y (y,\hat y) \beta (\df y,\df \hat y) \\
&=\kr{\nu-\hat \nu}\\
&= \kr{\int_{\Delta(Y)} \eta\tau (\df \eta) -\int_{\Delta(Y)}N(\eta)\tau (\df \eta)}\\
&\leq \int_{\Delta(Y)} \kr{\eta -N(\eta)}\tau (\df \eta),
\end{align*}
implying that both inequalities hold with equality. 
\end{proof}
\end{proof}

Next, $U$ is superdifferentiable, because it is concave and has bounded steepness at $(\mu,\nu)$ (which is implied by Lipschitz continuity). Indeed, the Duality Theorem in \cite{Gale1967} implies that there exists a continuous linear function $H$ on $M(X)\times M(Y)$ such that $H(\hat \mu-\mu,\hat \nu-\nu)\geq U(\hat \mu,\hat \nu)-U(\mu,\nu)$ for all $\hat \mu\in \Delta(X)$ and $\hat \nu \in \Delta(Y)$. By the Kantorovich-Rubinstein theorem (Exercise 8.10.143 in \citealp{bogachev2007} or Theorem 0 in \citealp{hanin}), there exist $p\in \Lip(X)$ and $q\in \Lip(Y)$ such that $H(\hat \mu-\mu,\hat \nu-\nu)=\int_X p(x)(\hat \mu -\mu)(\df x)+\int_Y q(y)(\hat \nu-\nu)(\df y)$ for all $\hat \mu\in \Delta(X)$ and $\hat \nu \in \Delta(Y)$. Adding a constant to $p$ if necessary, we obtain
\begin{align}
U (\mu,\nu)=&\int_X p(x) \mu(\df x)+\int_Y q(y)\nu(\df y),\label{e:U=H}\\
U (\hat \mu,\hat \nu)\leq &\int_X p(x) \hat \mu(\df x)+\int_Y q(y)\hat \nu(\df y),\quad \text{for all $\hat \mu\in \Delta(X)$ and $ \hat \nu\in \Delta(Y)$.}\label{e:U<H}
\end{align}

Given that $U$ has a supergradient $(p,q)$ at $(\mu,\nu)$, standard duality arguments imply equilibrium existence and its equivalence to solutions to \eqref{P} and \eqref{D}. Indeed, note that $(p,q) \in \Df$, where $\Df$ denotes the set of feasible solutions to \eqref{D}, as, for each $x\in X$ and $\eta\in \Delta(Y)$, we have
\begin{equation}
V(x,\eta)\leq U (\delta_x,\eta)\leq p(x)+\int_Y q(y) \eta(\df y), \label{e:dual}
\end{equation}
where the first inequality is by $\delta_{(x,\eta)}\in \mathcal P(\delta_x,\eta)$, and the second inequality is by \eqref{e:U<H} evaluated at $(\hat \mu,\hat \nu)=(\delta_x,\eta)$. Letting $\gamma$ be a solution to \eqref{P}, we have
\begin{align*}
&\, \int_{X\times \Delta(Y)} \left (V(x,\eta)-p(x)-\int_Y q(y)\eta (\df y)\right )\gamma(\df x,\df \eta)\\
=& \, U(\mu,\nu)-\int_X p(x)\mu(\df x)-\int_Y q(y)\nu(\df y)=0,	
\end{align*}
where the first equality is by the definition of $\gamma\in \Pf$, and the second equality is by \eqref{e:U=H}. Hence, \eqref{e:dual} holds with equality for $\gamma$-almost all $(x,\eta)$ and thus for all $(x,\eta)\in \supp(\gamma)$, by continuity of $V(x,\eta)-p(x)-\int_Y q(y) \eta(\df y)$ in $(x,\eta)$. So, $(\gamma;p,q)$ is a competitive equilibrium, which proves equilibrium existence. Next, if $(\gamma;p,q)$ is an equilibrium, then $\gamma$ solves \eqref{P} and $(p,q)$ solves \eqref{D}, because, for each $\tilde \gamma\in \Pf$ and $(\tilde p,\tilde q)\in \Df$, we have
\begin{align*}
\int_X \tilde p(x)\mu(\df x)+\int_Y \tilde q(y)\nu(\df y)&=\int_{X\times \Delta(Y)}  \left( \tilde p(x)+\int_Y \tilde q(y)\eta(\df y) \right)\gamma(\df x,\df \eta)\\
&\geq \int_{X\times \Delta(Y)} V(x,\eta)\gamma(\df x,\df \eta)\\
&=\int_{X\times \Delta(Y)}  \left( p(x)+\int_Y  q(y)\eta(\df y) \right)\gamma(\df x,\df \eta)\\
&=\int_X p(x)\mu(\df x)+\int_Y q(y)\nu(\df y)\\
&= \int_{X\times \Delta(Y)}  \left( p(x)+\int_Y q(y)\eta(\df y) \right)\tilde \gamma(\df x,\df \eta)\\
&\geq \int_{X\times \Delta(Y)} V(x,\eta)\tilde \gamma(\df x,\df \eta).
\end{align*}
Conversely, if $\gamma$ solves \eqref{P} and $(p,q)$ solves \eqref{D}, then $(\gamma;p,q)$ is an equilibrium. Indeed, since there exists an equilibrium $(\tilde \gamma;\tilde p,\tilde q)$, we have
\begin{align*}
0&= \int_{X\times \Delta(Y)} \left (V(x,\eta)-\tilde p(x)-\int_Y \tilde q(y)\eta (\df y)\right )\tilde \gamma(\df x,\df \eta)\\
&=\int_{X\times \Delta(Y)} V(x,\eta)\tilde \gamma(\df x,\df \eta)-\int_X \tilde p(x)\mu(\df x)-\int_Y \tilde q(y)\nu(\df y)\\
&\leq \int_{X\times \Delta(Y)} V(x,\eta) \gamma(\df x,\df \eta)-\int_X  p(x)\mu(\df x)-\int_Y q(y)\nu(\df y)\\
&=\int_{X\times \Delta(Y)} \left (V(x,\eta)- p(x)-\int_Y q(y)\eta (\df y)\right )\gamma(\df x,\df \eta)\leq 0,
\end{align*}
showing that the inequalities hold with equality, so $V(x,\eta)-p(x)-\int_Y q(y) \eta(\df y)=0$ for all $(x,\eta)\in \supp(\gamma )$.
\end{proof}

\begin{proof}[Proof of Remark~\ref{r:price}]
The remark follows from standard optimal transport arguments. Fix a solution $(p,q)$ to \eqref{D}. Let $p^\star:X\to \R$ and $V^\star:\Delta(Y)\to \R$ be given by, for all  $x\in X$ and $\eta\in \Delta(Y)$,
\begin{align*}
p^\star(x) =\max_{\eta\in \Delta(Y)} \left\{V(x,\eta)-\int_Y q(y)\eta (\df y)\right \}\quad \text{and}\quad V^\star(\eta) &=	\max_{x\in X} \{V(x,\eta)-p(x)\},
\end{align*}
where the maximums are attained by continuity and compactness.

By a standard envelope-theorem argument, $p^\star$ and $V^\star$ are $L$-Lipschitz, given that $V$ is $L$-Lipschitz. Indeed, for all $x,\hat x\in X$ and $\eta,\hat \eta \in \Delta(Y) $, we have
\begin{align*}
p^\star (x) &\geq V(x,\eta)-\int_Y q(y)\eta(\df y) \geq V(\hat x,\eta)-\int_Y q(y)\eta(\df y)  -L\rho_X(x,\hat x),\\
V^\star(\eta) &\geq 	V(x,\eta)-p(x)\geq V(x,\hat \eta)-p(x)-L\kr{\eta-\hat \eta},
\end{align*}
so taking the maximum over $\eta\in \Delta(Y)$ and $x\in X$ on the right-hand sides yields $p^\star(x)\geq p^\star (\hat x)-L \rho_X(x,\hat x)$ and $V^\star (\eta)\geq V^\star(\hat \eta)-L\kr{\eta-\hat \eta}$. Then, since \eqref{D} has a solution when $\mu=\delta_{x}$ and $V(x,\eta)=V^\star(\eta)$ for all $x\in X$ and $\eta\in \Delta(Y)$, there exists $q^\star$ that solves
\begin{equation}
	\min_{\hat q\in \Lip(Y)} \int _Y\hat q(y)\nu(\df y)\quad \text{s.t.}\quad \int_Y \hat q(y)\eta(\df y) \geq V^\star(\eta),\quad \text{for all $\eta\in  \Delta(Y)$}.\label{e:PD}
\end{equation} 
Note that $(p^\star,q)\in \Df$ and $p^\star(x)\leq p(x)$ for all $x\in X$, so $p(x)=p^\star(x)$ for all $x\in \supp(\mu)$, since $(p,q)$ solves \eqref{D}. Similarly, $(p,q^\star)\in \Df$ and $\int_Y q^\star(y)\nu (\df y)\leq \int _Y q(y)\nu (\df y)$, so $q$ also solves \eqref{e:PD}, since $(p,q)$ solves \eqref{D}.
\end{proof}

\subsection{Positive Assortative Segregation}

\begin{proof}[Proof of Theorem~\ref{t:fs}]
\emph{Part 3.} By Definition~\ref{d:CE} (take $\eta=\delta_y$), if $(\phi;p,q)$ is a competitive equilibrium, then
\[
p(x)=V(x,\delta_{S(x)})-q(S(x))=\max_{y\in Y} \{V(x,\delta_y)-q(y)\},
\]
so, by the envelope theorem (\citealp{MS}, Theorem 3) and continuity of $V_x$ and $S$, $p$ has a derivative $p'(x)=V_x(x,\delta_{S(x)})$ at each $x\in (\ul x,\ol x)$, yielding \eqref{e:ps}; and \eqref{e:qs} follows from $p(S^{-1}(y))+q(y)=V(S^{-1}(y),\delta_y)$.

\emph{Part 1.} By part 3, Definition~\ref{d:CE}, and some algebra, PAS is a competitive equilibrium if and only if \eqref{e:FS} holds. In addition, if \eqref{e:FS} holds strictly for all $(x,\eta)$ with $\eta\neq \delta_{S(x)}$, then, by Theorem~\ref{t:e} and continuity of $V$, $p$, and $q$, for any equilibrium assignment $\gamma$, we have 
\[
\supp(\gamma)\subset \Big \{(x,\eta)\in X\times \Delta (Y) :p(x)+\int_Y q(y)\, \eta (\df y)=V(x,\eta)\Big \}=\{(x,\delta_{S(x)}):x\in X\},
\]
showing that $\gamma=\phi$.

\emph{Part 2.} Suppose \eqref{e:spm} and \eqref{e:oc} hold. Let $p$ and $q$ be given by \eqref{e:ps} and \eqref{e:qs}. Then
\begin{align*}
p(x)+\int_Y q(y)\, \eta(\df y) &=\int_Y \Big(V(x,\delta_y)+\int_{S(x)}^y\int_{x}^{S^{-1}(\tilde y)}V_{xy}(\tilde x,\delta_{\tilde y})\,\df \tilde x\,\df \tilde y\Big)\,\eta(\df y)\\
&\geq \int_Y V(x,\delta_y)\, \eta(\df y) \geq V(x,\eta),
\end{align*}
for all $(x,\eta)\in X\times \Delta(Y)$, so $\phi$ is an equilibrium assignment. In addition, if $V_{xy}(x,\delta_y)>0$ for all $(x,y)\in X\times Y$, then the first inequality holds strictly for all $(x,\eta)$ with $\eta\neq \delta_{S(x)}$, so $\phi$ is the unique equilibrium assignment.

Conversely, suppose that $V_{xy}(x,\delta_y)< 0$ for some $(x,y)\in X\times Y$. By continuity of $V_{xy}$, there exist $x'>x$ in $X$ and $y'>y$ in $Y$ such that $V_{xy}(\tilde x,\delta_{\tilde y})<0$ for all $(\tilde x,\tilde y)\in [x,x']\times [y,y']$. Let $\mu$ and $\nu$ be such that $S(x)=y$ and $S(x')=y'$. Then, for $p$ and $q$ given by \eqref{e:ps} and \eqref{e:qs},
\begin{align*}
p(x)+q(y')-V(x,\delta_{y'})=\int_y^{y'}\int_x^{S^{-1}(\tilde y)}V_{xy}(\tilde x,\delta_{\tilde y})\, \df \tilde x\, \df \tilde y<0,
\end{align*}
showing that \eqref{e:FS} fails at $(x,\delta_{y'})$, so $\phi$ is not an equilibrium assignment.

Finally, suppose that $\int_Y V(x,\delta_y)\,\eta(\df y)+\varepsilon< V(x,\eta)$ for some $(x,\eta)\in X\times \Delta(Y)$ and $\varepsilon>0$. By continuity of $V_{xy}$, there exist $\mu\in \Delta(X)$ (with small enough $\kr{\mu-\delta_x}$) and $\nu\in \Delta(Y)$ such that $\int_{S(x)}^y\int_{x}^{S^{-1}(\tilde y)}V_{xy}(\tilde x,\delta_{\tilde y})\,\df \tilde x\,\df \tilde y<\varepsilon$ for all $y\in Y$. Then, for $p$ and $q$ given by \eqref{e:ps} and \eqref{e:qs},
\begin{align*}
p(x)+\int_Y q(y)\, \eta(\df y) &=\int_Y \Big(V(x,\delta_y)+\int_{S(x)}^y\int_{x}^{S^{-1}(\tilde y)}V_{xy}(\tilde x,\delta_{\tilde y})\,\df \tilde x\,\df \tilde y\Big)\,\eta(\df y)\\
&< \int_Y V(x,\delta_y)\, \eta(\df y)+\varepsilon < V(x,\eta),
\end{align*}
showing that \eqref{e:FS} fails at $(x,\eta)$, so $\phi$ is not an equilibrium assignment.
\end{proof}

\begin{proof}[Proof of Remark~\ref{r:moment}]
Fix $x\in X$ and $\eta\in \Delta(Y)$. By the Choquet Theorem (\citealp{winkler}, Theorem 3.1) and the Richter-Rogosinsky Theorem (\citealp{winkler}, Theorem 2.1), there exists $\tau \in \Delta(\Delta(Y))$ such that $\E_\tau [\tilde \eta]=\eta$, and, for each $\tilde \eta\in \supp(\tau)$, we have $|\supp(\tilde \eta)|\leq 2$ and $\E_{\tilde\eta}[m(x,y)]=\E_{\eta}[m(x,y)]$. Thus, if \eqref{e:FS} holds at all $\tilde \eta\in \supp(\tau)$, then 
\begin{gather*}
p(x)+\E_\eta [q(\tilde y)]=p(x)+\E_\tau [\E_{\tilde \eta}[q(\tilde y)]]\\
\geq \E_\tau[\E_{\tilde \eta} [v(x,y,\E_{\tilde\eta}[m(x,\tilde y)])]]=\E_{\eta} [v(x,y,\E_{\eta}[m(x,\tilde y)])],	
\end{gather*}
showing that \eqref{e:FS} holds at $\eta$. In addition, if $\eta\neq \delta_{S(x)}$ and \eqref{e:FS} holds strictly for all $\tilde \eta\in \supp(\tau)$ with $\tilde \eta\neq \delta_{S(x)}$, then $\supp(\tau)$ contains some $\tilde \eta \neq \delta_{S(x)}$ (by $\E_\tau [\tilde \eta]=\eta\neq \delta_{S(x)}$), at which the inequality is strict, so \eqref{e:FS} holds strictly at $\eta$. Similarly, if \eqref{e:oc} holds at all $\tilde \eta\in \supp(\tau)$, then 
\[
\int_Y V(x,\delta_{y})\,\eta(\df y)=\int_{\Delta(Y)}\int_Y V(x,\delta_{y})\,\tilde \eta(\df y)\, \tau(\df \tilde \eta)\geq \int_{\Delta(Y)} V(x,\tilde \eta )\, \tau(\df \tilde \eta)= V(x,\eta),
\]
showing that \eqref{e:oc} holds at $\eta$.
\end{proof}

\subsection{Duality in the Mean-Measurable Case}
This appendix allows for general $v\in \Lip (X\times Y)$, which may not be continuously differentiable, and general $F\in \Delta(X)$ and $G\in \Delta(Y)$, which may not have full support or densities. Theorem~\ref{t:el} continues to hold verbatim, and Lemma~\ref{l:CEeqv} is generalized to Lemma~\ref{l:CEeqvg}.
	
We start with three lemmas.
\begin{lemma}\label{l:Lip}
If $v\in \Lip_L(X\times Y)$, then $V\in \Lip_L(X\times \Delta(Y))$.
\end{lemma}
\begin{proof}
If $v$ is $L$-Lipschitz, then so is $V$, as, for all $x,\hat x\in X$ and $\eta,\hat \eta\in \Delta(Y)$, we have
\begin{align*}
V(x,\eta)-V(\hat x,\hat \eta)&=v(x,\E_\eta[y])-v(\hat x,\E_\eta[y])+v(\hat x,\E_{\eta}[y])-v(\hat x,\E_{\hat \eta}[y])\\
&\leq L\left(|x-\hat x|+\left|\int_Y y\eta(\df y)-\int_Y y\hat\eta(\df y) \right|\right)\\
&\leq L\left(|x-\hat x|+\kr{\eta-\hat \eta}\right).\qedhere
\end{align*}
\end{proof}

\begin{lemma}\label{l:Blackwell}
There exists $\gamma\in \Pf$ that induces $J\in \Delta(X\times Y)$ if and only if $J_X=F$ and $J_Y\succeq  G$. 
\end{lemma}
\begin{proof}
Let $\gamma \in \Pf$ and let $J^\gamma\in \Delta(X\times Y)$ denote the joint distribution of $(x,\E_\eta[y])$ induced by $\gamma$. Then $J^\gamma_X=F$, by $\gamma_X=\mu$, and $J^\gamma_Y\succeq  G $, because, for each $h\in \Conv (Y)$, we have, by Jensen's inequality and $\E_\gamma[\eta]=\nu$,
\[
\E_{J^\gamma_Y}[h(y)]=\E_{\gamma}[h(\E_\eta[y])]\leq \E_{\gamma} [\E_{\eta}[h(y)]]=\E_G [h(y)].
\]
Conversely, let $J\in \Delta(X\times Y)$ be such that $J_X=F$ and $J_Y\succeq  G$, and let $\mu$ and $\nu$ denote the measures induced by the cdfs $F$ and $G$. By \cite{blackwell1953equivalent}, there exists $\mathcal N:Y\to \Delta(Y)$ such that $\E_{\mathcal N(z)} [y]=z$, for all $z\in Y$, and $\E_{J_Y} [\mathcal N(z)] =\nu $. Define $\gamma\in \Delta(X\times \Delta(Y))$ by 
\[
\gamma(A,D)=\int_{A\times Y} \1\{\mathcal N(z)\in D\}\, J (\df x,\df z),\quad \text{for all $A\subset X$ and $D\subset \Delta(Y)$}.
\]
Then $J^\gamma_X=J_X$ and $J^\gamma_Y=J_Y$. Moreover, $\gamma_X=\mu$, by $J_X=F,$ and
$\E_\gamma[\eta]=\E_{J_Y} [\mathcal N(z)]=\nu$, by the construction of $\mathcal N$.
\end{proof}

Let $\DfL$ denote the set of feasible solutions to \eqref{DL}. For any $q\in \Lip(Y)$, define $\check q: Y\to \R$ by
\begin{equation}\label{e:qc}
\check{q}(z)=\min_{\eta\in \Delta(Y)} \E_\eta [q(y)]\quad \text{s.t.}\quad  \E_\eta[y]=z,\quad\text{for all $z\in Y$.}
\end{equation} 

\begin{lemma}\label{l:dual}
If $(p,q)\in \DfL $, then $(p,q)\in \Df $. If $(p,q)\in \Df $, then $(p,\check q)\in \DfL $.
\end{lemma}
\begin{proof}
If $(p,q)\in \DfL$, then, for all $x\in X$ and $\eta\in \Delta(Y)$, we have
\[
p(x)+\E_\eta [q(y)]\geq p(x)+q(\E_\eta[y])\geq v(x,\E_\eta[y])=V(x,\eta),
\]
so $(p,q)\in \Df$. 

Let $(p,q)\in \Df$.  Up to notation and sign change, \eqref{e:qc} features the same optimization problem as \eqref{P} when $|\supp(\mu)|=1$ and $|\supp(\nu)|=2$. Thus, as follows from Section~\ref{s:dual}, $\check q$ is Lipschitz continuous and convex. Moreover, letting $\eta_z$ denote a minimizer in \eqref{e:qc}, $(p,q)\in \Df$ yields
\[
p(x)+\check q(z)=p(x)+\E_{\eta_z}[q(y)]\geq V(x,\eta_z)=v(x,z),
\]
showing that $(p,\check q)\in \DfL$. 
\end{proof}

\begin{lemma}\label{l:CEeqvg}
Suppose that $\gamma\in \Delta(X\times \Delta(Y))$ induces $J\in \Delta(X\times Y)$. If $(J;p,q)$ is a competitive equilibrium in the sense of Definition~\ref{d:CEL}, then $(\gamma;p,q)$ is a competitive equilibrium in the sense of Definition~\ref{d:CE}. Conversely, if $(\gamma;p,q)$ is a competitive equilibrium in the sense of Definition~\ref{d:CE}, then $(J;p,\check q)$ is a competitive equilibrium in the sense of Definition~\ref{d:CEL}, and $\check q(y)=q(y)$ for all $y\in \supp(G)$.
\end{lemma}
\begin{proof}
If $(J;p,q)$ is a competitive equilibrium in the sense of Definition~\ref{d:CEL}, then $(\gamma;p,q)$ is a competitive equilibrium in the sense of Definition~\ref{d:CE}, because
\begin{align*}
\int_{X\times \Delta(Y)} V(x,\eta)\, \gamma(\df x,\df \eta)&\leq \int_{X\times \Delta(Y)} \Big (p(x)+\int_Y q(y)\, \eta(\df y)\Big )\, \gamma(\df x,\df \eta)\\
&= \int_X p(x)\, \mu(\df x)+\int_Y q(y)\, \nu (\df y)\\
&= \int_X p(x)\, F(\df x)+ \int_Y q(z)\, J_Y(\df z)\\
&= \int_{X\times Y} v(x,z)\, J(\df x,\df z)\\
&= \int_{X\times \Delta(Y)} V(x,\eta)\, \gamma(\df x,\df \eta),
\end{align*}
where the inequality holds by $(p,q)\in \Df$ (see Lemma~\ref{l:dual}). Conversely, if $(\gamma;p,q)$ is a competitive equilibrium in the sense of Definition~\ref{d:CE}, then $(J;p,\check q)$ is a competitive equilibrium in the sense of Definition~\ref{d:CEL}, and $\check q(y)=q(y)$ for all $y\in \supp(G)$, because
\begin{align*}
\int_{X\times Y} v(x,z)\, J(\df x,\df z)&\leq  \int_{X\times Y} (p(x)+\check q(z))\, J(\df x,\df z)\\
&=\int_X p(x)\, F(\df x)+\int_Y \check q(z)\, J_{Y}(\df z)\\
&\leq  	\int_X p(x)\, \mu (\df x)+\int_Y \check q(y)\, \nu (\df y)\\
&\leq  	\int_X p(x)\, \mu (\df x)+\int_Y q(y)\, \nu (\df y)\\
&= \int_{X\times \Delta(Y)} V(x,\eta)\, \gamma(\df x,\df \eta)\\
&=  \int_{X\times Y} v(x,z)\, J(\df x,\df z),
\end{align*}
where the first inequality holds by $(p,\check q)\in \DfL $, the second inequality holds by $\check q\in \Conv(Y)$ and $J_Y\succeq G$, and the third inequality holds by $\check q(y)\leq q(y)$ for all $y\in Y$.
\end{proof}
\begin{proof}[Proof of Theorem~\ref{t:el}]
By Lemma~\ref{l:Lip}, Theorem~\ref{t:e} implies that there exists a competitive equilibrium $(\gamma;p,q)$ in the sense of Definition~\ref{d:CE}. Let $J$ denote the joint distribution of $(x,\E_\eta[y])$ induced by $\gamma$. By Lemma~\ref{l:CEeqvg}, $(J;p,\check q)$ is a competitive equilibrium in the sense of Definition~\ref{d:CEL}. Moreover, by Theorem~\ref{t:e} and Lemmas~\ref{l:Blackwell}--\ref{l:CEeqvg}, $(J;p,q)$ is a competitive equilibrium if and only if $J$ solves \eqref{PL} and $(p,q)$ solves \eqref{DL}.
\end{proof}

\subsection{Equilibrium Characterization in the Mean-Measurable Case}

We first formalize the intuition for the segregation set $X^S$.
\begin{lemma}\label{l:I}
Let $\gamma \in \Pf$ induce $T$ satisfying \eqref{F}, and let $I\in \Delta(X\times Y)$ be given by
\[
I(x,y)=\int_{X\times\Delta(Y)}\eta([\ul y,y])\1\{\tilde x\leq x\}\, \gamma(\df \tilde x,\df  \eta).
\]
Then $x\in X^S$ if and only if $I(x,S(x))=F(x)$. Moreover,  if $x\in X^S$, with $x\neq \ul x,\ol x$, and $T$ is continuous, then  $T(x)=S(x)$.
\end{lemma}
\begin{proof}
Note that $I_X=F$ and $I_Y=G$. Thus, $I(x,y)\leq \min (F(x),G(y))$ for all $(x,y)\in X\times Y$. Moreover,
\begin{align*}
\int_{\ul x}^{x}T(\tilde x)\,F(\df \tilde x)
=&\int_{X\times \Delta(Y)}\int y\,\eta(\df y)\1\{\tilde x\le x\}\,\gamma(\df \tilde x,\df \eta)\\
=&\int_{X\times \Delta(Y)}\left(\ul y+\int_{\ul y}^{\ol y}(1-\eta([\ul y,y]))\,\df y\right)\1\{\tilde x\le x\}\,\gamma(\df \tilde x,\df \eta)\\
=&\ul yF(x)+\int_{\ul y}^{\ol y}\int_{X\times \Delta(Y)}(1-\eta([\ul y,y]))\1\{\tilde x\le x\}\,\gamma(\df \tilde x,\df \eta)\, \df y\\
=&\ul yF(x)+\int_{\ul y}^{\ol y}(F(x)-I(x,y))\, \df y,
\end{align*}
and, similarly,
\begin{align*}
\int_{\ul x}^{x}S(\tilde x)\,F(\df \tilde x)=&\ul yF(x)+\int_{\ul y}^{\ol y}(F(x)-\min(F(x),G(y)))\, \df y,
\end{align*}
so
\[
\int_{\ul x}^{x}T(\tilde x)\,F(\df \tilde x)-\int_{\ul x}^{x}S(\tilde x)\,F(\df \tilde x)=\int_{\ul y}^{\ol y}(\min(F(x),G(y))-I(x,y))\, \df y.
\]
Since the integrand is nonnegative, $x\in X^S$ if and only if $I(x,y)=\min (F(x),G(y))$ for all $y\in Y$, which is equivalent to $I(x,S(x))=F(x)$, since $G(S(x))=F(x)$.

Finally, if $x\in X^S\cap (\ul x,\ol x)$ and $T$ is continuous, then $\int_{\ul x}^{x}T(\tilde x)\,F(\df \tilde x)-\int_{\ul x}^{x}S(\tilde x)\,F(\df \tilde x)$ attains a minimum of $0$ at an interior point $x$, so $T(x)=S(x)$ by the first-order condition.
\end{proof}

\begin{proof}[Proof of Theorem~\ref{t:eqm}]
By Theorem~\ref{t:el}, $J$ solves \eqref{PL} if and only if there exists $(p,q)\in \DfL $ satisfying
\begin{gather*}
p(x)+q(z)=v(x,z),\quad \text{for all }(x,z)\in \supp (J),
\\
\int_Y q(z)\, J_Y(\df z)=\int_Y q(z)\, G(\df z).
\end{gather*} 
Define the contact set 
\[
\Gamma:=\{(x,z)\in  X\times Y:\, p(x)+q(z)=v(x,z)\},
\]
and notice that $\Gamma$ is a compact set containing $\supp(J)$. We first prove several lemmas.
\begin{lemma}\label{l:uniq}
There exists a continuous, increasing $T:X\to Y$  such that 
\[\Gamma=\{(x,T(x)):x\in X\}.\]
\end{lemma}
\begin{proof}
Since $\Gamma$ is closed and $J_X=F$ has full support on $X$, the set $\Gamma_x:=\{z\in Y:(x,z)\in \Gamma\}$ is nonempty for all $x\in X$.
Suppose for contradiction that there exist $x_0\leq x_1$ in $X$ and $z_0>z_1$ in $Y$ such that $(x_0,z_0)$ and $(x_1,z_1)$ are in $\Gamma$. Then, denoting $z=(z_0+z_1)/2$, we have
\begin{multline*}
p(x_0)+q(z_0)+p(x_1)+q(z_1)\geq p(x_0)+p(x_1)+2q(z)\geq v(x_0,z)+v(x_1,z)\\
>\tfrac{1}{2}(v(x_0,z_0)+v(x_0,z_1))+\tfrac{1}{2}(v(x_1,z_0)+v(x_1,z_1))\\
\geq v(x_0,z_0)+v(x_1,z_1)=p(x_0)+q(z_0)+p(x_1)+q(z_1),
\end{multline*}
where the first inequality is by $q\in \Conv(Y)$, the second inequality is by $(p,q)\in \DfL$, the third inequality is by strict concavity of $v$ in $z$, the fourth inequality is by supermodularity of $v$, and the equality is by the definition of $\Gamma$.
\end{proof}

Fix a continuous, increasing map $T$ satisfying \eqref{F}. Let $J\in \Delta(X\times Y)$ be given by $J(x,z)=\min (F(x),H(z))$ for all $(x,z)\in X\times Y$, with $H\in \Delta(Y)$ given by $H(z)=F(T^{-1}(z))$ for all $z\in Y$, where $T^{-1}$ is a generalized inverse of $T$, defined by $T^{-1}(z)=\max\{x\in X:\, T(x)\leq z\}$, for all $z\ge T(\ul x)$, and $T^{-1}(z)=\ul x$, for all $z<T(\ul x)$.  Then $J$ solves \eqref{PL} if and only if there exists $(p,q)\in \DfL$ such that
\begin{gather}
p(x)+q(T(x))=v(x,T(x)),\quad \text{for all } x\in X,\label{e:p+q=v}\\
\int_X q(S(x))\, F(\df x)=\int_X q(T(x))\, F(\df x).\label{e:qS=qT}
\end{gather}
Define the sets 
\begin{align*}
Y^S=&\left \{y\in Y:  \int_{\ul y}^y G(\tilde y)\, \df \tilde y=\int_{\ul y}^y H(\tilde y)\, \df \tilde y\right \},\quad \text{and}\\
Y^C=&\left \{y\in Y:  \int_{\ul y}^y G(\tilde y)\, \df \tilde y>\int_{\ul y}^y H(\tilde y)\, \df \tilde y\right \}=Y\setminus Y^S.	
\end{align*}
Notice that $Y^C$ is open and thus is a union of at most countably many disjoint open intervals, $Y^C=\bigcup_k(\ul y_k,\ol y_k)$. Moreover, recall that any $q\in \Lip(Y)\cap \Conv(Y)$ admits a right-continuous, bounded, and increasing subderivative $q'$.

\begin{lemma}\label{l:DM}
For $q\in \Lip (Y)\cap \Conv(Y)$, \eqref{e:qS=qT} holds if and only if $q$ is affine on each $(\ul y_k,\ol y_k)$. 
\end{lemma}
\begin{proof}
We have
\begin{multline*}
\int_X (q(S(x))-q(T(x)))\, F(\df x)=\int_Y q(y)\, \df G(y) - 	\int_Y q(y)\, \df H(y)\\=\int_Y (H(y)-G(y))\, q(\df y)
=\int_Y \Big(\int_{\ul y}^y G(\tilde y)\, \df \tilde y-\int_{\ul y}^y H(\tilde y)\, \df \tilde y\Big )\, q'(\df y),
\end{multline*}
where the first equality is by a change of variables, the second equality is by integration by parts and $G(\ol y)=H(\ol y)=1$, and the third equality is by integration by parts and $\int_{\ul y}^{\ol y} G(y)\, \df  y=\int_{\ul y}^{\ol y} H(y)\, \df y$. Since the integrand is (weakly) positive and $q'$ is increasing, \eqref{e:qS=qT} holds if and only if $q'$ is constant on each $(\ul y_k,\ol y_k),$ where the integrand is strictly positive.
\end{proof}

\begin{lemma}\label{l:IK}
Fix $x\in (\ul x,\ol x)$. Then $x\in X^C$ if and only if $T(x)\in Y^C$.
\end{lemma}
\begin{proof}
Integration by parts gives
\begin{align}
\int_{\ul x}^x T(\tilde x)\, F(\df \tilde x)+\int_{\ul y}^{T(x)} H(y)\, \df y=& F(x) T(x),	\label{e:ORT}\\
\int_{\ul x}^x S(\tilde x)\, F(\df \tilde x)+\int_{\ul y}^{S(x)} G(y)\, \df y=& F(x) S(x). \label{e:ORS}
\end{align}
If $x\in X^S$, then $T(x)=S(x)$ by Lemma~\ref{l:I}. Thus, $T(x)\in Y^S$ by \eqref{e:ORT} and \eqref{e:ORS}. Conversely, if $T(x)\in Y^S$, then $G(T(x))=H(T(x))$ (equivalently, $T(x)=S(x)$), since $\int_{\ul y}^y G(\tilde y)\, \df \tilde y-\int_{\ul y}^y H(\tilde y)\, \df \tilde y$ attains a minimum of $0$ at $T(x)$. Thus, $x\in X^S$ by \eqref{e:ORT} and \eqref{e:ORS}.
\end{proof}
\begin{lemma}\label{l:const}
If $J$ solves \eqref{PL} and $(p,q)$ solves \eqref{DL}, then,
for all $x_0,x_1\in (\ul x_i,\ol x_i)$, we have 
\[
q'(T(x_0))=v_z(x_0,T(x_0))=v_z(x_1,T(x_1))=q'(T(x_1)).
\]	
\end{lemma}
\begin{proof}
By Lemma~\ref{l:IK}, there exists $k$ such that $T(x)\in (\ul y_k,\ol y_k)$ for all $x\in (\ul x_i,\ol x_i)$. By Lemma~\ref{l:DM}, $q$ has the same derivative at $T(x_0)$ and $T(x_1)$, i.e., $q'(T(x_0))=q'(T(x_1))$. For $j=0,1$,
\[
p(x_j)=v(x_j,T(x_j))-q(T(x_j))=\max_{y\in Y} \{v(x_j,y)-q(y)\},
\]
where the first equality is by $(x_j,T(x_j))\in \Gamma$ and the second equality is by $(p,q)\in \DfL$.
Since $v$ and $q$ are continuously differentiable in $y$ on $(\ul y_k,\ol y_k)$, we get $q'(T(x_j))=v_z(x_j,T(x_j))$.
\end{proof}
\begin{lemma}\label{l:pq}
If $J$ solves \eqref{PL} and $(p,q)$ solves \eqref{DL}, then \eqref{e:p} and \eqref{e:q} hold for some $c\in \R$. Moreover, $v_z(x,T(x))$ is increasing on $X$.
\end{lemma}
\begin{proof}
Fix $x_0\in (\ul x,\ol x)$. By $(x_0,T(x_0))\in \Gamma$ and $(p,q)\in \DfL$, we have
\[
p(x_0)=v(x_0,T(x_0))-q(T(x_0))=\max_{y\in Y}\{v(x_0,y)-q(y)\}.
\]
By continuity of $v_x$ and $T$, the envelope theorem (\citealp{MS}, Theorem 3) implies that $p$ has a derivative $p'(x_0)=v_x(x_0,T(x_0))$ at $x_0$, so \eqref{e:p} holds for some $c\in \R $. Moreover, for all $y\in [T(\ul x),T(\ol x)]$, by \eqref{e:p}, \eqref{e:p+q=v}, and a change of variables,  we have
\begin{align*}
q(y)&=v(T^{-1}(y),y)-p(T^{-1}(y))
=v(\ul x,T(\ul x))+\int_{T(\ul x)}^y v_z(T^{-1}(\tilde y),\tilde y)\,\df \tilde y -c,
\end{align*}
so \eqref{e:q} holds for all $y\in [T(\ul x),T(\ol x)]$.  Next, if $T(\ul x)>\ul y$, then, since $G$ has positive density on $Y$, there exists $y_0\in (T(\ul x),\ol y]$ such that $\int_{\ul y}^y G(\tilde y)\, \df \tilde y>\int_{\ul y}^y H(\tilde y)\, \df \tilde y$ for all $y\in (\ul y,y_0)$. Thus, by Lemmas~\ref{l:DM}--\ref{l:const}, \eqref{e:q} holds for all $y\in [\ul y,T(\ul x))$. Similarly, if $T(\ol x)<\ol y$, then, since $G$ has positive density on $Y$, we have that \eqref{e:q} holds for all $y\in (T(\ol x),\ol y]$.

Since $q$ is given by \eqref{e:q}, and $q\in \Conv (Y)$, it follows that  $v_z(T^{-1}(y),y)$ is increasing on $[T(\ul x),T(\ol x)]$. Since $v_z$ is increasing in $x$ and $T^{-1}$ is increasing, it follows that $v_z(x,T(x))$ is increasing on $X$.
\end{proof}
\begin{lemma}\label{l:converse}
If $v_z(x,T(x))$ is increasing on $X$ and is constant on each $(\ul x_i,\ol x_i)$, then $(p,q)$ given by \eqref{e:p} and \eqref{e:q}, for any $c\in \R $, is in $\DfL$ and satisfies \eqref{e:p+q=v} and \eqref{e:qS=qT}.
\end{lemma}
\begin{proof}
Substituting $(p,q)$ from \eqref{e:p} and \eqref{e:q}, and using that $v_x$ is increasing in $z$ and $v_z$ is increasing in $x$ and decreasing in $z$,  we get
\begin{gather}
\begin{gathered}\label{e:M}
p(x)+q(T(\hat x))-v(x,T(\hat x))= \int_{\hat x}^x (v_x(\tilde x,T(\tilde x))-v_x(\tilde x,T(\hat x)))\, \df \tilde x\geq 0,	\\
\text{for all $x,\hat x\in X$, with equality at $\hat x=x$},
\end{gathered}\\
\begin{gathered}\label{e:L}
p(x)+q(y)-v(x,y)=\int_{\ul x}^x (v_x(\tilde x,T(\tilde x))-v_x(\tilde x,T(\ul x)))\, \df \tilde x\\
+\int_{y}^{T(\ul x)}(v_z(x,\tilde y)-v_z(\ul x,T(\ul x)))\, \df \tilde y\geq 0,\quad \text{for all  $x\in X$ and $y\in [\ul y,T(\ul x))$},
\end{gathered}\\
\begin{gathered}\label{e:R}
p(x)+q(y)-v(x,y)=\int_x^{\ol x} (v_x(\tilde x,T(\ol x))-v_x(\tilde x,T(\tilde x)))\, \df \tilde x\\
+\int_{T(\ol x)}^{y}(v_z(\ol x,T(\ol x))-v_z(x,\tilde y))\, \df \tilde y\geq 0,\quad \text{for all  $x\in X$ and $y\in (T(\ol x),\ol y]$}.
\end{gathered}
\end{gather}
Since $v_z(x,T(x))$ is increasing in $x$, we have $q\in \Conv(Y)$. Thus, by \eqref{e:M}--\eqref{e:R}, $(p,q)$ is in $\DfL$ and satisfies \eqref{e:p+q=v}. Finally, \eqref{e:qS=qT} holds by Lemmas~\ref{l:DM} and \ref{l:IK}, given that $v_z(x,T(x))$ is constant on each $(\ul x_i,\ol x_i)$, $q'(y)=q'(T(\ul x))$ for all $y\in [\ul y,T(\ul x))$, and $q'(y)=q'(T(\ol x))$ for all $y\in (T(\ol x),\ol y]$.
\end{proof}

We are now ready to prove Theorem~\ref{t:eqm}.

\emph{Part 1.} By Theorem~\ref{t:el}, any equilibrium assignment $J$ satisfies $\supp(J)\subset \Gamma$. Then, by Lemma~\ref{l:uniq}, there is a unique equilibrium matching $T$, which is continuous and increasing. We now show that $T$ is strictly increasing. Suppose for contradiction that there exist $x_0<x_1$ in $X$ and $z\in Y$ such that $z=T(x_0)=T(x_1)$. Note that $v_z(x,T(x))$ is increasing on $X$ by Lemma~\ref{l:pq} and is constant on each $(\ul x_i,\ol x_i)$ by Lemma~\ref{l:const}. 
Since $v_z$ is strictly increasing in $x$ and strictly decreasing in $z$, it follows that $T$ is strictly increasing on each $(\ul x_i,\ol x_i)$, so $[x_0,x_1]\subset X^S$, and $\int_{\ul x}^x T(\tilde x)\,F(\df \tilde x)=\int_{\ul x}^x S(\tilde x)\,F(\df \tilde x)$, for all $x\in [x_0,x_1]$. Then, by Lemma~\ref{l:I}, $S(x)=T(x)=z$ for all $x\in (x_0,x_1),$ yielding $G(z)-G(z_-)\geq F(x_1)-F(x_0)>0$, contradicting that $G$ has a density on $Y$. 

\emph{Part 2.} By part 1, there is a unique increasing equilibrium matching $T$. By Lemmas~\ref{l:const} and \ref{l:pq}, $v_z(x,T(x))$ is increasing on $X$ and is constant on each $(\ul x_i,\ol x_i)$. Conversely, consider an increasing matching $T$ such that $v_z(x,T(x))$ is increasing on $X$ and is constant on each $(\ul x_i,\ol x_i)$. Since $T$ and $v_z(x,T(x))$ are increasing, and $v_z$ is continuous and strictly decreasing in $z$, it follows that $T$ cannot have upward jumps and thus is continuous. Then by Theorem~\ref{t:el} and Lemma~\ref{l:converse}, $T$ is the equilibrium matching.

\emph{Part 3.} By Theorem~\ref{t:el}, the only if part follows from Lemma~\ref{l:pq}, and the if part follows from Lemma~\ref{l:converse}.
\end{proof}
\begin{proof}[Proof of Corollary~\ref{c:eqm}]
\emph{Part 1.} Follows from Theorem~\ref{t:eqm}, since, for an increasing $T$, $v_z(x,T(x))=a'(T(x)-x)+c$ is increasing in $x$ with $\df v_z(x,T(x))/\df x=0$ for all $x\in X^C$ if and only if $T$ is $1$-Lipschitz with $T'(x)=1$ for all $x\in X^C$.  

\emph{Part 2.} It suffices to show that $T$ given by part 2 satisfies part 1. Since $T(x)=x-E'(F(x))=S(x)$ if $\ol E(F(x))=E(F(x))$, and $T'(x)=1-\ol e'(F(x))F'(x)=1$ if $\ol E(F(x))<E(F(x))$, it follows that $T$ is increasing. Then, by $\ol E\in \Conv([0,1])$ and $F\in \Delta(Y)$, $\ol e\circ F$ is increasing, so $T$ is $1$-Lipschitz. 
Finally, $T$ satisfies \eqref{F} and $x\in X^C$ if and only if $E(F(x))>\ol E(F(x))$ (and hence $T'(x)=1)$, because 
\begin{align*}
\int_{\ul x}^{x} T(\tilde x) F( \df \tilde x) - \int_{\ul x}^{x} S(\tilde x) F(\df \tilde x) =E(F(x))-\ol E(F(x))\geq 0,
\end{align*}
for all $x\in X$, with equality at $x=\ol x$ by $E(1)=\ol E(1)$, and with strict inequality if and only if $E(F(x))>\ol E(F(x))$.

\emph{Part 3.} Take any $\tilde T\in \Tau$, and define $\tilde E:[0,1]\to \R $ by $\tilde E(t)=\int_0^t\bigl(F^{-1}(\tilde t)-\tilde T(F^{-1}(\tilde t))\bigr)\, \df \tilde t$ for all $t\in [0,1]$. Then $\tilde E$ is convex, because $\tilde T$ is $1$-Lipschitz. Moreover
\[
E(t)-\tilde E(t)=\int_{\ul x}^{F^{-1}(t)}\tilde T(x)\, F(\df x)-\int_{\ul x}^{F^{-1}(t)}S(x)\, F(\df x)\geq 0,
\]
for all $t\in [0,1]$, because $\tilde T$ satisfies \eqref{F}. Hence $\tilde E(t)\leq \ol E(t)$ for all $t\in [0,1]$, because $\ol E$ is the largest convex function that lies below $E$. Consequently, $T\circ F^{-1}\succeq  \tilde T\circ F^{-1}$, because
\[
\int_{\ul x}^{x} T(\tilde x) F( \df \tilde x) - \int_{\ul x}^{x} \tilde T(\tilde x) F(\df \tilde x)=\tilde E(F(x))-\ol E(F(x))\leq 0,
\]
for all $x\in X$.
\end{proof}

\begin{proof}[Proof of Theorem~\ref{c:psm}]	
\emph{Part 1.} Follows from Theorem~\ref{t:eqm}.

\emph{Part 2.} Follows from Lemma~\ref{l:uc} with $\hat x=\ul x$.

\emph{Part 3.} By strict quasiconcavity, there exists $\hat x\in X$ such that $v_z(x,S(x))$ is strictly increasing on $[\ul x,\hat x)$ and strictly decreasing on $(\hat x,\ol x]$. Assume that $\hat x<\ol x$, as otherwise part 1 applies. Then part 3 follows from the next lemma.
\begin{lemma}\label{l:uc}
There exists $x^*\in [\ul x,\hat x]$ and an increasing $T:X\to Y$ such that
\begin{gather}
\int_{x^*}^{\ol x} T(x)\, F(\df x)=\int_{x^*}^{\ol x} S(x)\, F(\df x), \label{e:x*1}\\
v_z(x^*,T(x^*))=v_z(x,T(x)),\quad \text{for all $x\in [x^*,\ol x]$},\label{e:x*2}\\
S(x^*)\leq T(x^*),\quad \text{with equality if $x^*\neq \ul x$}.\label{e:x*3}
\end{gather}
Moreover, the equilibrium matching is upper compression with $X^C=(x^*,\ol x)$. 
\end{lemma}
\begin{proof}
Define $O(x;z)=\int_{x}^{\ol x} (T(\tilde x)-S(\tilde x))\,F(\df \tilde x)$, where $T(\tilde x)$ is given by $v_z(x,z)=v_z(\tilde x, T(\tilde x)).$ First, consider the case $O(\ul x;S(\ul x))\leq 0$. In this case, $v_z(\ul x,S(\ul x))>v_z(\ol x,S(\ol x))$, as otherwise, for all $\tilde x\in (\ul x,\ol x)$, $v_z(\tilde x,T(\tilde x))<v_z(\tilde x,S(\tilde x))$ because $v_z(x,S(x))$ is strictly quasiconcave, so $T(\tilde x)>S(\tilde x)$ because $v_z$ is strictly decreasing in $z$, and thus $O(\ul x;S(\ul x))>0$, yielding a contradiction. Let $\hat z$ be given by $v_z(\ul x,\hat z)=v_z(\ol x,S(\ol x))$. Then, for all $\tilde x\in (\ul x,\ol x)$, $v_z(\tilde x,T(\tilde x))<v_z(\tilde x,S(\tilde x))$, so $T(\tilde x)>S(\tilde x)$, and thus $O(\ul x;\hat z)>0$. By the intermediate value theorem, there exists $z^*\in [S(\ul x),\hat z)$ such that $O(\ul x;z^*)=0$, so \eqref{e:x*1}--\eqref{e:x*3} hold. Since $v_z(x,S(x))$ is strictly quasiconcave and $v_z$ is strictly decreasing in $z$, there exists $x_0\in (\ul x,\ol x)$ such that $S(\tilde x)<T(\tilde x)$ for all $\tilde x\in (\ul x,x_0)$ and $S(\tilde x)>T(\tilde x)$ for all $\tilde x\in (x_0,\ol x]$, implying that \eqref{F} holds strictly for all $x\in (\ul x,\ol x)$. Thus, the equilibrium matching is full compression by Theorem~\ref{t:eqm}.

Next, consider the case $O(\ul x;S(\ul x))> 0$. Note that $O(\hat x;S(\hat x))<0$, as, for all $\tilde x\in (\hat x,\ol x]$, $v_z(\tilde x,T(\tilde x))=v_z(\hat x,S(\hat x))>v_z(\tilde x,S(\tilde x))$, so $T(\tilde x)<S(\tilde x)$. By the intermediate value theorem, there exists $x^*\in (\ul x,\hat x)$ such that $O(x^*;S(x^*))=0$, so \eqref{e:x*1}--\eqref{e:x*3} hold. Then there exists $x_0\in (x^*,\ol x)$ such that $S(\tilde x)<T(\tilde x)$ for all $\tilde x\in (x^*,x_0)$ and $S(\tilde x)>T(\tilde x)$ for all $\tilde x\in (x_0,\ol x]$. Thus the equilibrium matching is upper compression with $X^C=(x^*,\ol x)$ by Theorem~\ref{t:eqm}.
\end{proof}

\emph{Part 4.} Follows from part 3 by symmetry. 
\end{proof}
\subsection{Comparative Statics in the Mean-Measurable Case}
\begin{proof}[Proof of Remark~\ref{r:rms}]
Suppose by contradiction that $\hat v_z \unrhd v_z$ but there exist $x\in (\ul x,\ol x)$, $z\in (\ul y,\ol y)$, and $k>0$ such that 
\[
\frac{\hat v_{xz}(x,z)}{-\hat v_{zz}(x,z)}<k< \frac{v_{xz}(x,z)}{-v_{zz}(x,z)}.
\]
Then, by a Taylor expansion, there exists $\varepsilon>0$ such that $v_z(x,z)<v_z(x+\varepsilon,z+k\varepsilon)$ and $\hat v_z(x,z)>\hat v_z(x+\varepsilon,z+k\varepsilon)$, contradicting that $\hat v_z \unrhd v_z$.

Next, by rearrangement,
\[\frac{\hat v_{xz}(x,z)}{-\hat v_{zz}(x',z')}\geq \frac{v_{xz}(x,z)}{-v_{zz}(x',z')},\]
for all $(x,z)$ and $(x',z')$ in $X\times Y$ is equivalent to
\[\inf_{(x,z)\in X\times Y}\frac{\hat v_{xz}(x,z)}{v_{xz}(x,z)}\geq \sup_{(x',z')\in X\times Y}\frac{-\hat v_{zz}(x',z')}{-v_{zz}(x',z')},\]
which is equivalent to the existence of $c>0$ such that
\[
\frac{\hat v_{xz}(x,z)}{v_{xz}(x,z)}\geq c\geq  \frac{-\hat v_{zz}(x',z')}{-v_{zz}(x',z')},
\]
for all $(x,z)$ and $(x',z')$ in $X\times Y$, meaning that $\hat v_{xz}(x,z)\geq c v_{xz}(x,z)$ and $-\hat v_{zz}(x,z)\leq -c v_{zz}(x,z)$ for all $(x,z)$. 

Now, suppose that such $c>0$ exists. If $x' \geq x$, $z' \geq z$, and $v_z(x',z')\geq v_z(x,z)$, then
\begin{align*}
\hat v_z(x',z')-\hat v_z(x,z)&=\int_x^{x'}\hat v_{xz}(\tilde x,z)\,\df\tilde x+\int_z^{z'}\hat v_{zz}(x',\tilde z)\,\df\tilde z	\\
&\geq  c\int_x^{x'}v_{xz}(\tilde x,z)\,\df\tilde x+c\int_z^{z'}v_{zz}(x',\tilde z)\,\df\tilde z\\
&= c(v_z(x',z')-v_z(x,z))\geq 0,
\end{align*}
implying that $\hat v_z \unrhd v_z$.
\end{proof}

\begin{proof}[Proof of Theorem~\ref{t:CS}]
\emph{If parts of both parts.} Suppose that $\hat v_z \unrhd v_z$ and $\hat G\preceq G$. Let $T$ and $\hat T$ denote the equilibrium matchings under $(v,F,G)$ and $(\hat v,F,\hat G)$, and let $S=G^{-1}\circ F$ and $\hat S=\hat G^{-1}\circ F$. Suppose by contradiction that $\hat T\circ F^{-1}\nsucceq T\circ F^{-1}$. Then there exist $x_0<x_1$ in $X$ such that
$\int_{\ul x}^{x_i} \hat T(\tilde x)\,F(\df \tilde x) = \int_{\ul x}^{x_i} T(\tilde x)\,F(\df \tilde x)$ for $i=0,1$ and $\int_{\ul x}^{x} \hat T(\tilde x)\,F(\df \tilde x) > \int_{\ul x}^{x} T(\tilde x)\,F(\df \tilde x)$ for all  $x\in (x_0,x_1)$.
In turn, there exist $\tilde x_0<\tilde x_1$ in $(x_0,x_1)$ such that $\hat T(\tilde x_0)>T(\tilde x_0)$ and $\hat T(\tilde x_1)<T(\tilde x_1)$. Next, $\hat v_z(\tilde x_0,\hat T(\tilde x_0))=\hat v_z(\tilde x_1,\hat T(\tilde x_1))$ by Theorem~\ref{t:eqm}, because
\[
\int_{\ul x}^{x} \hat T(\tilde x)\,F(\df \tilde x) > \int_{\ul x}^{x} T(\tilde x)\,F(\df \tilde x) \geq \int_{\ul x}^{x} S(\tilde x)\,F(\df \tilde x) \geq \int_{\ul x}^{x} \hat S(\tilde x)\,F(\df \tilde x),
\]
for all $x\in (x_0,x_1)$. But then 
\[
\hat v_z(\tilde x_0, T(\tilde x_0))>\hat v_z(\tilde x_0, \hat T(\tilde x_0))=\hat v_z(\tilde x_1,\hat T(\tilde x_1))>\hat v_z(\tilde x_1, T(\tilde x_1)),
\]
which implies, by $\hat v_z \unrhd v_z$, that $v_z(\tilde x_0,T(\tilde x_0))>v_z(\tilde x_1,T(\tilde x_1))$, contradicting Theorem~\ref{t:eqm}, which shows that $v_z(x,T(x))$ is increasing in $x$.

\emph{Only if part of part 1.} Suppose that $\hat v_z \ntrianglerighteq v_z $, so there exist $x_1\geq x_0$ in $X$ and $y_1\geq y_0$ in $Y$ such that $v_z(x_1,y_1)\geq v_z(x_0,y_0) $ but $\hat v_z(x_1,y_1)<\hat v_z(x_0,y_0)$. Consider discrete distributions $F^d=\delta_{x_0}/2+\delta_{x_1}/2$ and $G^d=\delta_{y_0}/2+\delta_{y_1}/2$. Let  $J^d=\delta_{(x_0,y_0)}/2+\delta_{(x_1,y_1)}/2$, $r\in [v_z(x_0,y_0),v_z(x_1,y_1)]$, $p(x_i)=v(x_i,y_i)-ry_i$, and $q(z)=rz$ for all $z\in [y_0,y_1]$ and $i=0,1$. Then $J^d$ is the unique equilibrium assignment under $(v,F^d,G^d)$ by Theorem~\ref{t:el}, as
\begin{align*}
v(x_i,y_i)-r y_i+rz-v(x_i,z)&=r(z-y_i)-\int_{y_i}^z v_z(x_i,\tilde z)\,\df \tilde z\\
&\geq\int_{y_i}^z (v_z(x_i,y_i)-v_z(x_i,\tilde z))\,\df \tilde z\geq 0,
\end{align*}
for all $z\in [y_0,y_1]$ and $i=0,1$, with both inequalities holding as equality if and only if $z=y_i$. Next, there exist unique $\hat z_0<\hat z_1$ in $(y_0,y_1)$ such that $\hat z_0+\hat z_1=y_0+y_1$ and $\hat v_z(x_0,\hat z_0)=\hat v_z(x_1,\hat z_1)$. Let  $\hat J^d=\delta_{(x_0,\hat z_0)}/2+\delta_{(x_1,\hat z_1)}/2$, $\hat r= \hat v_z(x_0,\hat z_0)=\hat v_z(x_1,\hat z_1)$, $\hat p(x_i)=\hat v(x_i,\hat z_i)-\hat r\hat z_i$, and $\hat q(z)=\hat rz$ for all $z\in [y_0,y_1]$ and $i=0,1$. Then $\hat J^d$ is the unique equilibrium assignment under $(\hat v,F^d,G^d)$ by Theorem~\ref{t:el}, as
\begin{align*}
\hat v(x_i,\hat z_i)-\hat r \hat z_i+\hat rz-\hat v(x_i,z)=\int_{\hat z_i}^{z} (\hat v_z(x_i,\hat z_i)-\hat v_z(x_i,\tilde z))\,\df \tilde z\geq 0,
\end{align*}
for all $z\in [y_0,y_1]$ and $i=0,1$, with the inequality holding as equality if and only if $z=\hat z_i$. Clearly, $\hat J_Y^d\npreceq J_Y^d$. By continuity and approximation, there exist distributions $F$ and $G$ with continuous positive densities such that the equilibrium assignments under $(\hat v,F,G)$ and $(v,F,G)$ satisfy $\hat J_{Y}\npreceq J_{Y}$, and thus $\hat T\circ  F^{-1} \nsucceq T\circ  F^{-1}$.

\emph{Only if part of part 2.} Suppose that $\hat G\npreceq G$. Consider the convex-loss case with $F=G$, so that $T(x)=S(x)=x$ by Corollary~\ref{c:eqm}. Then $\hat T\circ  F^{-1} \nsucceq T\circ  F^{-1}$, as otherwise $\hat G^{-1}=\hat S\circ F^{-1}\succeq  \hat T\circ  F^{-1}\succeq  T\circ  F^{-1}=S\circ F^{-1}=G^{-1}$, contradicting that $\hat G\npreceq G$.
\end{proof}

\begin{proof}[Proof of Theorem~\ref{t:priceCS}]
\emph{If parts of parts 1 and 2.}
Suppose that $F\leq \hat F$ and $\hat G\leq  G$. Let $T$ and $\hat T$ denote the equilibrium matchings under $(v,F,G)$ and $(v,\hat F,\hat G)$, and let $S=G^{-1}\circ F$ and $\hat S=\hat G^{-1}\circ \hat F$. Suppose by contradiction that $\hat T\ngeq  T$. Then the set $A:=\{x\in X:\,\hat T(x)<T(x)\}$ is non-empty. By continuity of $\hat T$ and $T$, the set $A$ is open and is thus a union of at most countably many disjoint open intervals. Take one such interval $(\ul a,\ol a)$, with $\ul x\leq \ul a<\ol a\leq \ol x$. Note that $\int_{\ul x}^{\ul a}\hat T(x)\, \hat F(\df x)=\int_{\ul x}^{\ul a}\hat S(x)\, \hat F(\df x)$ and thus $\hat T(\ul a)\geq  \hat S(\ul a)=\hat G^{-1}(\hat F(\ul a))$ (with equality if $\ul a>\ul x$), as otherwise, $\int_{\ul x}^{\ul a}\hat T(x)\, \hat F(\df x)>\int_{\ul x}^{\ul a}\hat S(x)\, \hat F(\df x)$, so $\ul a>\ul x$ and, for some $x\in (\ul a,\ol a)$,
\[
v_z(\ul a,T(\ul a))=v_z(\ul a,\hat T(\ul a))=v_z(x,\hat T(x))>v_z(x,T(x))\geq v_z(\ul a,T(\ul a)).
\]
Similarly, $\int_{\ul x}^{\ol a}T(x)\, F(\df x)=\int_{\ul x}^{\ol a}S(x)\, F(\df x)$ and thus $T(\ol a)\leq S(\ol a)$ (with equality if $\ol a<\ol x$), as otherwise $\int_{\ul x}^{\ol a}T(x)\, F(\df x)>\int_{\ul x}^{\ol a}S(x)\, F(\df x)$, so $\ol a<\ol x$ and, for some $x\in (\ul a,\ol a)$,
\[
v_z(\ol a,\hat T(\ol a))\geq v_z(x,\hat T(x))>v_z(x,T(x))=v_z(\ol a,T(\ol a))= v_z(\ol a,\hat T(\ol a)).
\]

Next,
\begin{equation}\label{e:0<int}
\begin{aligned}
0\underset{(a)}{\leq} &\int_{\ul x}^{\ol a} (\hat T(x)-\hat S(x))\, \hat F(\df x)\underset{(b)}=\int_{\ul a}^{\ol a} (\hat T(x)-\hat S(x))\, \hat F(\df x)\\
\underset{(c)}=&\int_{\hat F(\ul a)}^{\hat F(\ol a)}(\hat T(\hat F^{-1}(t))-\hat G^{-1}(t))\, \df t \underset{(d)}< \int_{\hat F(\ul a)}^{\hat F(\ol a)}(T(\hat F^{-1}(t))-\hat G^{-1}(t))\, \df t,
\end{aligned}
\end{equation}
where (a) holds because $\hat T\circ \hat F^{-1}\preceq \hat G^{-1}$, (b) holds because $\int_{\ul x}^{\ul a}\hat T(x)\, \hat F(\df x)=\int_{\ul x}^{\ul a}\hat S(x)\, \hat F(\df x)$, (c) holds by a change of variable, and (d) holds because $\hat T(x)<T(x)$ for $x\in (\ul a,\ol a)$.

Note that
\begin{equation}\label{e:T_1<S_1}
T(\hat F^{-1}(t))\leq T(\ol a)\leq G^{-1}(F(\ol a))\leq G^{-1}(t) \leq \hat G^{-1}(t),\quad\text{for all $t\in [F(\ol a),\hat F(\ol a)]$},
\end{equation}
where the first inequality holds because $T$ is increasing and $t\leq \hat F(\ol a)$, the second inequality holds because $T(\ol a)\leq S(\ol a)$, the third inequality holds because $G^{-1}$ is increasing and $t\geq F(\ol a)$, and the fourth inequality holds because $G^{-1}\leq \hat G^{-1}$. Similarly, we have
\begin{equation}\label{e:T_0>S_0}
\hat T(F^{-1}(t))\geq \hat T(\ul a)\geq  \hat G^{-1}(\hat F(\ul a))\geq \hat G^{-1}(t)\geq G^{-1}(t),\quad\text{for all $t\in [F(\ul a),\hat F(\ul a)]$},
\end{equation}
 where the first inequality holds because $\hat T$ is increasing and $t\geq F(\ul a)$, the second inequality holds because $\hat T(\ul a)\geq \hat S(\ul a)$, the third inequality holds because $\hat G^{-1}$ is increasing and $t\leq \hat F(\ul a)$, and the fourth inequality holds because $\hat G^{-1}\geq G^{-1}$.

Finally, if $F(\ol a)\leq \hat F(\ul a)$, then $[\hat F(\ul a),\hat F(\ol a)]\subset [F(\ol a),\hat F(\ol a)]$, so, by \eqref{e:T_1<S_1}, the integrand in \eqref{e:0<int} is negative, yielding a contradiction. Otherwise (if $F(\ol a)> \hat F(\ul a)$), we again get a contradiction,
\begin{align*}
0\underset{(e)}<&\int_{\hat F(\ul a)}^{\hat F(\ol a)}(T(\hat F^{-1}(t))-\hat G^{-1}(t))\, \df t\underset{(f)}\leq \int_{\hat F(\ul a)}^{F(\ol a)}(T(\hat F^{-1}(t))-\hat G^{-1}(t))\, \df t\\
\underset{(g)}\leq & \int_{\hat F(\ul a)}^{F(\ol a)}(T(F^{-1}(t))-G^{-1}(t))\, \df t\underset{(h)}\leq \int_{F(\ul a)}^{F(\ol a)}(T(F^{-1}(t))-G^{-1}(t))\, \df t\\
\underset{(i)}= &\int_{\ul a}^{\ol a}(T(x)-S(x))\, F(\df x)\underset{(j)}=-\int_{\ul x}^{\ul a}(T(x)-S(x))\, F(\df x)\underset{(k)}\leq 0,
\end{align*}
where (e) holds by \eqref{e:0<int}, (f) holds by \eqref{e:T_1<S_1}, (g) holds because $T$ is increasing, $\hat F^{-1}\leq F^{-1}$, and $\hat G^{-1}\geq G^{-1}$, (h) holds because $T(F^{-1}(t))\geq G^{-1}(t)$ for all $t\in [F(\ul a),\hat F(\ul a)]$ by \eqref{e:T_0>S_0} and $T(F^{-1}(t))\geq \hat T(F^{-1}(t))$ for all $t\in [F(\ul a),F(\ol a)]\supset [F(\ul a),\hat F(\ul a)]$, (i) holds by a change of variable, (j) holds because $\int_{\ul x}^{\ol a}T(x)\, F(\df x)=\int_{\ul x}^{\ol a}S(x)\, F(\df x)$, and (k) holds because $T\circ F^{-1}\preceq G^{-1}$.

\emph{Only if part of part 1.} Suppose that $F\nleq \hat F$, so that there exists $x_0\in (\ul x,\ol x)$ such that $F(x_0)>\hat F(x_0)$. Consider the convex-loss case with $G$ uniform on $[0,\varepsilon]$ such that, for all $x\in X$, we have $\varepsilon\leq 1/F'(x)$ and $\varepsilon\leq 1/\hat F'(x)$, so $S'(x)=F'(x)\varepsilon\leq 1$ and $\hat S'(x)=\hat F'(x)\varepsilon\leq 1$. By Corollary~\ref{c:eqm}, the equilibrium matchings are $T=S$ and $\hat T=\hat S$, so $\hat T(x_0) = G^{-1}(\hat F(x_0)) < G^{-1}(F(x_0)) = T(x_0)$, showing that $\hat T\ngeq T$.

\emph{Only if part of part 2.} Suppose that $\hat G\nleq G$, so that there exists $y_0\in (\ul y,\ol y)$ such that $\hat G(y_0)>G(y_0)$. Consider the convex-loss case with $F$ uniform on $[0,1/\varepsilon]$ such that, for all $y\in Y$, we have $\varepsilon\leq G'(y)$ and $\varepsilon\leq \hat G'(y)$, so $S'(x)=\varepsilon /G'(S(x))\leq 1$ and $\hat S'(x)=\varepsilon / \hat G'(\hat S(x))\leq 1$. By Corollary~\ref{c:eqm}, the equilibrium matchings are $T=S$ and $\hat T=\hat S$, so, letting $x_0=F^{-1}(G(y_0))$, we have $\hat T(x_0) = \hat G^{-1}(F(x_0)) < G^{-1}(F(x_0)) = T(x_0)$, showing that $\hat T\ngeq T$.

\emph{Part 3.} Follows from part 3 of Theorem~\ref{t:eqm}.
\end{proof}

\subsection{Duality in the Moment-Measurable Case}\label{a:dm}

We first show that Theorem~\ref{t:e} applies to the moment-measurable case.

\begin{lemma}\label{l:LipM}
If $v\in \Lip(X\times \R)$ and $m\in \Lip(X\times Y)$, then $V(x,\eta)=v(x,\E_\eta[m(x,y)])\in \Lip(X\times \Delta(Y))$.
\end{lemma}
\begin{proof}
Letting $L_v$ and $L_m$ denote Lipschitz constants of $v$ and $m$, for all $x,\hat x\in X$ and $\eta,\hat \eta \in \Delta(Y)$, we have
\begin{align*}
V(x,\eta)-V(\hat x,\hat \eta)=&\left(v\left(x,\int_Y m(x,y)\eta(\df y)\right)-v\left(\hat x,\int_Y m(x,y)\eta(\df y)\right)\right )\\
+&\left(v\left(\hat x,\int_Y m(x,y)\eta(\df y)\right)-v\left(\hat x,\int_Y m(\hat x,y)\eta(\df y)\right)\right )\\
+&\left(v\left(\hat x,\int_Y m(\hat x,y)\eta(\df y)\right)-v\left(\hat x,\int_Y m(\hat x,y)\hat \eta(\df y)\right)\right )\\
\leq & L_v \rho_X(x,\hat x)+L_vL_m\rho_X(x,\hat x)+L_vL_m\kr{\eta -\hat \eta},
\end{align*}
so $V\in \Lip(X\times \Delta(Y))$
\end{proof}
Next, we simplify the primal and dual problems. Let $Z=[\ul z,\ol z]$, where $\ul z=\min_{x\in X}m(x,\ul y)$ and $\ol z=\max_{x\in X}m(x,\ol y)$. Moreover, for $\pi\in \Delta(X\times Y\times Z)$, let $\pi_{x,z}$ denote the conditional distribution of $y$ given $(x,z)$. The planner's problem is
\begin{equation}\tag{P$_M$} \label{PM}
\begin{gathered}
\max_{\pi \in  \Delta(X\times Y\times Z)} \E_{\pi } [v(x,z)]\\
\text{s.t.}\quad \pi_X=\mu,\quad \pi_Y=\nu,\quad\text{and}\quad \E_{\pi_{x,z}}[m(x,y)]=z,\quad\text{for all $(x,z)\in X\times Z$}.
\end{gathered}
\end{equation}
Let $B(X\times Z)$ denote the set of bounded, measurable functions on $X\times Z$. The dual problem is
\begin{equation}\tag{D$_M$} \label{DM}
\begin{gathered}
\min_{(p,q,r)\in \Lip(X)\times \Lip(Y)\times B(X\times Z)} \E_\mu[p (x)]+\E_\nu[q(y)]\\
\text{s.t.}\enspace  p(x)+q(y)\geq v(x,z)+r(x,z)(m(x,y)-z),\enspace \text{for all $(x,y,z)\in X\times Y\times Z$.}
\end{gathered}
\end{equation}
Let $\DfM$ denote the set of feasible solutions to \eqref{DM}.
\begin{lemma}\label{l:DDM}
Suppose that $v$ is continuously differentiable and $m$ is Lipschitz. Then $(p,q)\in \Df$ if and only if there exists $r\in B(X\times Z)$ such that $(p,q,r)\in \DfM$.
\end{lemma}
\begin{proof}
Suppose that $(p,q,r)\in \DfM$. Then, for all $\eta \in \Delta(Y)$ and $z=\E_\eta[m(x,y)]$, we have
\[
\int_Y (p(x)+q(y))\eta(\df y)\geq \int_Y (v(x,z)+r(x,z)(m(x,y)-z))\eta(\df y)=v(x,z)=V(x,\eta),
\]
so $(p,q)\in \Df$.

Suppose now that $(p,q)\in \Df$. Considering all $\eta\in \Delta(Y)$ with $|\supp(\eta)|=1$ gives
\begin{equation}\label{e:D1}
p(x)+q(y)\geq v(x,m(x,y)),\quad \text{for all $(x,y)\in X\times Y$}.
\end{equation}
Considering all $\eta\in \Delta(Y)$ with $|\supp(\eta)|=2$ gives
\begin{equation}\label{e:D2}
\begin{gathered}
\frac{p(x)+q(y_1)-v(x,z)}{m(x,y_1)-z}\leq \frac{p(x)+q(y_2)-v(x,z)}{m(x,y_2)-z},\\
\text{for all $y_1,y_2\in Y$ and $z\in Z$ such that $m(x,y_1)<z<m(x,y_2).$}
\end{gathered}
\end{equation}

Define $\tilde r$ on $X\times Y\times Z$ by 
\[
\tilde r(x,y,z)=
\begin{cases}
v_z(x,z), &z=m(x,y),\\
\frac{v(x,m(x,y))-v(x,z)}{m(x,y)-z}, &z\neq m(x,y).
\end{cases}
\]
Note that $\tilde r$ is continuous on $X\times Y\times Z$, because, by the mean value theorem, for each $(x,y,z)\in X\times Y\times Z$, there exists $\hat z$ between $z$ and $m(x,y)$ such that $\tilde r(x,y,z)=v_z(x,\hat z)$. Thus, there exists $C$ such that $|\tilde r(x,y,z)|\leq C$ for all $(x,y,z)\in X\times Y\times Z$.

Define $\ul r$ and $\ol r$ on $X\times Z$ by
\begin{align*}
\ul r(x,z)&=
\begin{cases}
\sup_{y_1\in Y:m(x,y_1)<z}\frac{p(x)+q(y_1)-v(x,z)}{m(x,y_1)-z},\quad &z>m(x,\ul y),\\
-\infty,\quad &z\leq m(x,\ul y),
\end{cases}\\
\ol r(x,z)&=
\begin{cases}
\inf_{y_2\in Y:m(x,y_2)>z}\frac{p(x)+q(y_2)-v(x,z)}{m(x,y_2)-z},\quad &z<m(x,\ol y),\\
+\infty,\quad &z\geq m(x,\ol y).
\end{cases}
\end{align*}
Note that the correspondence $R$ defined by $R(x,z)=[\ul r(x,z),\ol r(x,z)]$, for all $(x,z)\in X\times Z$, is nonempty valued, by \eqref{e:D2}. Moreover, for all $(x,z)\in X\times Z$, the set $R(x,z)\cap [-C,C]$ is nonempty, because $\ol r(x,z)\geq -C$ and $\ul r(x,z)\leq C$ for all $(x,z)\in X\times Z$, as, by \eqref{e:D1}, we have
\begin{align*}
\frac{p(x)+q(y_2)-v(x,z)}{m(x,y_2)-z}&\geq \frac{v(x,m(x,y_2))-v(x,z)}{m(x,y_2)-z}=\tilde r(x,y_2,z), &z<m(x,y_2),\\
\frac{p(x)+q(y_1)-v(x,z)}{m(x,y_1)-z}&\leq \frac{v(x,m(x,y_1))-v(x,z)}{m(x,y_1)-z}=\tilde r(x,y_1,z), &z>m(x,y_1).
\end{align*}
Then $r$ given by $r(x,z)=\argmin_{t\in R(x,z)} |t|$, for all $(x,z)\in X\times Z$, is a bounded, measurable function on $X\times Z$, by the measurable maximum theorem (\citealp{aliprantis2006}, Theorem 18.19). Finally, it is easy to verify that $(p,q,r)\in \DfM$. 
\end{proof}

By Theorem~\ref{t:e} and Lemmas~\ref{l:LipM} and \ref{l:DDM}, there exists a solution $(p,q,r)$ to \eqref{DM} such that $(p,q)$ is a solution to \eqref{D}. Define the contact set
\[
\Lambda=\left\{(x,\eta)\in X\times \Delta(Y):\, p(x)+\E_\eta[q(y)]=v(x,\E_\eta[m(x,y)]) \right \},
\]
which is compact by continuity in $(x,\eta)$. By Theorem~\ref{t:e}, $\gamma\in \Pf$ solves \eqref{P} if and only if $\supp(\gamma)\subset \Lambda$. Moreover, for every $(x,\eta)\in \Lambda$ and $z=\E_\eta[m(x,y)]$, we have
\[
\int_Y(p(x)+q(y)-v(x,z)-r(x,z)(m(x,y)-z))\eta(\df y)=p(x)+\E_\eta [q(y)]-v(x,z)=0.
\]
Since the integrand is nonnegative, by $(p,q,r)\in \DfM$, and continuous in $y$, we have, for all $y\in \supp(\eta)$,
\begin{equation}\label{e:DM}
\begin{gathered}
p(x)+q(y)=v(x,z)+r(x,z)(m(x,y)-z),\\
\text{for all $(x,\eta)\in \Lambda$ and $z=\E_\eta[m(x,y)]$.}	
\end{gathered}	
\end{equation}
\begin{lemma}\label{l:FOC}
Suppose that $v$ and $m$ are continuously differentiable.
Let $(p,q,r)$ be a solution to \eqref{DM}, let $(x,\eta)\in \Lambda$ with nondegenerate $\eta$, 
and let $z=\E_\eta[m(x,y)]$. Then $r(x,z)=v_z(x,z)$. Moreover, for Lebesgue almost every such $x$, 
$r$ has a partial derivative $r_x(x,z)$ with respect to $x$ at $(x,z)$ satisfying, for all $y\in \supp(\eta)$,
\begin{gather}\label{FOCx}
v_z(x,z)\E_\eta [m_x(x,y)]=v_z(x,z)m_x(x,y)+r_x(x,z)(m(x,y)-z).
\end{gather}
\end{lemma}
\begin{proof}
Since $\eta$ is nondegenerate and $m_y(x,y)>0$, for all $(x,y)\in X\times Y$, there exist $y_0,y_1\in \supp(\eta)$ such that $m(x,y_0)<z<m(x,y_1)$. By \eqref{e:DM} and $(p,q,r)\in \DfM$, we have, for all $\hat z\in (m(x,y_0),m(x,y_1))$ and $i=0,1$,  
\[
p(x)+q(y_i)=v(x,z)+r(x,z)(m(x,y_i)-z)\geq v(x,\hat z)+r(x,\hat z)(m(x,y_i)-\hat z),
\]
yielding
\begin{gather*}
(p(x)+q(y_0))\frac{m(x,y_1)-\hat z}{m(x,y_1)-m(x,y_0)}+(p(x)+q(y_1))\frac{\hat z-m(x,y_0)}{m(x,y_1)-m(x,y_0)}\\
=v(x,z)+r(x,z)(\hat z-z) \geq v(x,\hat z).
\end{gather*}
Thus, $r(x,z)=v_z(x,z)$, as follows from
\[
v_z(x,z) =\lim_{\hat z\downarrow z} \frac{v(x,\hat z)-v(x,z)}{\hat z-z} \frac{}{}\leq r(x,z)\leq \lim_{\hat z\uparrow z} \frac{v(x,\hat z)-v(x,z)}{\hat z-z}=v_z(x,z).
\]

Let $\hat X$ be the set of interior points $x$ of $X$ where $p$ has a derivative, which we denote by $p'(x)$. By Rademacher's theorem (e.g., \citealp{santambrogio2015}, Box 1.9), $X\setminus\hat X$ has zero Lebesgue measure. 
Fix $(x,\eta)\in \Lambda$ with $x\in \hat X$ and $|\supp(\eta)|>1$. Recall that there exist $y_0,y_1\in \supp (\eta)$ such that $m(x,y_0)<z<m(x,y_1)$. By \eqref{e:DM} and $(p,q,r)\in \DfM$, we have, for all $\hat x\in X$ with $m(\hat x,y_0)<z<m(\hat x,y_1)$, and $y\in \supp(\eta)$,  
\begin{equation}\label{e:Dx}
q(y)=v(x,z)+r(x,z)(m(x,y)-z)-p(x)\geq v(\hat x,z)+r(\hat x,z)(m(\hat x,y)-z) - p(\hat x).	
\end{equation}
Denote, for $i=0,1$,
\[
C_i=\frac{v_x(x,z)+r(x,z)m_x(x,y_i)-p'(x)}{z-m(x,y_i)}.
\]
Considering $y=y_0$ in \eqref{e:Dx} and recalling that $r(x,z)=v_z(x,z)$ yields
\[
r(\hat x,z)-r(x,z)\geq \frac{v(\hat x,z)-v(x,z)+v_z(x,z)(m(\hat x,y_0)-m(x,y_0))-(p(\hat x)-p(x))}{z-m(\hat x,y_0)},
\]
so
\begin{align*}
\ul r_x(x_+,z)&:=\liminf_{\hat x\downarrow x} \frac{r(\hat x,z)-r(x,z)}{\hat x-x}\geq C_0,\\
\ol r_x(x_-,z)&:=\limsup_{\hat x\uparrow x} \frac{r(\hat x,z)-r(x,z)}{\hat x-x}\leq C_0.	
\end{align*}
Similarly, considering $y=y_1$ in \eqref{e:Dx} yields
\begin{align*}
\ol r_x(x_+,z)&:=\limsup_{\hat x\downarrow x} \frac{r(\hat x,z)-r(x,z)}{\hat x-x}\leq C_1,\\
\ul r_x(x_-,z)&:=\liminf_{\hat x\uparrow x} \frac{r(\hat x,z)-r(x,z)}{\hat x-x}\geq C_1.	
\end{align*}
Thus,
\[
C_0\leq \ul r_x(x_+,z)\leq \ol r_x(x_+,z)\leq C_1\leq \ul r_x(x_-,z)\leq \ol r_x(x_-,z)\leq C_0,
\]
implying that all inequalities hold with equality, so $r$ has a derivative $r_x(x,z)$ at $(x,z)$, which then, by \eqref{e:Dx}, must satisfy, for all $y\in \supp(\eta)$,
\begin{equation}\label{e:p'}
v_x(x,z)+v_z(x,z)m_x(x,y)+r_x(x,z)(m(x,y)-z)=p'(x).
\end{equation}
Taking the expectation of \eqref{e:p'} with respect to $\eta$ gives $p'(x)=v_x(x,z)+v_z(x,z)\E_\eta[m_x(x,y)]$. Substituting $p'(x)$ back to \eqref{e:p'} gives \eqref{FOCx}.
\end{proof}

\subsection{Single-Dippedness in the Moment-Measurable Case}
We prove Theorems~\ref{t:sd} and \ref{t:stability} only in the single-dipped case; the single-peaked case is symmetric. Moreover, Theorem~\ref{t:sd} and its proof remain valid without the assumptions that $v_{xz}>0$ and that $G$ has full support or a density.
\begin{proof}[Proof of Theorem~\ref{t:sd}]
We first prove two key lemmas, which do not rely on the assumption that $v_{zz}<0$.
\begin{lemma}\label{l:sd}
For every $(x_0,\eta_0)$ and $(x_1,\eta_1)$ in $\Lambda$ such that there exist $y_0<y_1<y_0'$ with $y_0,y_0'\in \supp(\eta_0)$ and $y_1\in \supp(\eta_1)$, we have $x_0\geq x_1$.
\end{lemma}
\begin{proof}
Denote $z_i=\E_{\eta_i}[m(x_i,y)]$ for $i=0,1$. By \eqref{e:DM}, we have 
\begin{align}
p(x_0)+q(y_0)&=v(x_0,z_0)+r(x_0,z_0)(m(x_0,y_0)-z_0),\label{x0y0}\\
p(x_0)+q(y_0')&=v(x_0,z_0)+r(x_0,z_0)(m(x_0,y_0')-z_0),\label{x0y0'}\\
p(x_1)+q(y_1)&=v(x_1,z_1)+r(x_1,z_1)(m(x_1,y_1)-z_1).\label{x1y1}
\end{align}
By $(p,q,r)\in \DfM$, we have
\begin{align}
p(x_0)+q(y_1)&\geq v(x_0,z_0)+r(x_0,z_0)(m(x_0,y_1)-z_0),\label{x0y1}\\
p(x_1)+q(y_0)&\geq v(x_1,z_1)+r(x_1,z_1)(m(x_1,y_0)-z_1),\label{x1y0}\\
p(x_1)+q(y_0')&\geq v(x_1,z_1)+r(x_1,z_1)(m(x_1,y_0')-z_1).\label{x1y0'}
\end{align}
Substituting $q(y_0)$, $q(y_0')$, and $q(y_1)$ from \eqref{x0y0}, \eqref{x0y0'}, and \eqref{x1y1} into \eqref{x1y0}, \eqref{x1y0'}, and \eqref{x0y1} gives
\begin{align*}
p(x_0)-p(x_1)+v(x_1,z_1)+r(x_1,z_1)(m(x_1,y_1)-z_1)&\geq v(x_0,z_0)+r(x_0,z_0)(m(x_0,y_1)-z_0),\\
p(x_1)-p(x_0)+v(x_0,z_0)+r(x_0,z_0)(m(x_0,y_0)-z_0)&\geq v(x_1,z_1)+r(x_1,z_1)(m(x_1,y_0)-z_1),\\
p(x_1)-p(x_0)+v(x_0,z_0)+r(x_0,z_0)(m(x_0,y_0')-z_0)&\geq v(x_1,z_1)+r(x_1,z_1)(m(x_1,y_0')-z_1).
\end{align*}
Summing up these inequalities multiplied by $(m(x_1,y_0')-m(x_1,y_0))$, $(m(x_1,y_0')-m(x_1,y_1))$, $(m(x_1,y_1)-m(x_1,y_0))$, and then dividing by $r(x_0,z_0)=v_z(x_0,z_0)>0$ (see Lemma~\ref{l:FOC}) gives
\begin{gather*}
0\leq (m(x_0,y_0')-m(x_0,y_1))(m(x_1,y_1)-m(x_1,y_0))\\-(m(x_0,y_1)-m(x_0,y_0))(m(x_1,y_0')-m(x_1,y_1))\\
=\int_{y_1}^{y_0'} \int_{y_0}^{y_1} (m_y(x_0,y')m_y(x_1,y)-m_y(x_0,y)m_y(x_1,y'))\,\df y\,\df y'.
\end{gather*}
This inequality implies that $x_0\geq x_1$, as otherwise the integrand would be strictly negative by strict log-supermodularity of $m_y$, given that $x_0<x_1$ and $y<y'$.
\end{proof}

\begin{lemma}\label{l:2}
For $\mu$-almost every $x\in X$ and every $\eta\in \Delta(Y)$ such that $(x,\eta)\in \Lambda$, we have $|\supp(\eta)|\leq 2$.	
\end{lemma}
\begin{proof}
By Lemma~\ref{l:FOC}, for Lebesgue almost every $x\in X$ (and thus for $\mu$-almost every $x\in X$ since $\mu$ has a density), for every nondegenerate $\eta\in \Delta(Y)$ such that $(x,\eta)\in \Lambda$, for every $y\in \supp(\eta)$, and for $z=\E_\eta[m(x,y)]$, there exists $r_x(x,z)$ such that \eqref{FOCx} holds. Fix any such $x$ and suppose by contradiction that $\supp(\eta)$ contains $y_0<y_1<y_2$. Summing up \eqref{FOCx} for $i=0$ multiplied by $m(x,y_2)-m(x,y_1)$, for $i=1$ multiplied by $m(x,y_0)-m(x,y_2)$, and for $i=2$ multiplied by $m(x,y_1)-m(x,y_0)$ yields a contradiction:
\begin{gather*}
0=v_z(x,z)(m_x(x,y_2)-m_x(x,y_1))(m(x,y_1)-m(x,y_0))\\
-v_z(x,z)(m_x(x,y_1)-m_x(x,y_0))(m(x,y_2)-m(x,y_1))	\\
=v_z(x,z)\left[\int_{y_1}^{y_2}\int_{y_0}^{y_1}(m_{xy}(x,y')m_{y}(x,y)-m_{xy}(x,y)m_{y}(x,y'))\,\df y\,\df y'\right]\\
>v_z(x,z)\frac{m_{xy}(x,y_1)}{m_y(x,y_1)}\left[\int_{y_1}^{y_2}\int_{y_0}^{y_1}(m_{y}(x,y')m_{y}(x,y)-m_{y}(x,y)m_{y}(x,y'))\,\df y\,\df y'\right]=0,
\end{gather*}
where the inequality holds because $m_y$ is strictly log-supermodular, so
\[
\frac{m_{xy}(x,y')}{m_y(x,y')}>\frac{m_{xy}(x,y_1)}{m_y(x,y_1)}>\frac{m_{xy}(x,y)}{m_y(x,y)},\quad \text{for $y'>y_1>y$}.\qedhere
\]
\end{proof}
To complete the proof of Theorem~\ref{t:sd}, we now use the assumption that $v_{zz}<0$.  
Suppose for contradiction that there exist $(x,\eta_0),(x,\eta_1)\in \Lambda$ with $z_0=\E_{\eta_0}[m(x,y)]<\E_{\eta_1}[m(x,y)]=z_1$. Then, denoting $\eta=(\eta_0+\eta_1)/2$ and $z=(z_0+z_1)/2$, we have
\begin{multline*}
p(x)+\int_Y q(y)\, \eta(\df y)\geq v(x,z)>\tfrac{1}{2}v(x,z_0)+\tfrac{1}{2}v(x,z_1)\\
=\tfrac{1}{2}\Big (p(x)+\int_Y q(y)\, \eta_0(\df y)\Big )+\tfrac{1}{2}\Big (p(x)+\int_Y q(y)\, \eta_1(\df y)\Big )=p(x)+\int_Y q(y)\, \eta(\df y),
\end{multline*}
where the first inequality is by $(p,q)\in \Df$ and $V(x,\eta)=v(x,z)$, the second inequality is by $v_{zz}<0$, the first equality is by $(x,\eta_0),(x,\eta_1)\in \Lambda$ and the second equality is by the rearrangement. Hence $\E_{\eta_0}[m(x,y)]=\E_{\eta_1}[m(x,y)]$ for all $(x,\eta_0),(x,\eta_1)\in \Lambda$.

By Lemma~\ref{l:2}, there exists a full-measure set $X^*\subset X$ such that $|\supp(\eta)|\leq 2$ for all $(x,\eta)\in \Lambda$ with $x\in X^*$. Fix $(x,\eta_0),(x,\eta_1)\in \Lambda$ with $x\in X^*$. Denoting $\eta=(\eta_0+\eta_1)/2$ and recalling that $\E_{\eta_0}[m(x,y)]=\E_{\eta_1}[m(x,y)]=\E_{\eta}[m(x,y)]=z$, we have
\[
p(x)+\int_Y q(y)\, \eta(\df y)=\tfrac{1}{2}\Big(p(x)+\int_Y q(y)\, \eta_0(\df y)\Big)+\tfrac{1}{2}\Big(p(x)+\int_Y q(y)\, \eta_1(\df y)\Big)=v(x,z),
\]
showing that $(x,\eta)\in \Lambda$, so $|\supp(\eta)|=|\supp(\eta_0)\cup\supp(\eta_1)|\leq 2$, and thus $\eta_0=\eta_1$, by $\E_{\eta_0}[m(x,y)]=\E_{\eta_1}[m(x,y)]$ and $m_y>0$. So, there exist functions $\rho:X^*\to [0,1]$ and $T_d,T_u:X^*\to Y$, with $T_d(x)\leq T_u(x)$ for all $x\in X^*$, such that $\eta=(1-\rho(x))\delta_{T_d(x)}+\rho(x)\delta_{T_u(x)}$ for all $(x,\eta)\in \Lambda$ with $x\in X^*$. Moreover, $T_d(x'),T_u(x')\notin (T_d(x),T_u(x))$ for all $x<x'$ in $X^*$, by Lemma~\ref{l:sd}.
\end{proof}

\begin{proof}[Proof of Theorem~\ref{t:stability}]
By Theorem~\ref{t:sd}, each $\gamma_n$ is unique and is strictly single-dipped. Since $\Pf $ is compact, up to a subsequence, $\kr{\gamma_n-\gamma^*}\to 0$ for some $\gamma^*\in \Pf$.

We now show that $\gamma^*$ solves \eqref{P} under $m$. Fix arbitrary $\tilde\gamma\in \Pf$ and $\delta>0$. Define $V_n(x,\eta)=v(x,\E_\eta[m_n(x,y)])$ and $V(x,\eta)=v(x,\E_\eta[y])$. There exists $N\geq 1$ such that, for all $n\geq N$, we have $|V_n(x,\eta)-V(x,\eta)|\leq \delta$ for all $(x,\eta)\in X\times \Delta(Y)$, because $v$ is Lipschitz and $m_n$ converges uniformly to $m$, and $|\int V \,\df \gamma^*-\int V\,\df \gamma_n| \leq \delta$, because $V$ is Lipschitz, $X\times \Delta(Y)$ is compact, and $\kr{\gamma_n-\gamma^*}\to 0$. Thus,
\[
\int V\,\df\gamma^* \geq \int V\,\df\gamma_n-\delta \geq \int V_n\,\df\gamma_n-2\delta \geq \int V_n\,\df\tilde\gamma-2\delta\geq \int V\,\df\tilde\gamma-3\delta,
\]
showing that $\gamma^*$ solves \eqref{P}.

Next, we show that $\E_\eta[y]=T(x)$ for all $(x,\eta)\in \supp(\gamma^*)$. By Theorems~\ref{t:e} and \ref{t:el}, there exists $(p,q)\in \DfL$ such that, for all $(x,\eta)\in \supp(\gamma^*)$, we have $p(x)+q(\E_\eta[y])=v(x,\E_\eta[y])$, so $\E_\eta[y]=T(x)$ by Lemma~\ref{l:uniq}.

Next, we show that $\gamma^*$ is \emph{single-dipped}, meaning that there do not exist $(x_0,\eta_0),(x_1,\eta_1)\in \supp(\gamma^*)$, with $y_0,y_0'\in \supp(\eta_0)$ and $y_1\in \supp(\eta_1)$, such that $x_0<x_1$ and $y_0<y_1<y_0'$. Suppose by contradiction that such $(x_0,\eta_0)$ and $(x_1,\eta_1)$ exist. Since $\kr{\gamma_n-\gamma^*}\to 0$, lower hemicontinuity of the support correspondence (\citealp{aliprantis2006}, Theorem 17.14), with $\mu(X_n^*)=1$ (where $X_n^*$ is the full $\mu$-measure set from Lemma~\ref{l:2} for $\gamma_n$), implies that, for all large $n$, there exist $(x_{0,n},\eta_{0,n}),(x_{1,n},\eta_{1,n})\in \supp(\gamma_n)$, with $y_{0,n},y_{0,n}'\in \supp(\eta_{0,n})$ and $y_{1,n}\in \supp (\eta_{1,n})$, such that $x_{0,n},x_{1,n}\in X_n^*$, $x_{0,n}<x_{1,n}$, and $y_{0,n}<y_{1,n}<y_{0,n}'$, contradicting that $\gamma_n$ is strictly single-dipped.

Suppose by contradiction that $\gamma^*$ is not strictly single-dipped. Since $\gamma^*$ is single-dipped, $\E_\eta[y]=T(x)$, for all $(x,\eta)\in \supp(\gamma^*)$, and $\mu$ has a density, there must exist an uncountable set  $\tilde X\subset X$ such that, for each $x\in \tilde X$, we have $|\cup_{\eta:(x,\eta)\in \supp(\gamma^*)}\supp(\eta)|\geq 3$. By the pigeonhole principle (and $\E_\eta[y]=T(x)$ for all $(x,\eta)\in \supp(\gamma^*)$), there exist (rational) $y<y'$ in $Y$ and $(x_0,\eta_0), (x_1,\eta_1)\in \supp(\gamma^*)$ with $y_0,y_0'\in \supp(\eta_0)$ and $y_1\in \supp(\eta_1)$, such that $x_0<x_1$ and $y_0<y<y_1<y'<y_0'$, contradicting that $\gamma^*$ is single-dipped. 

By the subsequence criterion for convergence, to complete the proof of the theorem, it suffices to show that there is a unique single-dipped assignment $\gamma$ that induces $T$. Suppose by contradiction that there exist single-dipped (and thus strictly single-dipped) assignments $\gamma\neq \hat \gamma$ that induce $T$. Let $\psi$ and $\hat \psi$ be the distributions of $(T(x),\eta)$ induced by $\gamma$ and $\hat \gamma$. By part 1 of Theorem~\ref{t:eqm}, $T$ is continuous and strictly increasing, so $\psi\neq \hat \psi$  are both strictly single-dipped, contradicting Theorem 1.5 of \cite{beiglbock2016}.
\end{proof}

\subsection{Uniform Price Equilibrium}
\begin{proof}[Proof of Theorem~\ref{t:UPE}]
\emph{Part 4.} By definition, if $(\phi;p,q)$ is a UPE, then
\[
p(x)=W(x,S(x),\delta_{S(x)})-q(S(x))=\max_{y\in Y}\{W(x,y,\delta_{S(x)})-q(y)\},
\]
so, by the envelope theorem (\citealp{MS}, Theorem 3) and continuity of $W_x$, $W_z$, and $S'$,  $p$ has a derivative $p'(x)=W_x(x,S(x),\delta_{S( x)})+ W_z(x,S(x),\delta_{S(x)})S'(x)$ at each $x\in (\ul x,\ol x)$, yielding \eqref{e:pUPE}; and \eqref{e:qUPE} follows from $p(S^{-1}(y))+q(y)=W(S^{-1}(y),y,\delta_y)$.

\emph{Part 1.} PAS is a UPE if and only if, for all $(x,y)\in X\times Y$,
\begin{gather*}
q(y)-q(S(x))\geq W(x,y,\delta_{S(x)})-W(x,S(x),\delta_{S(x)}),
\end{gather*}
which is equivalent to \eqref{e:SUPE}, since $q$ is given by \eqref{e:qUPE}.

\emph{Part 2.} Suppose that \eqref{e:wxy} and \eqref{e:iwy} hold. Then, for all $(x,y)\in X\times Y$,
\begin{gather*}
\int_{S(x)}^{y}\int_{x}^{S^{-1}(\tilde y)} \frac{\df^2 W(\tilde x,\tilde y, \delta_{S(\tilde x)})}{\df \tilde x\,\df \tilde y}\, \df \tilde x\, \df \tilde y
=\int_{S(x)}^{y}\bigl(W_y(S^{-1}(\tilde y),\tilde y,\delta_{\tilde y})-W_y(x,\tilde y,\delta_{S(x)})\bigr)\,\df \tilde y\\
\ge \int_{S(x)}^{y}\bigl(W_y(x,\tilde y,\delta_{\tilde y})-W_y(x,\tilde y,\delta_{S(x)})\bigr)\,\df \tilde y
= \int_{S(x)}^{y} \int_{S(x)}^{\tilde y}
W_{yz}(x,\tilde y,\delta_{\tilde z})\,\df \tilde z\,\df \tilde y\ge\;0, 
\end{gather*}
so PAS is a UPE.

Conversely, suppose that PAS is a UPE for all continuous, positive densities $f$ and $g$. First, fix $x<x'$ in $(\ul x,\ol x)$, $y\in (\ul y,\ol y)$, and small $\varepsilon>0$. There exist continuous, positive densities such that $S(\tilde x)=y+\varepsilon(\tilde x-x)$ for $\tilde x\in [x,x'-\varepsilon-\varepsilon^2]$ and $S(\tilde x)=y+\varepsilon(x'-\varepsilon-x)+ (\tilde x-x'+\varepsilon)$ for $\tilde x\in [x'-\varepsilon,x']$. Since PAS is a UPE, \eqref{e:SUPE} holds. Taking $\varepsilon\downarrow 0$ yields $W_y(x',y,\delta_y)\geq W_y(x,y,\delta_y)$, so, by continuity of $W_{xy}$, \eqref{e:wxy} holds. Next, fix $x\in (\ul x,\ol x)$, $y> z$ in $(\ul y,\ol y)$ (the case $y< z$ is analogous and omitted), and small $\varepsilon>0$. There exist continuous, positive densities such that $S(\tilde x)=z+(y-z)(\tilde x-x)/\varepsilon $ for $\tilde x\in [x,x+\varepsilon]$. Since PAS is a UPE, \eqref{e:SUPE} holds. Taking $\varepsilon\downarrow 0$ yields $\int_z^y \int_z^{\tilde y}
W_{yz}(x,\tilde y,\delta_{\tilde z})\,\df \tilde z\,\df \tilde y\ge0$, so, by continuity of $W_{yz}$,  \eqref{e:iwy} holds.

\emph{Part 3.} Suppose that $W(x,y,\eta)=w(x,y,\E_\eta [m(x,\tilde y)])$, and \eqref{e:wxy>} and \eqref{e:iwy=} hold. Let $\gamma$ be a UPE. Fix $(x,\eta),(x',\eta')\in \supp(\gamma)$, $y\in \supp(\eta)$, and $y'\in \supp(\eta')$. By the intermediate value theorem, there exist $z,z'\in Y$ such that $\E_\eta[m(x,\tilde y)]=m(x,z)$ and $\E_{\eta'}[m(x',\tilde y)]=m(x',z')$. Then
\begin{gather*}
\int_y^{y'}\int_x^{x'} W_{xy}(\tilde x,\tilde y,\delta_{\tilde y})\, \df \tilde x\, \df \tilde y =W(x',y',\delta_{z'})-W(x',y,\delta_{z'})-W(x,y',\delta_{z}) +W(x,y,\delta_{z})\\
= p(x')+q(y)-W(x',y,\delta_{z'}) +p(x)+q(y')-W(x,y',\delta_{z}) \geq 0,
\end{gather*}
where the first equality holds because $W_{y}$ is independent of $z$, and the second equality and the inequality hold because $\gamma$ is a UPE. So, since $W_y(\tilde x,\tilde y,\delta_{\tilde y})$ is strictly increasing in $\tilde x$, we have $(x'-x)(y'-y)\geq 0$. That is, for all $(x,\eta),(x',\eta')\in \supp(\gamma)$ with $x<x'$, we have $\ol S(x)=\max \{\supp (\eta)\}\leq \min \{\supp(\eta')\}=\ul S(x')$, so, for almost all $x\in X$, $\ul S(x)=\ol S(x)=S(x)$, implying that $\gamma=\phi$.
\end{proof}

\begin{proof}[Proof of Corollary~\ref{c:affine}]
Since $w$ is affine in $y$, $w_y$ is independent of $y$. Thus, \eqref{e:SUPE} holds if and only if $w_y(x,y,S(x))=w_y(x,S(x),S(x))$ is increasing in $x$, \eqref{e:wxy} holds if and only if $w_y(x,y,z)=w_y(x,z,z)$ is increasing in $x$, and \eqref{e:iwy} holds if and only if $w_y(x,y,z)=w_y(x,z,z)$ is increasing in $z$.

Define the negative assortative matching (NAS) as the matching where each student $x$ goes to school $N(x)=G^{-1}(1-F(x))$. Note that NAS is a UPE if and only if 
\begin{equation*}\label{e:NAS}
\int_{N(x)}^{y}\int_{x}^{N^{-1}(\tilde y)} \frac{\df^2 w(\tilde x,\tilde y, N(\tilde x))}{\df \tilde x\,\df \tilde y}\, \df \tilde x\, \df \tilde y\geq 0,\quad\text{for all $(x,y)\in X\times Y$},
\end{equation*}
which holds if and only if $w_y(x,y,N(x))$ is decreasing in $x$, when $w$ is affine in $y$.

Suppose that PAS is the unique UPE for all densities $f$ and $g$. Then $w_y(x,z,z)$ is increasing in $x$ and $z$, and NAS is not a UPE for any $f$ and $g$. First, suppose for contradiction that there exist $x\in (\ul x,\ol x)$ and $z\in (\ul y,\ol y)$ such that $w_{yz}(x,z,z)>0$. Fix small $\varepsilon>0$, and let $f$ and $g$ be the uniform densities on $[x,x+\varepsilon^2]$ and $[z,z+\varepsilon]$. Then $w_y(\tilde x,z,N(\tilde x))$ is strictly decreasing in $x$ on $[x,x+\varepsilon^2]$, so NAS is a UPE. A contradiction. Thus, $w_{y}$ is independent of $z$. Next, suppose for contradiction that there exist $x<x'$ such that $w_y(x,y,z)=w_y(x',y,z)$. Let $f$ be the uniform density on $[x,x']$. Then $w_y(x,y,N(x))$ is constant for all $x\in [x,x']$, so NAS is a UPE. A contradiction. Thus, $w_{y}$ is strictly increasing in $x$.
The converse follows from the proof of Theorem~\ref{t:UPE}, which is valid if $w_y(x,y,y)$ is strictly increasing in $x$, rather than $w_{xy}(x,y,y)>0$ for all $(x,y)$.
\end{proof}

\end{document}